\documentclass{llncs}

\usepackage[utf8]{inputenc}
\usepackage{amssymb,amsmath,stmaryrd}
\usepackage{array}
\usepackage[table]{xcolor} 
\usepackage{blkarray} %
\usepackage{tikz} %
\usetikzlibrary{arrows,shapes,snakes,backgrounds,decorations.pathreplacing}

\usepackage{times}
\usepackage{mathptmx}
\usepackage{pdfpages}
\usepackage{fullpage}
\usepackage{float}
\usepackage{dsfont}
\DeclareSymbolFont{cmmathcal}{OMS}{cmsy}{m}{n}
\DeclareSymbolFontAlphabet{\mathcal}{cmmathcal}
\usepackage{colortbl}
\usepackage[colorlinks=true, citecolor=blue,linkcolor=blue,urlcolor=black]{hyperref}
\usepackage{algorithm}
\usepackage[noend]{algorithmic}
\usetikzlibrary{arrows,shapes,snakes,backgrounds}
\usepackage{paralist}
\usepackage{multirow}
\usepackage{enumerate}
\usepackage[lofdepth,lotdepth]{subfig}
\makeatletter
\renewcommand*{\@fnsymbol}[1]{\ensuremath{\ifcase#1\or \star\or \dagger\or \ddagger\or
       \mathsection\or \mathparagraph\or \|\or **\or \dagger\dagger
       \or \ddagger\ddagger \else\@ctrerr\fi}}
\makeatother

\let\oldmarginpar\marginpar
\renewcommand\marginpar[1]{\-\oldmarginpar[\raggedleft\scriptsize #1]%
{\raggedright\scriptsize #1}}
\newcommand{\abs}[1]{\ensuremath{\vert #1\vert} }

\newcommand{\states}{\ensuremath{S} }
\newcommand{\statesOne}{\ensuremath{S_{1}} }
\newcommand{\statesTwo}{\ensuremath{S_{2}} }
\newcommand{\initState}{\ensuremath{s_{{\sf init}}} }
\newcommand{\edges}{\ensuremath{E} }

\newcommand{\dimension}{\ensuremath{k} }
\newcommand{\weight}{\ensuremath{w} }
\newcommand{\largestW}{\ensuremath{W} }

\newcommand{\game}{\ensuremath{G} }
\newcommand{\gameFull}{\ensuremath{\game = \left( \statesOne, \statesTwo, \edges, \dimension, \weight\right)} }
\newcommand{\gameFullOneDim}{\ensuremath{\game = \left( \statesOne, \statesTwo, \edges,  \weight\right)} }
\newcommand{\integ}{\ensuremath{\mathbb{Z}} }
\newcommand{\nat}{\ensuremath{\mathbb{N}} }
\newcommand{\rat}{\ensuremath{\mathbb{Q}} }

\newcommand{\player}{\ensuremath{\mathcal{P}} }
\newcommand{\playerOne}{\ensuremath{\mathcal{P}_{1}} }
\newcommand{\playerTwo}{\ensuremath{\mathcal{P}_{2}} }
\newcommand{\play}{\ensuremath{\pi} }
\newcommand{\plays}{\ensuremath{{\sf Plays}(\game)} }
\newcommand{\prefixes}{\ensuremath{\mathsf{Prefs}(\game)} }
\newcommand{\prefixesArg}[1]{\ensuremath{\mathsf{Prefs}_{#1}(\game)} }
\newcommand{\prefix}{\ensuremath{\rho} }

\newcommand{\last}{\ensuremath{\mathsf{Last}} }

\newcommand{\mpay}{\ensuremath{\mathsf{MP}} }
\newcommand{\tpay}{\ensuremath{\mathsf{TP}} }
\newcommand{\tpaySup}{\ensuremath{\overline{{\sf TP}}} }
\newcommand{\tpayInf}{\ensuremath{\underline{{\sf TP}}} }
\newcommand{\mpaySup}{\ensuremath{\overline{{\sf MP}}} }
\newcommand{\mpayInf}{\ensuremath{\underline{{\sf MP}}} }
\newcommand{\tpayFiniteFull}{\ensuremath{\mathsf{TP}(\prefix) = \sum_{i = 0}^{i = n - 1} w(s_{i}, s_{i+1})} }
\newcommand{\mpayFiniteFull}{\ensuremath{\mathsf{MP}(\prefix) = \frac{1}{n} \tpay(\prefix)} }
\newcommand{\mpaySupFull}{\ensuremath{\mpaySup(\play) = \limsup_{n \rightarrow \infty} \mpay (\play(n))}}
\newcommand{\mpayInfFull}{\ensuremath{\mpayInf(\play) = \liminf_{n \rightarrow \infty} \mpay (\play(n))}}
\newcommand{\tpaySupFull}{\ensuremath{\tpaySup(\play) = \limsup_{n \rightarrow \infty} \tpay (\play(n))}}
\newcommand{\tpayInfFull}{\ensuremath{\tpayInf(\play) = \liminf_{n \rightarrow \infty} \tpay (\play(n))}}
\newcommand{\strat}{\ensuremath{\lambda} }
\newcommand{\strats}{\ensuremath{\Lambda} }
\newcommand{\stratsMemoryless}{\ensuremath{\Lambda^{M}} }
\newcommand{\stratsFinite}{\ensuremath{\Lambda^{F}} }
\newcommand{\outcome}[3]{\ensuremath{\mathsf{Outcome}_{\game}(#1,#2,#3)} }
\newcommand{\objective}{\ensuremath{\phi} }
\newcommand{\objMPInf}{\ensuremath{{\sf MeanInf}_{\game}(v)} }
\newcommand{\objTPInf}{\ensuremath{{\sf TotalInf}_{\game}(v)} }
\newcommand{\objTPInfT}[1]{\ensuremath{{\sf TotalInf}_{\game}(#1)} }
\newcommand{\objMPSup}{\ensuremath{{\sf MeanSup}_{\game}(v)} }
\newcommand{\objMPInfT}[1]{\ensuremath{{\sf MeanInf}_{\game}(#1)} }
\newcommand{\objMPSupT}[1]{\ensuremath{{\sf MeanSup}_{\game}(#1)} }

\newcommand{\objTPSup}{\ensuremath{{\sf TotalSup}_{\game}(v)} }
\newcommand{\objTPSupT}[1]{\ensuremath{{\sf TotalSup}_{\game}(#1)} }
\newcommand{\objMPInfFull}{\ensuremath{\objMPInf = \left\lbrace \play \in \plays \;\vert\; \mpayInf(\play) \geq v\right\rbrace} }
\newcommand{\objMPSupFull}{\ensuremath{\objMPSup = \left\lbrace \play \in \plays \;\vert\; \mpaySup(\play) \geq v\right\rbrace} }
\newcommand{\objTPInfFull}{\ensuremath{\objTPInf = \left\lbrace \play \in \plays \;\vert\; \tpayInf(\play) \geq v\right\rbrace} }
\newcommand{\objTPSupFull}{\ensuremath{\objTPSup = \left\lbrace \play \in \plays \;\vert\; \tpaySup(\play) \geq v\right\rbrace} }

\newcommand{\NPinter}{\ensuremath{\text{NP} \cap \text{coNP}}}
\newcommand{\windowMP}{\text{window mean-payoff}}

\newcommand{\goodWindowMPObj}{\ensuremath{\mathsf{GW}_{\game}(v, \sizeMax)} }
\newcommand{\goodWindowMPObjT}[1]{\ensuremath{\mathsf{GW}_{\game}(#1, \sizeMax)} }
\newcommand{\directFixedWindowMPObj}{\ensuremath{\mathsf{DirFixWMP}_{\game}(v, \sizeMax)} }
\newcommand{\fixedWindowMPObj}{\ensuremath{\mathsf{FixWMP}_{\game}(v, \sizeMax)} }
\newcommand{\fixedWindowMPObjT}[1]{\ensuremath{\mathsf{FixWMP}_{\game}(#1, \sizeMax)} }
\newcommand{\fixedWindowMPObjTW}[2]{\ensuremath{\mathsf{FixWMP}_{\game}(#1, #2)} }
\newcommand{\directFixedWindowMPObjT}[1]{\ensuremath{\mathsf{DirFixWMP}_{\game}(#1, \sizeMax)} }
\newcommand{\directFiniteWindowMPObj}{\ensuremath{\mathsf{DirBndWMP}_{\game}(v)} }
\newcommand{\finiteWindowMPObj}{\ensuremath{\mathsf{BndWMP}_{\game}(v)} }
\newcommand{\finiteWindowMPObjT}[1]{\ensuremath{\mathsf{BndWMP}_{\game}(#1)} }

\newcommand{\directFiniteWindowTPObjT}[1]{\ensuremath{\mathsf{DirBndWTP}_{\game}(#1)} }

\newcommand{\cycle}{\ensuremath{\mathcal{C}} }
\newcommand{\twoCM}{\ensuremath{\mathcal{M}} }

\newcommand{\winnings}{\ensuremath{\mathcal{W}} }
\newcommand{\winningsDir}{\ensuremath{\mathcal{W}} }
\newcommand{\winningsPos}{\ensuremath{\mathcal{W}_{g}} }
\newcommand{\wDirect}{\ensuremath{W_{d}} }
\newcommand{\wAttr}{\ensuremath{W_{attr}} }
\newcommand{\wWin}{\ensuremath{W} }
\newcommand{\wPos}{\ensuremath{W_{gw}} }
\newcommand{\wFP}{\ensuremath{W_{bp}} }
\newcommand{\posCost}[2]{\ensuremath{C_{#1}(#2)} }
\newcommand{\sizeMax}{\ensuremath{l_{\max}} }
\newcommand{\size}{\ensuremath{l} }

\newcommand{\windowMPAlg}{\ensuremath{\mathsf{FWMP}} }

\newcommand{\finiteProblemAlg}{\ensuremath{\mathsf{BoundedProblem}} }
\newcommand{\unbNegWindowAlg}{\ensuremath{\mathsf{UnbOpenWindow}} }
\newcommand{\negSupTPAlg}{\ensuremath{\mathsf{NegSupTP}} }
\newcommand{\directWinAlg}{\ensuremath{\mathsf{DirectFWMP}} }
\newcommand{\posSumAlg}{\ensuremath{\mathsf{GoodWin}} }
\newcommand{\natStrict}{\ensuremath{\mathbb{N}_{0}} }
\newcommand{\attr}{\ensuremath{\mathsf{Attr}} }
\newcommand{\state}{\ensuremath{s} }
\newcommand{\gameCB}{\ensuremath{\game^{c}} }
\newcommand{\playCB}{\ensuremath{\play^{c}} }
\newcommand{\statesCB}{\ensuremath{\states^{c}} }

\newcommand{\sumCB}{\ensuremath{\sigma} }
\newcommand{\stepsCB}{\ensuremath{\tau} }
\newcommand{\edgesCB}{\ensuremath{\edges^{c}} }

\newcommand{\gameCBFull}{\ensuremath{\game^{c} = (\states^{c}_{1}, \states^{c}_{2}, \edgesCB)} }
\newcommand{\we}{\ensuremath{\weight_{e}} }
\newcommand{\sink}{\ensuremath{\varsigma} }
\newcommand{\yes}{\ensuremath{\textsc{Yes}} }

\newcommand{\zeroVector}{\ensuremath{\{0\}^{\dimension}} }
\newcommand{\statesSize}{\ensuremath{\vert\states\vert} }
\newcommand{\statesSizeP}[1]{\ensuremath{\vert\states^{#1}\vert} }
\newcommand{\edgesSize}{\ensuremath{\vert\edges\vert} }
\newcommand{\bits}{\ensuremath{V} }
\newcommand{\directWinComp}{\ensuremath{\mathbb{C}_{{\sf DW}}} }
\newcommand{\posSumComp}{\ensuremath{\mathbb{C}_{{\sf GW}}} }
\newcommand{\coBuchi}{\ensuremath{\text{co-Büchi}} }
\newcommand{\gadgets}{\ensuremath{K} }
\newcommand{\cdGame}{\ensuremath{\mathcal{C}} }
\newcommand{\cdStates}{\ensuremath{\mathcal{S}} }
\newcommand{\cdTransitions}{\ensuremath{\mathcal{T}} }

\newcommand{\aptm}{\ensuremath{\mathcal{M}} }
\newcommand{\word}{\ensuremath{\zeta} }
\newcommand{\wordSize}{\ensuremath{\vert\zeta\vert} }
\newcommand{\alphabet}{\ensuremath{\Sigma} }
\newcommand{\tapeCell}{\ensuremath{h} }

\newcommand{\resetNet}{\ensuremath{\mathcal{N}} }
\newcommand{\resetNetPlaces}{\ensuremath{P} }
\newcommand{\resetNetTrans}{\ensuremath{T} }
\newcommand{\resetNetTransSize}{\ensuremath{\vert T\vert} }
\newcommand{\resetNetPlacesSize}{\ensuremath{\vert P\vert} }
\newcommand{\resetNetInitMarking}{\ensuremath{\overline{m_{0}}} }

\newcommand{\resetNetMarking}{\ensuremath{\overline{m}} }
\newcommand{\resetNetMarkingBis}{\ensuremath{\overline{n}} }
\newcommand{\resetNetLargestM}{\ensuremath{M} }
\newcommand{\resetNetOneVector}{\ensuremath{\overline{\mathbf{1}}} }

\newcommand{\resetNetVector}{\ensuremath{\overline{\mathbf{a}_{p\rightarrow b}}} }
\newcommand{\resetNetOneZeroVector}{\ensuremath{\overline{\mathbf{1}_{p\rightarrow 0}}} }
\newcommand{\resetNetZeroOneVector}{\ensuremath{\overline{\mathbf{0}_{q\rightarrow 1}}} }
\newcommand{\resetNetZeroMinusOneVector}{\ensuremath{\overline{\mathbf{0}_{q\rightarrow -1}}} }
\newcommand{\resetNetZeroVector}{\ensuremath{\mathbf{\overline{0}}} }
\newcommand{\resetNetInputT}[1]{\ensuremath{\mathbf{I}(#1)} }
\newcommand{\resetNetOutputT}[1]{\ensuremath{\mathbf{O}(#1)} }
\newcommand{\resetNetInput}{\ensuremath{\mathbf{I}} }
\newcommand{\resetNetOutput}{\ensuremath{\mathbf{O}} }
\newcommand{\resetNetReset}{\ensuremath{r} }

\newcommand{\bounded}{\text{bounded}}
\newcommand{\Bounded}{\text{Bounded}}
\newcommand{\playSuffix}[1]{\ensuremath{\play(#1,\infty)} }

\floatstyle{plain}
\newfloat{myalgo}{tbhp}{mya}
{\begin{myalgo}[#1]
\centering
\begin{minipage}{#2}
\begin{algorithm}[H]}%
{\end{algorithm}
\end{minipage}
\end{myalgo}}

\renewcommand{\arraystretch}{1.2}

\let\doendproof\endproof
\renewcommand\endproof{~\hfill\qed\doendproof}

\title{Looking at Mean-Payoff and Total-Payoff through Windows\thanks{Work partially supported by European project CASSTING (FP7-ICT-601148).}}

\author{Krishnendu Chatterjee\inst{1}$^{,}$\thanks{Author supported by Austrian Science Fund (FWF) Grant No P 23499-N23, FWF NFN Grant No S11407 (RiSE), ERC Start Grant (279307: Graph Games), Microsoft faculty fellowship.} \and Laurent Doyen\inst{2} \and Mickael Randour\inst{3}$^{,}$\thanks{Author supported by F.R.S.-FNRS fellowship.} \and \mbox{Jean-Fran\c{c}ois Raskin\inst{4}}$^{,}$\thanks{Author supported by ERC Starting Grant (279499: inVEST).}}

\institute{
IST Austria (Institute of Science and Technology Austria)\\ \and LSV - ENS Cachan, France\\ \and Computer Science Department, Université de Mons (UMONS), Belgium\\
\and D\'epartement d'Informatique, Universit\'e Libre de Bruxelles (U.L.B.), Belgium
}

\begin{document}

\maketitle

\begin{abstract}
We consider two-player games
played 
on weighted directed graphs with 
mean-payoff and total-payoff objectives,
two classical quantitative
objectives. While for single-dimensional
games the complexity and memory bounds for 
both objectives coincide, we show that in contrast to multi-dimensional
mean-payoff games that are known to be coNP-complete, multi-dimensional total-pay\-off
games are undecidable. We introduce conservative approximations of these objectives, where 
the payoff is considered over a local finite window sliding along a play,
instead of the whole play. For single dimension, we show that $(i)$~if the window size
is polynomial, 
deciding the winner 
takes polynomial time, and $(ii)$~the existence
of a $\bounded$ window can be decided in NP~$\cap$~coNP, and is at least as hard as 
solving mean-payoff games.
For multiple dimensions, we show that $(i)$~the problem with fixed window size is EXPTIME-complete,
and $(ii)$~there is no primitive-recursive algorithm to decide the existence of a $\bounded$ window.
\end{abstract}

\section{Introduction}
\smallskip\noindent{\bf Mean-payoff and total-payoff games.}
Two-player mean-payoff and total-payoff ga\-mes are played on finite weighted directed graphs 
(in which every edge has an integer weight) with two types of vertices:
in player-$1$ vertices,  player~$1$ chooses the successor vertex from the set
of outgoing edges; in player-$2$ vertices,  player~$2$ does likewise.
The game results in an infinite path through the graph, called a \emph{play}.
The mean-payoff (resp. total-payoff) value of a play is the long-run average 
(resp. sum) of the edge-weights along the path.
While traditionally games on graphs with $\omega$-regular objectives have been 
studied for system analysis, research efforts have recently focused on 
quantitative extensions to model resource constraints of embedded systems, 
such as power consumption, or buffer size~\cite{CdAHS03}.
Quantitative games, such as mean-payoff games, are crucial for the formal analysis of
resource-constrained reactive systems. 
For the analysis of systems with multiple resources, multi-dimension games, where edge weights are integer vectors, provide the appropriate framework.

\smallskip\noindent{\bf Decision problems.} The decision problem for mean-payoff
and total-payoff games asks, given a starting vertex, whether player~1 has a strategy
that against all strategies of the opponent ensures a play with value at least~0.
For both objectives, \emph{memoryless} winning strategies
exist for both players (where a memoryless strategy is independent of the 
past and depends only on the current state)~\cite{EM79,gimbert2004}. 
This ensures that the decision problems belong to NP~$\cap$~coNP; and they belong
to the intriguing class of problems that are in NP~$\cap$~coNP but whether
they are in P (deterministic polynomial time) are long-standing open questions.
The study of mean-payoff games has also been extended to multiple dimensions 
where the problem is shown to be coNP-complete~\cite{VR11,chatterjee_FSTTCS10}.
While for one dimension all the results for mean-payoff and total-payoff coincide,
our first contribution shows that quite unexpectedly (in contrast to multi-dimensional
mean-payoff games) the multi-dimensional total-payoff games are undecidable.

\smallskip\noindent{\bf Window objectives.} On the one hand, the complexity of
single-dimensional mean-payoff and total-payoff games is a long-standing open
problem, and on the other hand, the multi-dimensional problem is undecidable
for total-payoff games.
In this work, we propose to study variants of these objectives, namely,
\emph{$\bounded$ window mean-payoff} and \emph{fixed window mean-payoff} objectives.
In a $\bounded$ window mean-payoff objective instead of the long-run average along the 
whole play we consider payoffs over a local $\bounded$ window sliding along a play,
and the objective is that the average weight must be at least zero over every 
$\bounded$ window from some point on. 
This objective can be seen as a strengthening of the mean-payoff objective 
(resp. of the total-payoff objective if we require that the window objective is 
satisfied from the beginning of the play rather than from some point on), i.e., 
winning for the $\bounded$ window mean-payoff objective implies winning for the 
mean-payoff objective.
In the fixed window mean-payoff objective the window length is fixed and
given as a parameter.
Observe that winning for the fixed window objective implies winning for the 
$\bounded$ window objective.

\smallskip\noindent{\bf Attractive features for window objectives.} 
First, they are a strengthening of the mean-payoff objectives
and hence provide conservative approximations for mean-payoff objectives.
Second, the window variant is very natural to study in system analysis. 
Mean-payoff objectives require average to satisfy certain threshold in the long-run 
(or in the limit of the infinite path), whereas the window objectives require to 
provide guarantee on the average, not in the limit, but within a bounded time, 
and thus provide better time guarantee than the mean-payoff objectives. 
Third, the window parameter 
provides flexibility, as it can be adjusted specific to applications requirement of strong or weak 
time guarantee for system behaviors.
Finally, we will establish that our variant in the single dimension is more computationally 
tractable, which makes it an attractive alternative to mean-payoff objectives.

\renewcommand{\arraystretch}{1.2}
\begin{table}[tb]
  \centering   
\begin{footnotesize}
\begin{tabular}{|c||c|c|c||c|c|c|}
\cline{2-7} \multicolumn{1}{c|}{} &  \multicolumn{3}{c||}{~one-dimension~} & \multicolumn{3}{c|}{~$\dimension$-dimension~} \\ 
\cline{2-7} \multicolumn{1}{c|}{} &  ~~complexity~~ & ~~$\playerOne$ mem.~~ & ~~$\playerTwo$ mem.~~ & ~~complexity~~ & ~~$\playerOne$ mem.~~ & ~~$\playerTwo$ mem.~~\\ 
\hline ~~$\mpayInf$ / $\mpaySup$~~ & \NPinter & \multicolumn{2}{c||}{mem-less} & coNP-c. / $\NPinter$ & ~~infinite~~ & ~~mem-less~~\\
\hline ~~$\tpayInf$ / $\tpaySup$~~ & \NPinter & \multicolumn{2}{c||}{mem-less} & ~~\textbf{undec.} (Thm.~\ref{thm:undecidableTP})~~ & - & - \\
\hline ~~WMP: fixed~~ & \multirow{2}{*}{\textbf{P-c.} (Thm.~\ref{thm:oneDimFixed})} & \multicolumn{2}{c||}{\multirow{4}{*}{\begin{tabular}{c}\textbf{mem. req.}\\$\leq$ \textbf{linear($\vert\states\vert \cdot \sizeMax$)}\\(Thm.~\ref{thm:oneDimFixed})\end{tabular}}}  & \textbf{PSPACE-h.} (Thm.~\ref{thm:multiDimFixed}) & \multicolumn{2}{c|}{}\\
~polynomial window~ & & \multicolumn{2}{c||}{} & \textbf{EXP-easy} (Thm.~\ref{thm:multiDimFixed}) & \multicolumn{2}{c|}{\textbf{exponential}} \\
\cline{1-2}\cline{5-5} ~~WMP: fixed~~ & \multirow{2}{*}{\textbf{P($\vert\states\vert, \bits, \sizeMax$)}~~(Thm.~\ref{thm:oneDimFixed})} & \multicolumn{2}{c||}{} & \multirow{2}{*}{\textbf{EXP-c.} (Thm.~\ref{thm:multiDimFixed})} & \multicolumn{2}{c|}{(Thm.~\ref{thm:multiDimFixed})} \\
~~arbitrary window~~ &  & \multicolumn{2}{c||}{} & & \multicolumn{2}{c|}{}\\
\hline ~~WMP: $\bounded$~~ & \multirow{2}{*}{\textbf{\NPinter}~~(Thm.~\ref{thm:oneDimFinite})} & \textbf{mem-less} & \textbf{infinite} & \multirow{2}{*}{\textbf{NPR-h.} (Thm.~\ref{thm:multiDimFinite})} & \multirow{2}{*}{-} & \multirow{2}{*}{-}\\
~~window problem~~ & & (Thm.~\ref{thm:oneDimFinite}) & (Thm.~\ref{thm:oneDimFinite}) & & &\\
\hline
\end{tabular}
\end{footnotesize}
\vspace*{2mm}
\caption{Complexity of deciding the winner and memory required, with $\statesSize$ the number of states of the game (vertices in the graph), $\bits$ the length of the binary encoding of weights, and $\sizeMax$ the window size. New results in bold (h. for hard and c. for complete).}
\label{table:complexityAndMemory}
\end{table}

\smallskip\noindent{\bf Applicability.} 
In the context of $\omega$-regular objectives, the traditional infinitary notion of liveness has been strengthened to finitary liveness~\cite{AH98}, where instead of requiring that good events happen eventually, they are required to happen within a finite time bound. The notion of finitary parity games was introduced and studied in~\cite{CH06}, and a polynomial time algorithm for finitary parity games was given in~\cite{DBLP:journals/tocl/ChatterjeeHH09}, and also studied for pushdown games~\cite{CF_CSL13}. The notion of finitary conditions has also been extended to prompt setting where the good events are required to happen as promptly as possible~\cite{KPV09}. Our work extends the study of such finite time frames in the setting of quantitative objectives, and our window objectives can be viewed as an extension of finitary conditions for mean-payoff and total-payoff objectives.

With regard to applications, our window variants provide a natural framework to reason about quantitative properties under local finite horizons. To illustrate this point, consider a classical example of application with mean-payoff aspects, as presented by Bohy et al. in the context of synthesis from LTL specifications enriched with mean-payoff objectives~\cite{DBLP:conf/tacas/BohyBFR13}. Consider the synthesis of a suitable controller for a computer server having to grant requests to different types of clients. The LTL specification can express that all grants should eventually be granted. Adding quantities and a mean-payoff objective helps in defining priorities between requests and associating costs to the delays between requests and grants, depending of the relative priority of the request. Window objectives are useful for modeling such applications. Indeed, it is clear that in a desired controller, requests should not be placed on hold for an arbitrary long time. Similarly, if we have two types of requests, with different priorities, and we want to ensure guarantees on the mean waiting time per type of request, it seems natural that an adequate balance between the two types should be observable within reasonable time frames (which can be defined as part of the specification with our new objectives) instead of possible great variations that are allowed by the classical mean-payoff objective.

\vspace*{1mm}
\noindent{\bf Our contributions.} The main contributions of this 
work (along with the undecidability of multi-dimensional total-payoff games) are 
as follows:
\begin{enumerate}

\item \emph{Single dimension.}
For the single-dimensional case we present an algorithm for the fixed window problem that is 
polynomial in the size of the game graph times the length of the binary encoding of weights times the size of the
fixed window. Thus if the window size is polynomial, we have a polynomial-time algorithm.
For the $\bounded$ window problem we show that the decision problem is in 
NP~$\cap$~coNP, and at least as hard as solving mean-payoff games.
However, winning for mean-payoff games does not imply
winning for the $\bounded$ window mean-payoff objective, i.e., the winning sets
for mean-payoff games and $\bounded$ window mean-payoff games do not coincide.
Moreover, the structure of winning strategies is also very different, e.g., 
in mean-payoff games both players have memoryless winning strategies, but in 
$\bounded$ window mean-payoff games we show that player~2 requires infinite memory.
We also show that if player~1 wins the $\bounded$ window mean-payoff objective, then
a window of size $(\abs{S} - 1) \cdot (\abs{S} \cdot W + 1)$ is sufficient
where~$S$ is the state space (the set of vertices of the graph), and~$W$ is the largest absolute weight value.
Finally, we show that $(i)$ a winning strategy for the $\bounded$ window mean-payoff objective
ensures that the mean-payoff is at least~$0$ regardless of the strategy of the opponent, 
and $(ii)$ a strategy that ensures that the mean-payoff is strictly greater than~$0$
is winning for the $\bounded$ window mean-payoff objective.

\item \emph{Multiple dimensions}.
For multiple dimensions, we show that the fixed window problem is 
EXPTIME-complete (both for arbitrary dimensions with weights in 
$\{-1,0,1\}$ and for two dimensions with arbitrary weights); and if the window 
size is polynomial, then the problem is PSPACE-hard.
For the $\bounded$ window problem we show that the problem is non-primitive
recursive hard (i.e., there is no primitive recursive algorithm to decide
the problem).

\item \emph{Memory requirements.} 
For all the problems for which we prove decidability we also characterize the memory required by winning strategies.

\end{enumerate}

The relevant results are summarized in Table~\ref{table:complexityAndMemory}: our results are in
bold fonts.
In summary, the fixed window problem provides an attractive approximation of the 
mean-payoff and total-payoff games that we show have better algorithmic 
complexity.
In contrast to the long-standing open problem of mean-payoff games, the one-dimension fixed
window problem with polynomial window size can be solved in polynomial time;
and in contrast to the undecidability of multi-dimensional total-payoff games,
the multi-dimension fixed window problem is EXPTIME-complete.

\smallskip\noindent{\bf Related work.}
This paper extends the results presented in its preceding conference version~\cite{chatterjee_ATVA2013} and gives a full presentation of the technical details.
Mean-payoff games have been first studied by Ehrenfeucht and Mycielski in~\cite{EM79}
where it is shown that memoryless winning strategies exist for both players.
This entails that the decision problem lies in NP~$\cap$~coNP~\cite{KL93,ZP96}, 
and it was later shown to belong to
UP~$\cap$~coUP~\cite{jurdzinski98}. 
Despite many efforts~\cite{GKK88,ZP96,P99,LP07,BV07},
no polynomial-time algorithm for the mean-payoff games problem is known so far.
Gurvich, Karzanov, Khachivan and Lebedev~\cite{GKK88,KL93} provided the first 
(exponential) algorithm for mean-payoff games,
later extended by Pisaruk~\cite{P99}.  
The first pseudo-polynomial-time algorithm for mean-payoff games was given in~\cite{ZP96}
and was improved in~\cite{BCDGR11}.
Lifshits and Pavlov~\cite{LP07} propose an algorithm which is polynomial in the encoding of weights but exponential in the number of vertices of the graph: it is based on a graph decomposition procedure. 
Bjorklund and Vorobyov~\cite{BV07} present a \emph{randomized} 
algorithm which is both subexponential and pseudo-polynomial.
Special cases for mean-payoff games can also be solved in polynomial time
depending on the weight structure~\cite{CHKN14}, and the algorithmic problem
has also been studied for graphs (with one player only)~\cite{Karp78,CHKLR14}. Extension of the worst-case threshold problem - the classical decision problem on mean-payoff games - with guarantees on the expected performance faced to a stochastic adversary has been considered in~\cite{DBLP:conf/stacs/BruyereFRR14}. 
While all the above works are for single dimension, multi-dimensional
mean-payoff games have been studied in~\cite{VR11,chatterjee_FSTTCS10,DBLP:journals/acta/ChatterjeeRR14}.
One-dimension total-payoff games have been studied 
in~\cite{gawlitza2009} where it is shown that memoryless winning strategies exist 
for both players and the decision problem is in UP~$\cap$~coUP.

\section{Preliminaries}
\label{preliminaries}
We consider two-player turn-based games and denote the two \textit{players} by $\playerOne$ and $\playerTwo$.

\smallskip\noindent\textbf{Multi-weighted two-player game structures.} \textit{Multi-weighted two-player game structures} are weighted graphs \gameFull where (\textit{i}) \statesOne and \statesTwo resp. denote the finite sets of vertices, called \textit{states}, belonging to $\playerOne$ and $\playerTwo$, with $\statesOne \cap \statesTwo = \emptyset$ and $\states = \statesOne \cup \statesTwo$; (\textit{ii}) $\edges \subseteq \states \times \states$ is the set of \textit{edges} such that for all $s \in \states$, there exists $s' \in \states$ with $(s, s') \in \edges$; (\textit{iii}) $k \in \nat$ is the \textit{dimension} of the weight vectors; and (\textit{iv}) $\weight \colon \edges \rightarrow \integ^{\dimension}$ is the multi-weight labeling function. When it is clear from the context that a game $\game$ is one-dimensional ($\dimension = 1$), we omit $\dimension$ and write it as $\gameFullOneDim$. The game structure $\game$ is \textit{one-player} if $\statesTwo = \emptyset$. We denote by $\largestW$ the largest absolute weight that appears in the game.
For complexity issues, we assume that weights are encoded in binary. Hence we differentiate between pseudo-polynomial algorithms (polynomial in $\largestW$) and truly polynomial algorithms (polynomial in $\bits = \lceil\log_{2} \largestW\rceil$, the number of bits needed to encode the weights).

A \textit{play} in $\game$ from an initial state $\initState \in \states$ is an infinite sequence of states $\play = s_{0}s_{1}s_{2}\ldots{}$ such that $s_{0} = \initState$ and $(s_{i}, s_{i+1}) \in \edges$ for all $i \geq 0$. The \textit{prefix} up to the $n$-th state of $\play$ is the finite sequence $\play(n) = s_{0}s_{1}\ldots{}s_{n}$. Let $\last(\play(n)) = s_{n}$ denote the last state of $\play(n)$. A prefix $\play(n)$ belongs to $\player_{i}$, $i \in \lbrace 1, 2\rbrace$, if $\last(\play(n)) \in \states_{i}$. The set of plays of $\game$ is denoted by \plays and the corresponding set of prefixes is denoted by $\prefixes$. The set of prefixes that belong to $\player_{i}$ is denoted by $\prefixesArg{i}$. The infinite suffix of a play starting in $s_{n}$ is denoted $\playSuffix{n}$.

The \textit{total-payoff} of a prefix $\prefix = s_{0}s_{1}\ldots{}s_{n}$ is $\tpayFiniteFull$, and its \textit{mean-payoff} is $\mpayFiniteFull$. This is naturally extended to plays by considering the componentwise limit behavior (i.e., limit taken on each dimension). The \textit{infimum (resp. supremum) total-payoff} of a play $\play$ is $\tpayInfFull$ (resp. $\tpaySupFull$). The \textit{infimum (resp. supremum) mean-payoff} of $\play$ is $\mpayInfFull$ (resp.  $\mpaySupFull$).

\smallskip\noindent\textbf{Strategies.} A \textit{strategy} for $\player_{i}$, $i \in \lbrace 1, 2\rbrace$, in $\game$ is a function $\strat_{i} \colon \prefixesArg{i} \rightarrow \states$ such that $(\last(\prefix), \strat_{i}(\prefix)) \in \edges$ for all $\prefix \in \prefixesArg{i}$.
A strategy $\strat_{i}$ for $\player_{i}$ has \textit{finite-memory} if it can be encoded by a deterministic Moore machine $(M,m_0,\alpha_u,\alpha_n)$ where $M$ is a finite set of states (the memory of the strategy), $m_0 \in M$ is the initial memory state, $\alpha_u \colon M \times S \to M$ is an update function, and $\alpha_n \colon M \times S_{i} \to S$ is the next-action function. If the game is in $s \in S_{i}$ and $m \in M$ is the current memory value, then the strategy chooses $s' = \alpha_n(m,s)$ as the next state of the game. When the game leaves a state $s \in S$, the memory is updated to $\alpha_u(m,s)$. Formally, $\left\langle M, m_0, \alpha_u, \alpha_n\right\rangle $ defines the strategy $\strat_{i}$ such that $\strat_{i}(\rho\cdot s) = \alpha_n(\hat{\alpha}_u(m_0, \rho), s)$ for all $\rho \in S^*$ and $s \in S_{i}$, where $\hat{\alpha}_u$ extends $\alpha_u$ to sequences of states as expected. A strategy is \emph{memoryless} if $\vert M\vert = 1$, i.e., it does not depend on history but only on the current state of the game. We resp. denote by $\strats_{i}, \stratsFinite_{i}$, and $\stratsMemoryless_{i}$ the sets of general (i.e., possibly infinite-memory), finite-memory, and memoryless strategies for player $\player_{i}$.

A play \play is said to be \textit{consistent} with a strategy $\strat_{i}$ of $\player_{i}$ if for all $n \geq 0$ such that $\last(\play(n)) \in \states_{i}$, we have $\last(\play(n+1)) = \strat_{i}(\play(n))$. Given an initial state $\initState \in \states$, and two strategies, $\strat_{1}$ for $\playerOne$ and $\strat_{2}$ for $\playerTwo$, the unique play from $\initState$ consistent with both strategies is the \textit{outcome} of the game, denoted by $\outcome{\initState}{\strat_{1}}{\strat_{2}}$.

\smallskip\noindent\textbf{Attractors.} The \textit{attractor} for $\playerOne$ of a set $A \subseteq \states$ in $\game$ is denoted by $\attr_{\game}^{\playerOne}(A)$ and computed as the fixed point of the sequence $\attr_{\game}^{\playerOne,\,n+1}(A) = \attr_{\game}^{\playerOne,\,n}(A) \cup \{s \in \states_{1} \,\vert\, \exists\, (s,t) \in \edges,\, t \in \attr_{\game}^{\playerOne,\,n}(A)\} \cup \{s \in \states_{2} \,\vert\, \forall\, (s,t) \in \edges,\, t \in \attr_{\game}^{\playerOne,\,n}(A)\}$, with $\attr_{\game}^{\playerOne,\,0}(A) = A$. The attractor $\attr_{\game}^{\playerOne}(A)$ is exactly the set of states from which $\playerOne$ can ensure to reach $A$ no matter what $\playerTwo$ does. The attractor $\attr_{\game}^{\playerTwo}(A)$ for $\playerTwo$ is defined symmetrically.

\smallskip\noindent\textbf{Objectives.} An \textit{objective} for \playerOne in \game is a set of plays $\objective \subseteq \plays$. A play $\play \in \plays$ is \textit{winning} for an objective $\objective$ if $\play \in \objective$. Given a game $\game$ and an initial state $\initState \in \states$, a strategy $\strat_{1}$ of $\playerOne$ is winning if $\outcome{\initState}{\strat_{1}}{\strat_{2}} \in \objective$ for all strategies $\strat_{2}$ of $\playerTwo$.
Given a rational threshold vector $v \in \rat^{\dimension}$, we define the \textit{infimum (resp. supremum) total-payoff (resp. mean-payoff) objectives} as follows:

\vspace{5mm}
\noindent\begin{minipage}[b]{0.47\linewidth}
\begin{flushleft}
\begin{itemize}
\item $\objTPInfFull$
\item $\objTPSupFull$
\end{itemize}
\end{flushleft}
\end{minipage}
\hfill
\begin{minipage}[b]{0.47\linewidth}
\begin{center}
\begin{itemize}
\item $\objMPInfFull$
\item $\objMPSupFull$
\end{itemize}
\end{center}
\end{minipage}
\vspace{5mm}

\smallskip\noindent\textbf{Decision problem.} Given a game structure $\game$, an initial state $\initState \in \states$, and an inf./sup. total-payoff/mean-payoff objective $\objective \subseteq \plays$, the \textit{threshold problem} asks to decide if $\playerOne$ has a winning strategy for this objective. 
For the mean-payoff, the threshold $v$ can be taken equal to $\zeroVector$ (where $\zeroVector$ denotes the $\dimension$-dimension zero vector) w.l.o.g. as we transform the weight function $\weight$ to $b\cdot\weight - a$ for any threshold $\frac{a}{b}$, $a \in \integ^{\dimension}$, $b \in \natStrict = \nat \setminus \{0\}$. For the total-payoff, the same result can be achieved by adding an initial edge of value $-a$ to the game.

\section{Mean-Payoff and Total-Payoff Objectives}
\label{sec:MPTP}
In this section, we discuss classical mean-payoff and total-payoff objectives. We show that while they are closely related in one dimension, this relation breaks in multiple dimensions. Indeed, we establish that the threshold problem for total-payoff becomes undecidable, both for the infimum and supremum variants.

First, consider one-dimension games. In this case, memoryless strategies exist for both players for both objectives~\cite{liggett_SR69,EM79,filar1997,gimbert2004} and the sup. and inf. mean-payoff problems coincide (which is not the case for total-payoff). Threshold problems for mean-payoff and total-payoff are closely related as witnessed by Lemma~\ref{lem:linkMPTP} and both have been shown to be in $\NPinter$~\cite{ZP96,gawlitza2009}.

\begin{lemma}
\label{lem:linkMPTP}
Let $\gameFull$ be a two-player game structure and $\initState \in \states$ be an initial state. Let A, B, C and D resp. denote the following assertions.
\begin{enumerate}[A.]
\item Player $\playerOne$ has a winning strategy for $\objMPSupT{\zeroVector}$.
\item Player $\playerOne$ has a winning strategy for $\objMPInfT{\zeroVector}$.
\item There exists a threshold $v \in \rat^{\dimension}$ such that $\playerOne$ has a winning strategy for $\objTPInfT{v}$.
\item There exists a threshold $v' \in \rat^{\dimension}$ such that $\playerOne$ has a winning strategy for $\objTPSupT{v'}$.
\end{enumerate}
For games with one-dimension ($\dimension = 1$) weights, all four assertions are equivalent. For games with multi-dimension ($\dimension > 1$) weights, the only implications that hold are: $C \Rightarrow D \Rightarrow A$ and $C \Rightarrow B \Rightarrow A$. All other implications are false.
\end{lemma}

The statement of Lemma~\ref{lem:linkMPTP} is depicted in Fig.~\ref{fig:linkMPTP}: the only implications that extend to the multi-dimension case are depicted by solid arrows.

\begin{figure}[htb]
\centering
\scalebox{0.9}{\begin{tikzpicture}[dash pattern=on 10pt off 5,->,>=stealth',double,double distance=2pt,shorten >=1pt,auto,node
    distance=2.5cm,bend angle=45,scale=0.6,font=\normalsize]
    \tikzstyle{p1}=[]
    \tikzstyle{p2}=[draw,rectangle,text centered,minimum size=7mm]
    \node[p1]  (A)  at (-0.5, 0) {$A\colon\:\exists\,\strat^{A}_{1} \vDash \objMPSupT{\zeroVector}$};
    \node[p1]  (D) at (12.5, 0) {$D\colon\:\exists\, v \in \rat^{\dimension},\, \exists\,\strat^{D}_{1} \vDash \objTPSupT{v}$};
    \node[p1]  (B) at (-0.5, -4) {$B\colon\:\exists\,\strat^{B}_{1} \vDash \objMPInfT{\zeroVector}$};
    \node[p1]  (C) at (12.5, -4) {$C\colon\:\exists\, v' \in \rat^{\dimension},\, \exists\,\strat^{C}_{1} \vDash \objTPInfT{v'}$};
    \path
    ;
	\draw[dashed,dash phase =4pt,->,>=stealth,thin,double,double distance=1.5pt] (5.5,0) to (7,0);
	\draw[<-,>=stealth,thin,double,double distance=1.5pt,solid] (4,0) to (5.5,0);
	\draw[dashed,dash phase =4pt,->,>=stealth,thin,double,double distance=1.5pt] (5.5,-4) to (7,-4);
	\draw[<-,>=stealth,thin,double,double distance=1.5pt,solid] (4,-4) to (5.5,-4);
	\draw[<-,>=stealth,thin,double,double distance=1.5pt,solid] (0,-1) to (0,-2);
	\draw[dashed,dash phase =4pt,->,>=stealth,thin,double,double distance=1.5pt] (0,-2) to (0,-3);
	\draw[<-,>=stealth,thin,double,double distance=1.5pt,solid] (12,-1) to (12,-2);
	\draw[dashed,dash phase =4pt,->,>=stealth,thin,double,double distance=1.5pt] (12,-2) to (12,-3);
	\draw[<-,>=stealth,thin,double,double distance=1.5pt,solid] (3,-1) to (5.5,-2);
	\draw[dashed,dash phase =4pt,->,>=stealth,thin,double,double distance=1.5pt] (5.5,-2) to (8,-3);
	\draw[dashed,dash phase =4pt,<-,>=stealth,thin,double,double distance=1.5pt] (3,-3) to (5.5,-2);
	\draw[dashed,dash phase =4pt,->,>=stealth,thin,double,double distance=1.5pt] (5.5,-2) to (8,-1);
\end{tikzpicture}}
\caption{Equivalence between threshold problems for mean-payoff and total-payoff objectives. Dashed implications are only valid for one-dimension games.}
\label{fig:linkMPTP}
\vspace{-3mm}
\end{figure}
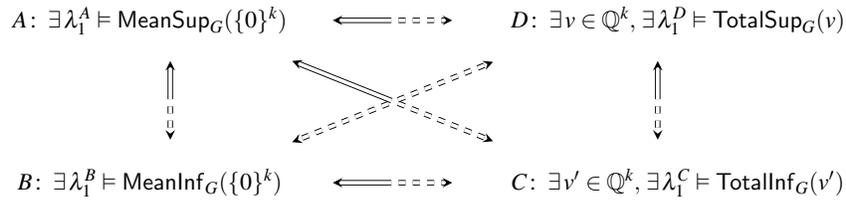

\begin{proof}
Specifically, the implications that remain true in multi-weighted games are the trivial ones: satifaction of the infimum version of a given objective trivially implies satisfaction of its supremum version, and satisfaction of infimum (resp. supremum) total-payoff for some finite threshold $v \in \rat^{k}$ implies satisfaction of infimum (resp. supremum) mean-payoff for threshold $\zeroVector$ as from some point on, the corresponding sequence of mean-payoff infima (resp. suprema) in all dimensions $t$, $1 \leq t \leq \dimension$, can be lower-bounded by a sequence of elements of the form $\frac{v(t)}{n}$ with $n$ the length of the prefix, which tends to zero for an infinite play. That is thanks to the sequence of total-payoffs over prefixes being a sequence of integers: it always achieves the value of its limit $v(t)$ instead of only tending to it asymptotically as could a sequence of rationals such as the mean-payoffs. This sums up to $C \Rightarrow D \Rightarrow A$ and $C \Rightarrow B \Rightarrow A$ being true even in the multi-dimension setting.

In the one-dimension case, all assertions are equivalent. First, we have that infimum and supremum mean-payoff problems coincide as memoryless strategies suffice for both players. Thus, we add $A \Rightarrow B$ and $D \Rightarrow B$ by transitivity. Second, consider an optimal strategy for $\playerOne$ for the mean-payoff objective of threshold $0$. This strategy is such that all cycles formed in the outcome have non-negative effect, otherwise $\playerOne$ cannot ensure winning. Thus, the total-payoff over any outcome that is consistent with the same optimal strategy is at all times bounded from below by $-2\cdot(\statesSize-1)\cdot\largestW$ (once for the initial cycle-free prefix, and once for the current cycle being formed). Therefore, we have that $B \Rightarrow C$, and we obtain all other implications by transitive closure.

\begin{figure}[htb]
\begin{minipage}[b]{0.47\linewidth}
\centering
\scalebox{0.85}{\begin{tikzpicture}[->,>=stealth',shorten >=1pt,auto,node
    distance=2.5cm,bend angle=45,scale=0.6,font=\normalsize]
    \tikzstyle{p1}=[draw,circle,text centered,minimum size=8mm]
    \tikzstyle{p2}=[draw,rectangle,text centered,minimum size=7mm]
    \node[p1]  (1)  at (0, 0) {$s$};
    \coordinate[shift={(0mm,5mm)}] (init) at (1.north);
    \path
    (1) edge [loop left, out=140, in=220,looseness=3, distance=2cm] node [left] {$(1,-2)$} (1)
    (1) edge [loop right, out=40, in=320,looseness=3, distance=2cm] node [right] {$(-2,1)$} (1)
    (init) edge (1);
\end{tikzpicture}}
\vspace{2mm}
      \caption{Satisfaction of supremum TP does not imply satisfaction of infimum MP.}
      \label{fig:linkMPTPex1}
\end{minipage}
\hfill
\begin{minipage}[b]{0.47\linewidth}
\centering

\scalebox{0.85}{\begin{tikzpicture}[->,>=stealth',shorten >=1pt,auto,node
    distance=2.5cm,bend angle=45,scale=0.6,font=\normalsize]
    \tikzstyle{p1}=[draw,circle,text centered,minimum size=8mm]
    \tikzstyle{p2}=[draw,rectangle,text centered,minimum size=7mm]
    \node[p1]  (1)  at (0, 0) {$s_{1}$};
    \node[p1]  (2) at (4, 0) {$s_{2}$};
    \coordinate[shift={(0mm,5mm)}] (init) at (1.north);
    \path
    (1) edge [loop left, out=140, in=220,looseness=3, distance=2cm] node [left] {$(-1,1,0)$} (1)
    (2) edge [loop right, out=40, in=320,looseness=3, distance=2cm] node [right] {$(1,-1,0)$} (2)
    (init) edge (1);
	\draw[->,>=latex] (1) to[out=45,in=135] node [above, yshift=0mm] {$(-1, -1,-1)$} (2);
	\draw[->,>=latex] (2) to[out=235,in=315] node [below, yshift=0mm] {$(-1,-1,-1)$} (1);
\end{tikzpicture}}
      \caption{Satisfaction of infimum MP does not imply satisfaction of supremum TP.}
\label{fig:linkMPTPex2}
\end{minipage}
\end{figure}

For multi-weighted games, all dashed implications are false. We specifically consider two of them.
\begin{enumerate}
\item To show that implication $D \Rightarrow B$ does not hold, consider the one-player game depicted in Fig.~\ref{fig:linkMPTPex1}. Clearly, any finite vector $v \in \rat^{2}$ for the supremum total-payoff objective can be achieved by an infinite memory strategy consisting in playing both loops successively for longer and longer periods, each time switching after getting back above the threshold in the considered dimension. However, it is impossible to build any strategy, even with infinite memory, that provides an infimum mean-payoff of $(0,0)$ as the limit mean-payoff would be at best a linear combination of the two cycles values, i.e., strictly less than $0$ in at least one dimension in any case.
\item Lastly, implication $B \Rightarrow D$ failure in multi-weighted games can be witnessed in Fig.~\ref{fig:linkMPTPex2}. Clearly, the strategy that plays for $n$ steps in the left cycle, then goes for $n$ steps in the right one, then repeats for $n' > n$ and so on, is a winning strategy for the infimum mean-payoff objective of threshold $(0,0,0)$. Nevertheless, for any strategy of $\playerOne$, the outcome is such that either (i) it only switches between cycles a finite number of time, in which case the sum in dimension 1 or 2 will decrease to infinity from some point on, or (ii) it switches infinitely and the sum of weights in dimension 3 decreases to infinity. In both cases, the supremum total-payoff objective is not satisfied for any finite vector $v \in \rat^{3}$.
\end{enumerate}

All other implications are deduced false as they would otherwise contradict the last two cases by transitivity.
\end{proof}

In multi-dimension games, recent results have shown that the threshold problem for inf. mean-payoff is coNP-complete whereas it is in $\NPinter$ for sup. mean-payoff~\cite{VR11,velner_corr2012}. In both cases, $\playerOne$ needs infinite memory to win, and memoryless strategies suffice for $\playerTwo$~\cite{chatterjee_FSTTCS10,velner_corr2012}. When restricted to finite-memory strategies, the problem is coNP-complete~\cite{chatterjee_FSTTCS10,velner_corr2012} and requires memory at most exponential for $\playerOne$~\cite{DBLP:journals/acta/ChatterjeeRR14}.

The case of total-payoff objectives in multi-weighted game structures has never been considered before. Surprisingly, the relation established in Lemma~\ref{lem:linkMPTP} cannot be fully transposed in this context. We show that the threshold problem indeed becomes undecidable for multi-weighted game structures, even for a fixed number of dimensions.

\begin{theorem}
\label{thm:undecidableTP}
The threshold problem for infimum and supremum total-payoff objectives is undecidable in multi-dimen\-sion games, for five dimensions.
\end{theorem}

\begin{proof}
We reduce the halting problem for two-counter machines (2CMs) to the threshold problem for two-player total-payoff games with five dimensions. From a two-counter machine $\twoCM$, we construct a two-player game $\game$ with five dimensions and an infimum (equivalently supremum) total-payoff objective such that $\playerOne$ wins for threshold $(0,0,0,0,0)$ if and only if the 2CM halts. Counters take values $(v_{1}, v_{2}) \in \nat^{2}$ along an execution, and can be incremented or decremented (if positive). A counter can be tested for equality to zero, and the machine can branch accordingly.
The halting problem for 2CMs is undecidable~\cite{minsky1961}. Assume w.l.o.g. that we have a 2CM $\twoCM$ such that if it halts, it halts with the two counters equal to zero. This is w.l.o.g.~as it suffices to plug a machine that decreases both counters to zero at the end of the execution of the considered machine. In the game we construct, $\playerOne$ has to faithfully simulate the 2CM~$\twoCM$. The role of $\playerTwo$ is to ensure that he does so by retaliating if it is not the case, hence making the outcome losing for the total-payoff objective.

The game is built as follows. The states of $\game$ are copies of the control states of $\twoCM$ (plus some special states discussed in the following). Edges represent transitions between these states. The payoff function maps edges to $5$-dimensional vectors of the form $(c_{1}, -c_{1}, c_{2}, -c_{2}, d)$, that is, two dimensions for the first counter $C_{1}$, two for the second counter $C_{2}$, and one additional dimension. Each increment of counter $C_{1}$ (resp. $C_{2}$) in $\twoCM$ is implemented in $\game$ as a transition of weight $(1, -1, 0, 0, -1)$ (resp. $(0, 0, 1, -1, -1)$. For decrements, we have weights respectively $(-1, 1, 0, 0, -1)$ and $(0, 0, -1, 1, -1)$ for $C_{1}$ and $C_{2}$. Therefore, the current value of counters $(v_{1}, v_{2})$ along an execution of the 2CM $\twoCM$ is represented in the game as the current sum of weights, $(v_{1}, -v_{1}, v_{2}, -v_{2}, -v_{3})$, with $v_{3}$ the number of steps of the computation. Hence, along a faithful execution, the 1st and 3rd dimensions are always non-negative, while the 2nd, 4th and 5th are always non-positive. The two dimensions per counter are used to enforce faithful simulation of non-negativeness of counters and zero test. The last dimension is decreased by one for every transition, except when the machine halts, from when it is incremented forever (i.e., the play in $\game$ goes to an absorbing state with self-loop $(0,0,0,0,1)$). This is used to ensure that a play in $\game$ is winning iff $\twoCM$ halts.

We now discuss how this game $\game$ ensures faithful simulation of the 2CM $\twoCM$ by $\playerOne$.
\begin{itemize}
\item \textit{Increment and decrement} of counter values are easily simulated using the first four dimensions.
\item \textit{Values of counters may never go below zero}. To ensure this, we allow $\playerTwo$ to branch after every step of the 2CM simulation to two special states, $s_{stop\_neg}^{1}$ and $s_{stop\_neg}^{2}$, which are absorbing and with self-loops of respective weights $(0, 1, 1, 1, 1)$ and $(1, 1, 0, 1, 1)$. If a negative value is reached on counter $C_{1}$ (resp. $C_{2}$), $\playerTwo$ can clearly win the game by branching to state $s_{stop\_neg}^{1}$ (resp. $s_{stop\_neg}^{2}$), as the total-payoff in the dimension corresponding to the negative counter will always stay strictly negative. On the contrary, if $\playerTwo$ decides to go to $s_{stop\_neg}^{1}$ (resp. $s_{stop\_neg}^{2}$) when the value of $C_{1}$ (resp. $C_{2}$) is positive, then $\playerOne$ wins the game as this dimension will be positive and the other four will grow boundlessly. So these transitions are only used if $\playerOne$ cheats.
\item \textit{Zero tests are correctly executed}. In the same spirit, we allow $\playerTwo$ to branch to two absorbing special states after a zero test, $s_{pos\_zero}^{1}$ and $s_{pos\_zero}^{2}$ with self-loops of weights $(1, 0, 1, 1, 1)$ and $(1, 1, 1, 0, 1)$. Such states are used by $\playerTwo$ if $\playerOne$ cheats on a zero test (i.e., pass the test with a strictly positive counter value). Indeed, if a zero test was passed with the value of counter $C_{1}$ (resp. $C_{2}$) strictly greater than zero, then the current sum $(v_{1}, -v_{1}, v_{2}, -v_{2}, v_{3})$ is such that $-v_{1}$ (resp. $-v_{2}$) is strictly negative. By going to $s_{pos\_zero}^{1}$ (resp. $s_{pos\_zero}^{2}$), $\playerTwo$ ensures that this sum will remain strictly negative in the considered dimension forever and the play is lost for~$\playerOne$.
\end{itemize}

Therefore, if $\playerOne$ does not faithfully simulate $\twoCM$, he is guaranteed to lose in $\game$. On the other hand, if $\playerTwo$ stops a faithful simulation, $\playerOne$ is guaranteed to win. It remains to argue that he wins iff the machine halts. Indeed, if the machine $\twoCM$ halts, then $\playerOne$ simulates its execution faithfully and either he is interrupted and wins, or the simulation ends in an absorbing state with a self-loop of weight $(0, 0, 0, 0, 1)$ and he also wins. Indeed, given that this state can only be reached with values of counters equal to zero (by hypothesis on the machine $\twoCM$, without loss of generality), the running sum of weights will reach values $(0, 0, 0, 0, n)$ where $n$ grows to infinity, which ensures satisfaction of the infimum (and thus supremum) total-payoff objective for threshold $(0,0,0,0,0)$. On the opposite, if the 2CM $\twoCM$ does not halt, $\playerOne$ has no way to reach the halting state by means of a faithful simulation and the running sum in the fifth dimension always stays negative, thus inducing a losing play for $\playerOne$, for both variants of the objective.

Consequently, we have that solving multi-weighted games for either the supremum or the infimum total-payoff objective is undecidable.
\end{proof}

We end this section by noting that in multi-weighted total-payoff games, $\playerOne$ may need infinite memory to win, even when all states belong to him ($\states_{2} = \emptyset$). Consider the game depicted in Fig.~\ref{fig:linkMPTPex1}. As discussed in the proof of Lemma~\ref{lem:linkMPTP}, given any threshold vector $v \in \rat^{2}$, $\playerOne$ has a strategy to win the supremum total-payoff objective: it suffices to alternate between the two loops for longer and longer periods, each time waiting to get back above the threshold in the considered dimension before switching. This strategy needs infinite memory and actually, there exists no finite-memory strategy that can achieve a finite threshold vector: the negative amount to compensate grows boundlessly with each alternation, and thus no amount of finite memory can ensure to go above the threshold infinitely often.

\section{Window Mean-Payoff Objective}
\label{sec:WMP}
In one dim\-en\-sion, no polynomial algorithm is known for mean-payoff and total-payoff, and in multiple dimensions, total-payoff is undecidable.
In this section, we introduce the \textit{window mean-payoff objective}, a conservative approximation in which local deviations from the threshold must be compensated in a parametrized number of steps. We consider a \textit{window}, sliding along a play, within which the compensation must happen. Our approach can be applied both to mean-payoff and total-payoff objectives. Since we consider \textit{finite} windows, both versions coincide for threshold zero. Hence we present our results for mean-payoff.

In Sec.~\ref{subsec:wmp_def}, we define the objective and discuss its relation with mean-payoff and total-payoff objectives. We then divide our analysis into two subsections: Sec.~\ref{subsec:wmp_oneDim} for one-dimension games and Sec.~\ref{subsec:wmp_multiDim} for multi-dimension games. Both provide thorough analysis of the \textit{fixed window problem} (the bound on the window size is a parameter) and the \textit{$\bounded$ window problem} (existence of a bound is the question). We establish solving algorithms, prove complexity lower bounds, and study the memory requirements of these objectives. In Sec.~\ref{sec:directObj}, we briefly discuss the extension of our results to a variant of our objective modeling stronger requirements.

\subsection{\textbf{Definition and comparison}}
\label{subsec:wmp_def}

\smallskip\noindent\textbf{Objectives and decision problems.} Given a multi-weighted two-player game $\gameFull$ and a rational threshold $v \in \rat^{\dimension}$, we define the following objectives.
\begin{itemize}
\item Given $\sizeMax \in \natStrict$, the \textit{good window} objective
\begin{align}
\goodWindowMPObj =  \Big\lbrace\play \in \plays \;\vert\; \forall\, t,\, 1 \leq t \leq \dimension,\, \exists\, \size \leq \sizeMax, \dfrac{1}{\size}\sum_{p = 0}^{\size - 1} \weight\Big( e_{\play}(p,p+1)\Big)(t) \geq v(t)\Big\rbrace\label{eq:goodWindowObj},
\end{align}
where $e_{\play}(p,p+1)$ is the edge $(\last(\play(p)), \last(\play(p+1)))$, requires that for all dimensions, there exists a window starting in the first position and bounded by $\sizeMax$ over which the mean-payoff is at least equal to the threshold.

\item Given $\sizeMax \in \natStrict$, the \textit{direct fixed window mean-payoff} objective
\begin{align}
\directFixedWindowMPObj =  \Big\lbrace\play \in \plays \;\vert\; \forall\, j \geq 0,\; \playSuffix{j} \in \goodWindowMPObj \Big\rbrace\label{eq:directFixedWindowObj}
\end{align}
requires that good windows bounded by $\sizeMax$ exist in all positions along the play.

\item The \textit{direct $\bounded$ window mean-payoff} objective
\begin{align}
\directFiniteWindowMPObj =  \Big\lbrace\play \in \plays \;\vert\; &\exists\, \sizeMax > 0,\; \play \in \directFixedWindowMPObj\Big\rbrace\label{eq:directFiniteWindowObj}
\end{align}
asks that there exists a bound $\sizeMax$ such that the play satisfies the direct fixed objective.

\item Given $\sizeMax \in \natStrict$, the \textit{fixed window mean-payoff}  objective
\begin{align}
\fixedWindowMPObj =  \Big\lbrace\play \in \plays \;\vert\; \exists\, i \geq 0,\; \playSuffix{i} \in \directFixedWindowMPObj \Big\rbrace\label{eq:fixedWindowObj}
\end{align}
is the \textit{prefix-independent} version of the direct fixed window objective: it asks for the existence of a suffix of the play satisfying it.

\item The \textit{$\bounded$ window mean-payoff} objective
\begin{align}
\finiteWindowMPObj =  \Big\lbrace\play \in \plays \;\vert\; &\exists\, \sizeMax > 0,\; \play \in \fixedWindowMPObj\Big\rbrace\label{eq:finiteWindowObj}
\end{align}
is the \textit{prefix-independent} version of the direct $\bounded$ window objective.
\end{itemize}
For any $v \in \rat^{\dimension}$ and $\sizeMax \in \natStrict$, the following inclusions are true: 
\begin{gather}
\directFixedWindowMPObj \subseteq \fixedWindowMPObj \subseteq \finiteWindowMPObj,\\
\directFixedWindowMPObj \subseteq \directFiniteWindowMPObj \subseteq \finiteWindowMPObj.
\end{gather}
Similarly to classical objectives, all objectives can be equivalently expressed for threshold $v = \zeroVector$ by modifying the weight function. Hence, given any variant of the objective, the associated \textit{decision problem} is to decide the existence of a winning strategy for $\playerOne$ for threshold $\zeroVector$.
Lastly, for complexity purposes, we make a difference between \textit{polynomial} (in the size of the game) and \textit{arbitrary} (i.e., non-polynomial) window sizes.

Notice that all those objectives define Borel sets. Hence they are determined by Martin's theorem~\cite{martin_AM75}.

Let $\play = s_{0}s_{1}s_{2}\ldots{}$ be a play. Fix any dimension $t, 1 \leq t \leq \dimension$. The window from position $j$ to $j'$, $0 \leq j < j'$, is \textit{closed} iff there exists $j''$, $j < j'' \leq j'$ such that the sum of weights in dimension $t$ over the sequence $s_{j}\ldots{}s_{j''}$ is non-negative. Otherwise the window is \textit{open}. Given a position $j'$ in $\play$, a window is still open in $j'$ iff there exists a position $0 \leq j < j'$ such that the window from $j$ to $j'$ is open. Consider any edge $(s_{i}, s_{i+1})$ appearing along $\play$. If the edge is non-negative in dimension $t$, the window starting in $i$ immediately closes. If not, a window opens that must be closed within $\sizeMax$ steps. Consider the \textit{first} position $i'$ such that this window closes, then we have that all intermediary opened windows also get closed by $i'$, that is, for any $i''$, $i < i'' \leq i'$, the window starting in $i''$ is closed before or when reaching position~$i'$. Indeed, the sum of weights over the window from $i''$ to $i'$ is strictly greater than the sum over the window from $i$ to $i'$, which is non-negative. We call this fact the \textit{inductive property of windows}.

\begin{figure}[htb]
\begin{minipage}[t]{0.52\linewidth}
\centering
  \scalebox{0.85}{\begin{tikzpicture}[->,>=stealth',shorten >=1pt,auto,node
    distance=2.5cm,bend angle=45,scale=0.55,font=\normalsize]
   \tikzstyle{p1}=[draw,circle,text centered,minimum size=8mm]
    \tikzstyle{p2}=[draw,rectangle,text centered,minimum size=7mm]
    \node[p1]  (0)  at (-4, 0) {$s_{1}$};
    \node[p1]  (1)  at (0, 0) {$s_{2}$};
    \node[p1]  (2)  at (4, 0) {$s_{3}$};
    \node[p1]  (3)  at (8, 0) {$s_{4}$};
    \coordinate[shift={(0mm,5mm)}] (init) at (0.north);
    \path
    (0) edge node[above] {$1$} (1)
    (1) edge node[above] {$-1$} (2)
    (init) edge (0);
	\draw[->,>=latex] (2) to[out=45,in=135] node [above, yshift=0mm] {$-1$} (3);
	\draw[->,>=latex] (3) to[out=235,in=315] node [below, yshift=0mm] {$1$} (2);
\end{tikzpicture}}
\caption{Fixed window is satisfied for $\sizeMax \geq 2$, whereas even direct bounded window is not.}
      \label{fig:wmpEx1}
\end{minipage}
\hfill
\begin{minipage}[t]{0.42\linewidth}
\centering
\scalebox{0.85}{\begin{tikzpicture}[->,>=stealth',shorten >=1pt,auto,node
    distance=2.5cm,bend angle=45,scale=0.55,font=\normalsize]
    \tikzstyle{p1}=[draw,circle,text centered,minimum size=8mm]
    \tikzstyle{p2}=[draw,rectangle,text centered,minimum size=7mm]
    \node[p1]  (1)  at (0, 0) {$s_{1}$};
    \node[p2]  (2) at (4, 0) {$s_{2}$};
    \coordinate[shift={(0mm,5mm)}] (init) at (1.north);
    \path
    (2) edge [loop right, out=40, in=320,looseness=3, distance=2cm] node [right] {$0$} (2)
    (init) edge (1);
	\draw[->,>=latex] (1) to[out=45,in=135] node [above, yshift=0mm] {$-1$} (2);
	\draw[->,>=latex] (2) to[out=235,in=315] node [below, yshift=0mm] {$1$} (1);
\end{tikzpicture}}
      \caption{Mean-payoff is satisfied but none of the window objectives is.}
\label{fig:wmpEx2}
\end{minipage}
\end{figure} 

\smallskip\noindent\textbf{Illustration.}
Consider the game depicted in Fig.~\ref{fig:wmpEx1}. It has a unique outcome, and it is winning for the classical mean-payoff objective of threshold $0$, as well as for the infimum (resp. supremum) total-payoff objective of threshold $-1$ (resp. $0$). Consider the fixed {\windowMP} objective for threshold $0$.
If the size of the window is bounded by~$1$, the play is losing.\footnote{A window size of one actually requires that all infinitely often visited edges are of non-negative weights.} However, if the window size is at least $2$, the play is winning, as in $s_{3}$ we close the window in two steps and in $s_{4}$ in one step. Notice that by definition of the objective, it is clear that it is also satisfied for all larger sizes.\footnote{The existential quantification on the window size $\size$, bounded by $\sizeMax$, is indeed crucial in Eq.~\eqref{eq:goodWindowObj} to ensure monotonicity with increasing maximal window sizes, a desired behavior of the definition for theoretical properties and intuitive use in specifications.} As the fixed window objective is satisfied for size $2$, the $\bounded$ window objective is also satisfied. On the other hand, if we restrict the objectives to their direct variants, then none is satisfied, as from $\state_{2}$, no window, no matter how large it is, gets closed.

Consider the game of Fig.~\ref{fig:wmpEx2}. Again, the unique strategy of $\playerOne$ satisfies the mean-payoff objective for threshold $0$. It also ensures value $-1$ for the infimum and supremum total-payoffs. Consider the strategy of $\playerTwo$ that takes the self-loop once on the first visit of $s_{2}$, twice on the second, and so on. Clearly, it ensures that windows starting in $s_{1}$ stay open for longer and longer numbers of steps (we say that $\playerTwo$ \textit{delays} the closing of the window), hence making the outcome losing for the $\bounded$ window objective (and thus the fixed window objective for any $\sizeMax \in \natStrict$). This illustrates the added guarantee (compared to mean-payoff) asked by the window objective: in this case, no upper bound can be given on the time needed for a window to close, i.e., on the time needed to get the local sum back to non-negative. Note that $\playerTwo$ has to go back to $s_{1}$ at some point: otherwise, the prefix-independence of the objectives\footnote{Fixed and bounded window mean-payoff objectives are prefix-independent: for all $\prefix \in \prefixes$, $\play \in \plays$, we have that $\prefix \cdot \play$ is winning if and only if $\play$ is winning.} allows $\playerOne$ to wait for $\playerTwo$ to settle on cycling and win. For the direct variants, $\playerTwo$ has a simpler winning strategy consisting in looping forever, as enforcing one permanently open window is sufficient.

\smallskip\noindent\textbf{Relation with classical objectives.}
We introduce the $\bounded$ window objectives as conservative approximations of mean-payoff and total-payoff in one-dimension games. Indeed, in Lemma~\ref{lem:relationBoundedClassical}, we show that winning the $\bounded$ window (resp. direct $\bounded$ window) objective implies winning the mean-payoff (resp. total-payoff) objective while the converse implication is only true if a strictly positive mean-payoff (resp. arbitrary high total-payoff) can be ensured. 

\begin{lemma}
\label{lem:relationBoundedClassical}
Given a one-dimension game $\gameFullOneDim$, the following assertions hold.
\begin{enumerate}[(a)] 
\item If the answer to the $\bounded$ window mean-payoff problem is $\yes$, then the answer to the mean-payoff threshold problem for threshold zero is also $\yes$.
\item If there exists $\varepsilon > 0$ such that the answer to the mean-payoff threshold problem for threshold $\varepsilon$ is $\yes$, then the answer to the $\bounded$ window mean-payoff problem is also $\yes$.
\item If the answer to the direct $\bounded$ window mean-payoff problem is $\yes$, then the answer to the supremum total-payoff threshold problem for threshold zero is also $\yes$.
\item If the answer to the supremum total-payoff threshold problem is $\yes$ for all integer thresholds (i.e., the total-payoff value is $\infty$), then the answer to the direct $\bounded$ window mean-payoff problem is also $\yes$.
\end{enumerate}
\end{lemma}

Assertions \textit{(a)} and \textit{(c)} follow from the decomposition of winning plays into bounded windows of non-negative weights. The key idea for assertions \textit{(b)} and \textit{(d)} is that mean-payoff and total-payoff objectives always admit \textit{memoryless} winning strategies, for which the consistent outcomes can be decomposed into \textit{simple cycles} (i.e., with no repeated edge) over which the mean-payoff is at least equal to the threshold and which length is bounded. Hence they correspond to closing windows. Note that strict equivalence with the classical objectives is not verified, as witnessed before (Fig.~\ref{fig:wmpEx2}). 

\begin{proof}
\textit{Assertion (a)}. In the one-dimension case, sup. and inf. mean-payoff problems coincide. Let $\play \in \plays$ be  such that $\play \in \finiteWindowMPObjT{0}$. There exists $i \geq 0$ such that the suffix of $\play$ starting in $i$ can be decomposed into an infinite sequence of bounded segments (i.e., windows) of non-negative weight. Thus, this suffix satisfies the sup. mean-payoff objective as there are infinitely many positions where the total sum from $i$ is non-negative. Since the mean-payoff objective is prefix-independent, the play $\play$ is itself winning.

\textit{Assertion (b)}. Consider a memoryless winning strategy of $\playerOne$ for the mean-payoff of threshold $\varepsilon > 0$. Only strictly positive simple cycles can be induced by such a strategy. Consider any outcome $\play = \state_{0}\state_{1}\state_{2}\ldots{}$ consistent with it. We claim that for any position $j$ along this play, there exists a position $j+\size$, with $\size \leq \sizeMax = (\statesSize - 1)\cdot (1 + \statesSize \cdot \largestW)$, such that the sum of weights over the sequence $\prefix = \state_{j} \ldots \state_{j+l}$ is non-negative. Clearly, if it is the case, then objective $\fixedWindowMPObj$ is satisfied and so is objective $\finiteWindowMPObj$. Consider the cycle decomposition $\mathcal{A}\cycle_{1}\cycle_{2}\ldots{}\cycle_{n}\mathcal{B}$ of this sequence obtained as follows. We push successively $\state_{0},\state_{1},\ldots{}$ onto a stack, and whenever we push a state that is already in the stack, a simple cycle is formed that we remove from the stack and append to the cycle decomposition. The sequence $\prefix$ is decomposed into an acyclic part ($\mathcal{A} \cup \mathcal{B}$), whose length\footnote{The length of a sequence is the number of \textit{edges} it involves.}
is at most $(\statesSize - 1)$ and whose total sum is at least $-(\statesSize - 1)\cdot \largestW$, and simple cycles of total sum at least $1$ and length at most $\statesSize$. Given the window size $\sizeMax$, we have at least $(\statesSize - 1)\cdot \largestW$ simple cycles in the cycle decomposition. Hence, the total sum over $\prefix$ is at least zero, which proves our point.

\textit{Assertion (c)}. Consider a play $\play \in \directFiniteWindowTPObjT{0}$. Using the same decomposition argument as for assertion \textit{(a)}, we have that the sequence of total sums takes infinitely often values at least equal to zero. Thus the limit of this sequence of moments bounds from below the limit of the sequence of suprema and is at least equal to zero, which shows that the supremum total-payoff objective is also satisfied by play $\play$.

\textit{Assertion (d)}. In one-dimension games, the value of the total-payoff (i.e., the largest threshold for which $\playerOne$ has a winning strategy) is $\infty$ if and only if the value of mean-payoff is strictly positive~\cite{gawlitza2009}. Hence, we apply the argument of assertion \textit{(b)}, further noticing that  the window open in position $j$ is closed in at most $\sizeMax$ steps for any $j \geq 0$, which is to say that the \textit{direct} objective is satisfied.
\end{proof}

\subsection{\textbf{Games with one dimension}}
\label{subsec:wmp_oneDim}

We now study the \textit{fixed window mean-payoff} and the \textit{$\bounded$ window mean-payoff} objectives in one-dimension games. For the fixed window problem, we establish an algorithm that runs in time polynomial in the size of the game and in the size of the window and we show that memory is needed for both players. Note that this is in contrast to the mean-payoff objective, where $\playerTwo$ is memoryless even in the multi-dimension case (cf. Table~\ref{table:complexityAndMemory}). Moreover, the problem is shown to be P-hard even for polynomial window sizes. For the $\bounded$ window problem, we show equivalence with the fixed window problem for size $(\statesSize - 1)\cdot(\statesSize\cdot\largestW + 1)$, i.e., this window size is sufficient to win if possible. The $\bounded$ window problem is then shown to be in $\NPinter$ and at least as hard as mean-payoff games.

\captionsetup[algorithm]{font=scriptsize}

\smallskip\noindent\textbf{Fixed window: algorithm.} Given a game $\gameFullOneDim$ and a window size $\sizeMax \in \natStrict$, we present an iterative algorithm $\windowMPAlg$ (Alg.~\ref{alg:windowMP}) to compute the winning states of $\playerOne$ for the objective $\fixedWindowMPObjT{0}$. 
Initially, all states are potentially losing for $\playerOne$. The algorithm iteratively declares states to be winning, removes them, and continues the computation on the remaining subgame as follows. 
In every iteration, \textit{i)} $\directWinAlg$ computes the set $\wDirect$ of states from which $\playerOne$ can win the direct fixed window objective; 
\textit{ii)} it computes the attractor to $\wDirect$; and  
then proceeds to the next iteration on the remaining subgame 
(the restriction of $\game$ to a subset of states $A \subseteq \states$ is denoted 
$\game \downharpoonright A$). 
In every iteration, the states of the computed set $\wDirect$ are 
obviously winning for the fixed window objective.
Thanks to the prefix-independence of the fixed window objective, 
the attractor to $\wDirect$ is also winning.
Since $\playerTwo$ must avoid entering this attractor, 
$\playerTwo$ must restrict his choices to stay in the subgame,
and hence we iterate on the remaining subgame.
Thus states removed over all iterations are winning for $\playerOne$. This sequence of steps is essentially the computation of a greatest fixed point.
The key argument to establish correctness is as follows: when the algorithm stops, the remaining set of states $\overline{W}$ is such that $\playerTwo$ can ensure to stay in $\overline{W}$ and falsify the direct fixed window objective by forcing the appearance of one open window larger than $\sizeMax$. Since he stays in $\overline{W}$, he can repeatedly use this strategy to falsify the fixed window objective. Thus the remaining set $\overline{W}$ is winning for $\playerTwo$, and the correctness of the algorithm follows.
\vspace{-8mm}

\begin{figure}[htb]
\begin{minipage}[c]{0.52\linewidth}
\begin{algorithm}[H]
\caption{$\windowMPAlg(\game, \sizeMax)$}
\label{alg:windowMP}
\begin{algorithmic}
\footnotesize
    \REQUIRE $\gameFullOneDim$ and $\sizeMax \in \natStrict$
    \ENSURE $\wWin$ is the set of winning states for $\playerOne$ for $\fixedWindowMPObjT{0}$
    \STATE $n := 0$ ; $W := \emptyset$
    \REPEAT 
    \STATE $\wDirect^{n} := \directWinAlg(\game, \sizeMax)$
    \STATE $\wAttr^{n} := \attr_{\game}^{\playerOne}(\wDirect^{n})$ \COMMENT{attractor for $\playerOne$}
    \STATE $\wWin := \wWin \cup \wAttr^{n}$ ; $\game := \game \downharpoonright (\states \setminus \wWin)$ ; $n := n+1$
    \UNTIL{$\wWin = \states$ or $\wAttr^{n-1} = \emptyset$}
    \RETURN $\wWin$
\end{algorithmic}
\end{algorithm}
\end{minipage}
\hfill
\begin{minipage}[c]{0.42\linewidth}
\begin{algorithm}[H]
\caption{$\directWinAlg(\game, \sizeMax)$}
\label{alg:directWin}
\begin{algorithmic}
\footnotesize
    \REQUIRE $\gameFullOneDim$ and $\sizeMax \in \natStrict$
    \ENSURE $\wDirect$ is the set of winning states for $\playerOne$ for the objective $\directFixedWindowMPObjT{0}$
    \STATE $\wPos := \posSumAlg(\game, \sizeMax)$
    \IF{$\wPos = \states$ or $\wPos = \emptyset$}
        \STATE $\wDirect := \wPos$
    \ELSE
        \STATE $\wDirect := \directWinAlg(\game \downharpoonright \wPos, \sizeMax)$
    \ENDIF
    \RETURN $\wDirect$
\end{algorithmic}
\end{algorithm}
\end{minipage}
\end{figure} 

\vspace{-8mm}

\begin{algorithm}
\caption{$\posSumAlg(\game, \sizeMax)$}
\label{alg:posSum}
\begin{algorithmic}
\footnotesize
    \REQUIRE $\gameFullOneDim$ and $\sizeMax \in \natStrict$
    \ENSURE $\wPos$ is the set of winning states for $\goodWindowMPObjT{0}$
	\FORALL{$\state \in \states$}   
	\STATE $\posCost{0}{\state} := 0$
	\ENDFOR
	\FORALL{$i \in \{1, \ldots{}, \sizeMax\}$}
	\FORALL{$\state \in \states_{1}$}   
	\STATE $\posCost{i}{\state} := \max_{(\state, \state') \in \edges} \{ \weight((\state, \state')) + \posCost{i-1}{\state'}\}$
	\ENDFOR
	\FORALL{$\state \in \states_{2}$}      
    \STATE $\posCost{i}{\state} := \min_{(\state, \state') \in \edges} \{ \weight((\state, \state')) + \posCost{i-1}{\state'}\}$
	\ENDFOR
	\ENDFOR
    \RETURN $\wPos := \{ \state \in \states \,\vert\, \exists\, i,\, 1 \leq i \leq \sizeMax,\, \posCost{i}{\state} \geq 0\}$
\end{algorithmic}
\end{algorithm}

The main idea of algorithm $\directWinAlg$ (Alg.~\ref{alg:directWin}) is that to win the direct fixed window objective, $\playerOne$ must be able to repeatedly win the good window objective, which consists in ensuring a non-negative sum in at most $\sizeMax$ steps. Thus the algorithm consists in computing a least fixed point. A winning strategy of $\playerOne$ in a state $s$ is a strategy that enforces a non-negative sum and, \textit{as soon as the sum turns non-negative} (in some state $s'$), starts doing the same from~$s'$. It is important to start again immediately as it ensures that all suffixes along the path from $s$ to $s'$ also have a non-negative sum thanks to the inductive property of windows. That is, for any state $s''$ in between, the window from~$s''$ to $s'$ is closed.
The set of states from which $\playerOne$ can ensure winning for the good window objective is computed by subroutine $\posSumAlg$ (Alg.~\ref{alg:posSum}).
Intuitively, given a state $s \in \states$ and a number of steps $i \geq 1$, the value $\posCost{i}{s}$ is computed iteratively (from $\posCost{i-1}{s}$) and represents the best sum that $\playerOne$ can ensure from $s$ in exactly $i$ steps:
\begin{equation*}
\forall\, \state \in \states,\; \posCost{0}{\state} = 0 \;\wedge\; \posCost{i\,\geq 1}{\state} = \begin{cases}
\max_{(\state, \state') \in \edges} \{ \weight((\state, \state')) + \posCost{i-1}{\state'}\} &\text{ if } \state \in \states_{1}, \\
\min_{(\state, \state') \in \edges} \{ \weight((\state, \state')) + \posCost{i-1}{\state'}\} &\text{ if } \state \in \states_{2}.
\end{cases}
\end{equation*}
Hence, the set of winning states for $\playerOne$ is the set of states for which there exists some $i$, $1 \leq i \leq \sizeMax$ such that $\posCost{i}{s} \geq 0$. We state the correctness of $\posSumAlg$ in Lemma~\ref{lem:posSumAlg}.

\begin{lemma}
\label{lem:posSumAlg}
Algorithm $\posSumAlg$ computes the set of winning states of $\playerOne$ for the good window objective in time $\mathcal{O}\left(\edgesSize \cdot \sizeMax \cdot \bits\right)$, with $\bits = \lceil\log_{2} \largestW\rceil$, the length of the binary encoding of weights.
\end{lemma}

\begin{proof}
Let $\winningsPos \subseteq \states$ denote the winning states for $\goodWindowMPObjT{0}$. We prove that (a) $\state \in \winningsPos \Rightarrow \state \in \posSumAlg(\game, \sizeMax)$, and (b) $\state \in \posSumAlg(\game, \sizeMax) \Rightarrow s \in \winningsPos$.

We first consider case (a). From $s$, there exists a strategy of $\playerOne$ that enforces a non-negative sum after $\size$ steps, for some $\size$, $1 \leq \size \leq \sizeMax$. Hence, the value $\posCost{\size}{s}$ computed by the algorithm is non-negative and $\state \in \posSumAlg(\game, \sizeMax)$.

Case (b). Assume $\state \in \posSumAlg(\game, \sizeMax)$. By definition of the algorithm $\posSumAlg$, there exists some $\size \leq \sizeMax$ such that $\posCost{\size}{s}$ is positive. Consequently, taking the choice of $\size$ edges that achieves the maximum value defines a strategy for $\playerOne$ that ensures a positive sum after $\size$ steps, hence closing the window started in $\state$. That is, $\state \in \winningsPos$.

It remains to discuss the complexity of $\posSumAlg$. Clearly, it takes a number of elementary arithmetic operations which is bounded by $\mathcal{O}\left(\edgesSize \cdot \sizeMax\right)$ to compute the set $\wPos$ as each edge only needs to be visited once at each step $i$. Each elementary arithmetic operation takes time linear in the number of bits $\bits$ of the encoding of weights, that is, logarithmic in the largest weight $\largestW$. Hence, the time complexity of $\posSumAlg$ is $\mathcal{O}\left(\edgesSize \cdot \sizeMax \cdot \bits\right)$.
\end{proof}

Thanks to the previous lemma, we establish the algorithm solving the direct fixed window objective.

\begin{lemma}
\label{lem:directWinAlg}
Algorithm $\directWinAlg$ computes the set of winning states of $\playerOne$ for the direct fixed window mean-payoff objective in time $\mathcal{O}\left( \statesSize \cdot \edgesSize \cdot \sizeMax \cdot \bits\right)$, with $\bits = \lceil\log_{2} \largestW\rceil$, the length of the binary encoding of weights.
\end{lemma}

\begin{proof}
Let $\winningsDir$ be the set of winning states for $\directFixedWindowMPObjT{0}$, i.e., 
\begin{equation*}
\state \in \winningsDir \;\Leftrightarrow\; \exists\; \strat_{1} \in \strats_{1},\; \forall\, \strat_{2} \in \strats_{2},\; \outcome{\state}{\strat_{1}}{\strat_{2}} \in \directFixedWindowMPObjT{0}.
\end{equation*}
We first prove (a) $\state \in \directWinAlg(\game, \sizeMax) \Rightarrow  \state \in \winningsDir$, and then (b) $\state \in \winningsDir \Rightarrow \state \in \directWinAlg(\game, \sizeMax)$. First of all, notice that $\directWinAlg$ exactly computes the set of states $\wDirect$ such that a non-negative sum is achievable in at most $\sizeMax$ steps, using only states from which a non-negative sum can also be achieved in at most $\sizeMax$ steps (hence the property is defined recursively).

Consider case (a). Let $\state \in \wDirect$. Consider the following strategy of $\playerOne$.
\begin{enumerate}
\item Play the strategy prescribed by $\posSumAlg$ until a non-negative sum is reached. This is guaranteed to be the case in at most $\sizeMax$ steps. Let $\state'$ be the state that is reached in this manner.
\item By construction of $\wDirect$, we have that $\state' \in \wDirect$. Thus, play the strategy prescribed by $\posSumAlg$ in $\state'$.
\item Continue ad infinitum.
\end{enumerate}
We denote this strategy by $\strat_{1}$ and claim it is winning for the direct fixed window objective, i.e., $\state \in \winningsDir$. Indeed, consider any strategy of $\playerTwo$ and let $\play = \outcome{\state}{\strat_{1}}{\strat_{2}}$. We have $\play = s_{1}s_{2}\ldots{}s_{m_{1}}s_{m_{1}+1}\ldots{}s_{m_{2}}s_{m_{2}+1}\ldots{}$ with for all $j \geq 0,\, s_{j} \in \states$ and $s_{1} = s_{m_{0}} = \state$, such that all sequences $\rho(n) = s_{m_{n}} \ldots{} s_{m_{n+1}}$ are of length at most $\sizeMax+1$ ($\sizeMax$ steps) and such that all strict prefixes of $\rho(n)$ are strictly negative and all suffixes of $\rho(n)$ are positive. Indeed, starting in some state $s_{m_{n}}$, the strategy $\strat_{1}$ keeps a memory of the current sum and tries to reach a non-negative value (using the strategy prescribed by $\posSumAlg$). As soon as such a value is reached in a state $s_{m_{n+1}}$, the memory of the current sum kept by the strategy is reset to zero and the process is restarted. That way, for all $j$, $m_{n} \leq j <  m_{n+1}$, we have that the sum over the sequence from $s_{j}$ to $s_{m_{n+1}}$ is non-negative, hence all intermediate windows are also closed. Thus, the window property is satisfied everywhere along the play $\play$, starting in $s_{1} = \state$, which proves that $\state \in \winningsDir$.

Case (b). Let $\strat_{1}$ be a winning strategy of $\playerOne$ for $\directFixedWindowMPObjT{0}$. For any strategy $\strat_{2}$ of $\playerTwo$, the outcome is a play $\play = s_{1}s_{2}\ldots{}$ with $s_{1} = \state$ such that the window property is satisfied from all states. In particular, this implies, that for all $s_{j}$, strategy $\strat_{1}$ enforces a positive sum in at most $\sizeMax$ steps, that is, $s_{j} \in \posSumAlg(\game, \sizeMax)$. Since it is the case for all states $s_{j}$, we have that $\playerOne$ has a strategy to ensure a positive sum in at most $\sizeMax$ steps using only states from which this property is ensured. Therefore, we conclude that $\state \in \wDirect$.

Again, the number of calls of this algorithm is at most the number of states $\vert\states\vert$. Let $\posSumComp$ denote the complexity of algorithm $\posSumAlg$. Then, the complexity of algorithm $\directWinAlg$ is $\mathcal{O}\left( \statesSize \cdot \posSumComp\right)$.
\end{proof}

Finally, we prove the correctness of the algorithm for the fixed window problem.

\begin{lemma}
\label{lem:windowMPAlg}
Algorithm $\windowMPAlg$ computes the set of winning states of $\playerOne$ for the fixed window mean-payoff objective in time $\mathcal{O}\left( \statesSize^{2} \cdot \edgesSize \cdot \sizeMax \cdot \bits\right)$, with $\bits = \lceil\log_{2} \largestW\rceil$, the length of the binary encoding of weights.
\end{lemma}

\begin{proof}
Let $\winnings \subseteq \states$ be the set of states that are winning for $\fixedWindowMPObjT{0}$, i.e.,
\begin{equation*}
\state \in \winnings \;\Leftrightarrow\; \exists\, \strat_{1} \in \strats_{1},\; \forall\, \strat_{2} \in \strats_{2},\; \outcome{\state}{\strat_{1}}{\strat_{2}} \in \fixedWindowMPObjT{0}.
\end{equation*}
Note that since we set the threshold to be $0$ (w.l.o.g.), we may ignore the division by the window size $\size$ in Eq.~\eqref{eq:goodWindowObj}. We claim that $\windowMPAlg(\game, \sizeMax) = \winnings$. The proof is in two parts: (a) $\state \in \windowMPAlg(\game, \sizeMax) \Rightarrow \state \in \winnings$, and (b) $\state \in \winnings \Rightarrow \state \in \windowMPAlg(\game, \sizeMax)$.

We begin with (a). Let $(\wDirect)^{n \geq 0}$ and $(\wAttr)^{n \geq 0}$ be the finite sequences of sets computed by the algorithm.
We have that $\windowMPAlg(\game, \sizeMax) = \bigcup_{n \geq 0} \wAttr^{n}$. For any $n, n'$ such that $n \neq n'$, we have that $\wAttr^{n} \cap \wAttr^{n'} = \emptyset$ and $\wDirect^{n} \cap \wDirect^{n'} = \emptyset$. Moreover, for all $n \geq 0$, $\wDirect^{n} \subseteq \wAttr^{n}$. Let $\state \in \windowMPAlg(\game, \sizeMax)$. There exists a unique $n \geq 0$ such that $\state \in \wAttr^{n}$. By construction, from $\state$, $\playerOne$ has a strategy to reach and stay in $\wDirect^{n} \cup \wAttr^{n-1} \cup \wAttr^{n-2} \cup \ldots{} \wAttr^{0}$ and thus $\state$ is winning in the subgame $\game \downharpoonright (\states \setminus \wAttr^{n-1})$. However, $\playerTwo$ still has the possibility to leave $\wDirect^{n}$ and reach the set $\wAttr^{n-1} \cup \wAttr^{n-2} \cup \ldots{} \wAttr^{0}$. Since the sequence is finite and $\playerTwo$ cannot leave $\wDirect^{0}$, we have that at some point, any outcome is trapped in some set $\wDirect^{m}$, $0 \leq m \leq n$, in which $\playerOne$ wins the direct fixed window objective. Let $x$ be the length of the finite prefix outside the set $\wDirect^{m}$. The outcome satisfies the fixed {\windowMP} objective for $i = x$. Therefore, we have that $\state \in \winnings$.

Now consider (b). Let $\state \in \winnings$ be a winning state for $\fixedWindowMPObjT{0}$. We claim that $\state \in \windowMPAlg(\game, \sizeMax)$. Suppose it is not the case and consider the sequences $(\wDirect)^{n \geq 0}$ and $(\wAttr)^{n \geq 0}$ as before. We have that for all $n \geq 0$, $\state \not\in \wAttr^{n}$. In particular, $\playerTwo$ can force staying in $\states_{trap} = \states \setminus \bigcup_{n \geq 0} \wAttr^{n}$ when starting in $\state$. Since the algorithm has stopped, we have that $\directWinAlg(\game \downharpoonright \states_{trap}, \sizeMax) = \emptyset$. As algorithm $\directWinAlg$ is correct, from all states of $\states_{trap}$, $\playerTwo$ has a strategy to spoil the direct fixed window game, i.e., $\playerTwo$ can force a sequence of states such that there exists a position $j$ along it for which the window starting in $j$ stays open for at least $(\sizeMax + 1)$ steps, and such that this sequence remains in $\states_{trap}$. Therefore, $\playerTwo$ can force staying in $\states_{trap}$ and seeing infinitely often such sequences, hence $\playerOne$ is losing for the fixed {\windowMP} objective, which contradicts the fact that $\state \in \winnings$.

Finally, consider the complexity of the recursive algorithm $\windowMPAlg$. Notice that at least one state is declared winning at each iteration. The number of calls is thus at most the number of states $\vert\states\vert$. Computing the attractor is linear in the number of edges $\vert\edges\vert \leq \vert\states\vert^{2}$. The overall complexity is thus $\mathcal{O}\left( \statesSize \cdot (\edgesSize + \directWinComp)\right)$, where $\directWinComp$ is the complexity of the $\directWinAlg$ algorithm.
\end{proof}

\smallskip\noindent\textbf{Fixed window: lower bounds.} Thanks to the correctness of algorithm $\windowMPAlg$, we also deduce linear upper bounds (in $\statesSize \cdot \sizeMax$) on the memory needed for both players (Lemma~\ref{lem:memoryOneDimFixed}). Indeed, let $s \in \states$ be a winning state for $\playerOne$. A winning strategy $\strat_{1}$ for $\playerOne$ is to (a) reach the set of states $\wDirect^{n}$ that are winning for the direct fixed window objective in the subgame restricted to states $\wDirect^{n} \setminus \wAttr^{n-1}$, then (b) repeatedly play the strategy prescribed by $\posSumAlg$ in this subgame (i.e., enforce a non-negative sum in less than $\sizeMax$ steps, see proof of Lemma~\ref{lem:directWinAlg}). If $\playerTwo$ leaves for a lower subgame restricted to $\wAttr^{n'}$, $n' < n$, the strategy is to start again part (a) in this subgame. Part (a) is memoryless as it uses a classical attractor strategy. Part (b) requires to consider, for each state $\state'$ in the set computed by $\directWinAlg$, a number of memory states which is bounded by $\sizeMax$, as the only memory needed is to select the corresponding successor state that will maximize the $\posCost{\size}{\state'}$ value, for all possible values of $\size$, the number of steps remaining to close a window. Similarly, $\playerTwo$ needs to be able to prevent the closing of a window repeatedly, and therefore also possibly needs $\sizeMax$ memory states for each state of the game.

To illustrate that memory is needed by both players, consider the following examples. First, consider a game where all states belong to $\playerOne$ and such that the play starts in a central state $s$ and in $s$, there are three outgoing edges, towards three simple cycles $\cycle_{1}$, $\cycle_{2}$, and $\cycle_{3}$. All other states have only one outgoing edge. Cycle $\cycle_{1}$ is composed of six edges of successive weights $3, 3, 5, -1, -1$ and $-5$. Cycle $\cycle_{2}$ is $7, -1$ and $-9$. Cycle $\cycle_{3}$ is $5, 5$ and $-11$. The objective is $\fixedWindowMPObjTW{0}{\sizeMax = 4}$. Clearly, from some point on, a winning strategy of $\playerOne$ has to infinitely alternate between cycles in the following way: $(\cycle_{1}\cycle_{2}\cycle_{3})^{\omega}$. Any other alternation leads to a bad window appearing infinitely often: hence, the decision of $\playerOne$ in $s$ depends on the remaining number of steps to ensure a good window. Second, consider a similar game but with all states belonging to $\playerTwo$. Again, the initial state is central and there are two cycles $\cycle_{1}$ and $\cycle_{2}$ such that $\cycle_{1}$ is $1$ followed by $-1$, and $\cycle_{2}$ is $-1, -1$ and $2$. The objective is $\fixedWindowMPObjTW{0}{\sizeMax = 3}$. If $\playerTwo$ is memoryless, both possible strategies induce a winning play for $\playerOne$. On the other hand, if $\playerTwo$ is allowed to alternate, he can choose the play $(\cycle_{1}\cycle_{2})^{\omega}$ which will be losing for $\playerOne$ as the window $-1, -1, -1$ will appear infinitely often.

\begin{lemma}
\label{lem:memoryOneDimFixed}
In one-dimension games with a fixed window mean-payoff objective, memory is needed by both players and linear memory in the number of states times the window size is sufficient.
\end{lemma}

Through Lemma~\ref{lem:windowMPAlg}, we have shown that the fixed window problem admits a polynomial (in $\statesSize$, $\bits$ and $\sizeMax$) algorithm. In Lemma~\ref{lem:oneDimPHardness}, we prove that even for window size $\sizeMax = 1$ and weights $\{-1, 1\}$, the problem is P-hard. This is via a reduction from reachability games. By making the target states absorbing with a self-loop of weight $1$, and giving weight $-1$ on all other edges, we obtain the reduction, as reaching a target state is now the only way to ensure that windows close.

\begin{lemma}
\label{lem:oneDimPHardness}
In two-player one-dimension games, the fixed window mean-payoff problem is P-hard, even for $\sizeMax = 1$ and weights $\{-1, 1\}$.
\end{lemma}

\begin{proof}
Let $\game_{r} = (\states_{1}, \states_{2}, \edges)$ be an unweighted game with a reachability objective asking to visit (at least once) a state of the set $R \subseteq \states$. We build the game $\game = (\states_{1}, \states_{2}, \edges', \weight)$ by (a) making the target states absorbing with a self-loop of weight $1$, i.e., for all $s \in R$, we have $(s, s) \in \edges'$ and $\weight((s,s)) = 1$, and (b) putting weight $-1$ on all other edges, i.e., for all edge $(s, t) \in \edges$ such that $s \not\in R$, we have $(s, t) \in \edges'$ and $\weight((s,s)) = -1$. We claim that $\playerOne$ has a winning strategy in $\game_{r}$ from a state $s \in \states$ if and only if he has a winning strategy for the objective $\fixedWindowMPObjTW{0}{\sizeMax = 1}$ in $\game$ from $s \in \states$. Indeed, it is clear that any outcome that never reaches the target set is such that all windows stay indefinitely open, and conversely, an outcome that reaches this set after $n$ steps is winning for the fixed window objective with $i = n$. Since deciding the winner in reachability games is P-complete, this concludes our proof.
\end{proof}

\smallskip\noindent\textbf{Fixed window: summary.} We sum up the complexity analysis of the fixed window problem in Theorem~\ref{thm:oneDimFixed}.

\begin{theorem}
\label{thm:oneDimFixed}
In two-player one-dimension games, (a) the fixed arbitrary window mean-payoff problem is decidable in time $\mathcal{O}\left( \statesSize^{2} \cdot \edgesSize \cdot \sizeMax \cdot \bits\right)$, with $\bits = \lceil\log_{2} \largestW\rceil$, the length of the binary encoding of weights, and (b) the fixed polynomial window mean-payoff problem is P-complete. In general, both players require memory, and memory of size linear in $\statesSize \cdot \sizeMax$ is sufficient.
\end{theorem}

\smallskip\noindent\textbf{$\Bounded$ window: algorithm.} In the following, we focus on the $\bounded$ window mean-payoff problem for two-player one-dimension games. We start with two technical lemmas related to the classical supremum total-payoff threshold problem. Using these lemmas, we establish an algorithm to solve the $\bounded$ window problem. This algorithm uses a polynomial number of calls to an oracle solving the total-payoff threshold problem, hence proving that the $\bounded$ window problem is in $\NPinter$~\cite{gawlitza2009}. As a corollary, we get an interesting bound on the window size needed to win the fixed window problem if possible.

The first technical lemma (Lemma~\ref{lem:oneDimFinTech1}) states that if $\playerOne$ has a strategy to win the supremum total-payoff objective from some state $\initState$, then he can force a non-negative sum from this state in at most $(\statesSize - 1)\cdot(\statesSize\cdot\largestW + 1)$ steps, i.e., he wins the good window objective for this window size.

\begin{lemma}
\label{lem:oneDimFinTech1}
Let $\gameFullOneDim$ be a two-player one-dimension game. If $\playerOne$ has a strategy to win for objective $\objTPSupT{0}$ from initial state $\initState \in \states$, then $\playerOne$ also has a strategy to win for the good window objective $\goodWindowMPObjT{0}$ from $\initState$ for $\sizeMax = (\statesSize - 1)\cdot(\statesSize\cdot\largestW + 1)$.
\end{lemma}

This result is obtained by considering a memoryless winning strategy of $\playerOne$ for the total-payoff and the decomposition in simple cycles of any consistent outcome where (a) either simple cycles are strictly positive, or (b) they are of value zero but preceded by a non-negative prefix.

\begin{proof}
Let $\strat_{1} \in \stratsMemoryless_{1}$ be a memoryless winning strategy of $\playerOne$ for $\objTPSupT{0}$. Our claim is that for all possible outcome $\play$ consistent with $\strat_{1}$ starting in the initial state $\initState$, there exists a prefix $\prefix$ of $\play$ of size at most $\sizeMax$ such that the total sum of weights over $\prefix$ is non-negative. Let $\play$ be any outcome consistent with $\strat_{1}$ and $\prefix_{1}$ its prefix of length $(\statesSize - 1)\cdot(\statesSize\cdot\largestW + 1)$. Consider the cycle decomposition (see the proof of Lemma~\ref{lem:relationBoundedClassical}) of $\prefix_{1}$: $\mathcal{A}, \cycle_{1}, \cycle_{2}, \ldots{}, \cycle_{m}, \mathcal{B}$, with $\mathcal{A}$ the prefix before the first cycle and $\mathcal{B}$ the suffix after the last cycle in $\prefix_{1}$. The total length of the acyclic part is $\vert\mathcal{A}\vert + \vert\mathcal{B}\vert < \statesSize - 1$. We claim that there exists a prefix $\prefix$ of $\prefix_{1}$ such that the total sum of weights over $\prefix$ is non-negative. Consider the following arguments:
\begin{enumerate}
\item No cycle $\cycle$ in $\{\cycle_{1}, \ldots{}, \cycle_{m}\}$ can be strictly negative. Otherwise, since $\strat_{1}$ is memoryless, $\playerTwo$ could force cycling in such a cycle forever and the play would be losing for the supremum total-payoff objective, which contradicts $\strat_{1}$ being a winning strategy.
\item Assume that there exists a cycle $\cycle$ in $\{\cycle_{1}, \ldots{}, \cycle_{m}\}$ such that the sum of weights over this cycle is zero. We define the \textit{high point} of a cycle as the first state where the sum from the start of the cycle takes its highest value. Then, the prefix $\prefix$ of $\prefix_{1}$ up to this high point is non-negative and we are done. Indeed, assume it is not the case. Then, the running sum over the outcome $\play$ is strictly negative when reaching the high point, and stays strictly negative in all positions along the cycle $\cycle$, by definition of the high point. Therefore, $\playerTwo$ can force cycling forever in $\cycle$ since $\strat_{1}$ is memoryless and the outcome becomes losing for the total-payoff objective.

\item So assume there are only strictly positive cycles in the cycle decomposition of $\prefix_{1}$, that is, they all have a total sum of value at least $1$. The total sum over $\cycle_{1}, \ldots{}, \cycle_{m}$ is at least equal to $m$. Since each cycle is of length at most $\statesSize$ and $\mathcal{A} \cup \mathcal{B}$ is of length at most $\statesSize - 1$, we have that the number of cycles $m$ in the cycle decomposition of $\prefix_{1}$ is at least $((\statesSize - 1)\cdot(\statesSize\cdot\largestW + 1) - (\statesSize - 1))\,/\,\statesSize = (\statesSize - 1)\cdot\largestW$. Given that the total sum over prefix $\mathcal{A}$ is at least $- (\statesSize - 1)\cdot\largestW$, we obtain that $\prefix = \mathcal{A}\, \cycle_{1}\, \cycle_{2} \ldots{} \cycle_{m}$ is the desired prefix with a non-negative total sum, and its length is bounded by $(\statesSize - 1)\cdot(\statesSize\cdot\largestW + 1)$.
\end{enumerate}
This concludes our proof.
\end{proof}

The second technical lemma (Lemma~\ref{lem:oneDimFinTech2}) shows that if $\playerTwo$ has a strategy to ensure that the supremum total-payoff from some state $\initState$ is strictly negative, then he has a memoryless strategy to do so and any outcome $\play$ starting in $\initState$ and consistent with this strategy is such that the direct $\bounded$ window mean-payoff objective is not satisfied.

\begin{lemma}
\label{lem:oneDimFinTech2}
Let $\gameFullOneDim$ be a two-player one-dimension game. If $\playerTwo$ has a spoiling strategy for objective $\objTPSupT{0}$ from initial state $\initState \in \states$, then $\playerTwo$ has a strategy $\strat_{2} \in \stratsMemoryless_{2}$ to ensure that for all possible outcome $\play = s_{0}s_{1}\ldots{}$ consistent with $\strat_{2}$ starting in $s_{0} = \initState$, there exists a position $i \geq 0$ such that for all window sizes $l \geq 1$, the total sum of weights on the window from $s_{i}$ to $s_{i+l}$ is strictly negative.
\end{lemma}

\begin{proof}
By contradiction. Let $\strat_{2} \in \stratsMemoryless_{2}$ be a memoryless spoiling strategy for objective $\objTPSupT{0}$ from $\initState \in \states$. Let $\play$ be a consistent outcome and assume that it does not respect the lemma, i.e., for all positions $i \geq 0$, there exists a window size $l \geq 1$ such that the window from $s_{i}$ to $s_{i+l}$ is non-negative. Then the play $\play$ can be decomposed as a sequence of finite windows of non-negative weights. Hence, the total sum from $s_{0} = \initState$ takes infinitely often values at least equal to zero and the limit of its suprema is non-negative. This is in contradiction to $\strat_{2}$ being a winning strategy for $\playerTwo$.
\end{proof}

Thanks to Lemma~\ref{lem:oneDimFinTech1} and Lemma~\ref{lem:oneDimFinTech2}, we are now able to establish an algorithm (Alg.~\ref{alg:finiteProblem}) to solve the $\bounded$ window mean-payoff problem on two-player one-dimension games, and to deduce $\NPinter$-membership of the problem. Lemma~\ref{lem:oneDimFiniteNPinter} states its correctness.

Algorithm $\finiteProblemAlg$ (Alg.~\ref{alg:finiteProblem}) computes via a subroutine $\unbNegWindowAlg$ the set of states from which $\playerTwo$ can force the visit of a position such that the window opening in this position never closes.
Clearly, to prevent $\playerOne$ from winning the $\bounded$ window problem, $\playerTwo$ must be able to do so repeatedly as the prefix-independence of the objective otherwise gives the possibility to wait that all such bad positions are encountered before taking the windows into account. Therefore, the states that are not in $\unbNegWindowAlg(\game)$, as well as their attractor, are winning for $\playerOne$. Since the choices of $\playerTwo$ are reduced because of the attractor of $\playerOne$ being declared winning, we compute in several steps, adding new states to the set of winning states for $\playerOne$ up to stabilization.

Now consider the subroutine $\unbNegWindowAlg$ (Alg.~\ref{alg:unbNegWindow}). Its correctness is based on Lemma~\ref{lem:oneDimFinTech2}. Indeed, it computes the set of states from which $\playerTwo$ can force a position for which the window never closes. To do so, it suffices to compute the attractor for $\playerTwo$ of the set of states from which $\playerTwo$ can enforce a strictly negative supremum total-payoff. Routine $\negSupTPAlg$ denotes a call to an oracle solving the total-payoff problem, which is known to belong to $\NPinter$~\cite{gawlitza2009}. Precisely,
\begin{equation*}
\negSupTPAlg(\game) = \left\lbrace \state \in \states \mid \exists\, \strat_{2} \in \strats_{2},\, \forall\, \strat_{1} \in \strats_{1},\, \forall\, \play \in \outcome{\state}{\strat_{1}}{\strat_{2}},\, \tpaySup(\play) < 0 \right\rbrace .
\end{equation*}
Again, we compute the fixed point of the sequence as the choices of $\playerOne$ are reduced at each iteration.

\begin{figure}[tb]
\begin{minipage}[c]{0.44\linewidth}
\begin{algorithm}[H]
\caption{$\finiteProblemAlg(\game)$}
\label{alg:finiteProblem}
\begin{algorithmic}
\footnotesize
    \REQUIRE Game $\gameFullOneDim$
    \ENSURE $\wFP$ is the set of winning states for $\playerOne$ for the $\bounded$ window mean-payoff problem
    \STATE $\wFP := \emptyset$
    \STATE $L := \unbNegWindowAlg(\game)$
    \WHILE{$L \neq \states \setminus \wFP$}
    \vspace{1mm}
    	\STATE $\wFP := \attr^{\playerOne}_{\game} (\states \setminus L)$
    \vspace{1mm}
    	\STATE $L := \unbNegWindowAlg\Big(\game \downharpoonright (\states \setminus \wFP)\Big)$
    \ENDWHILE
    \RETURN $\wFP$
\end{algorithmic}
\end{algorithm}
\end{minipage}
\hfill
\begin{minipage}[c]{0.5\linewidth}
\begin{algorithm}[H]
\caption{$\unbNegWindowAlg(\game)$}
\label{alg:unbNegWindow}
\begin{algorithmic}
\footnotesize
    \REQUIRE Game $\gameFullOneDim$
    \ENSURE $L$ is the set of states from which $\playerTwo$ can force a position for which the window never closes
    \STATE $p := 0$ ; $L_{0} := \emptyset$
    \REPEAT 
    \STATE $L_{p+1} := L_{p} \cup \attr^{\playerTwo}_{\game \downharpoonright (\states \setminus L_{p})} \Big( \negSupTPAlg\big(\game \downharpoonright (\states \setminus L_{p})\big)\Big)$
    \STATE $p := p + 1$
    \UNTIL{$L_{p} = L_{p-1}$}
    \RETURN $L := L_{p}$
\end{algorithmic}
\end{algorithm}
\end{minipage}
\end{figure}

The main idea of the correctness proof is that from all states in $\overline{\wFP}$, $\playerTwo$ has an infinite-memory winning strategy which is played in rounds, and in round $n$ ensures an open window of size at least $n$ by playing the total-payoff strategy of $\playerTwo$ for at most $n\cdot\statesSize$ steps, and then proceeds to round $(n+1)$ to ensure an open window of size $(n+1)$, and so on. Hence, windows stay open for arbitrary large periods and the bounded window objective is falsified.

\begin{lemma}
\label{lem:oneDimFiniteNPinter}
Given a two-player one-dimension game $\gameFullOneDim$, the algorithm $\finiteProblemAlg$ computes the set of winning states for $\playerOne$ for the $\bounded$ window mean-payoff objective of threshold $0$ in time $\mathcal{O}(\statesSize^{2} \cdot (\edgesSize + \mathbb{C}))$, where $\mathbb{C}$ is the complexity of algorithm $\negSupTPAlg$, i.e., the complexity of computing the set of winning states in a two-player one-dimension supremum total-payoff game. Thus, algorithm $\finiteProblemAlg$ is in $\NPinter$.
\end{lemma}

\begin{proof}
It suffices to show that for all states in $\wFP = \finiteProblemAlg(\game)$, there exists a winning strategy of $\playerOne$, whereas for all states in $\states \setminus \wFP$, there exists one of $\playerTwo$. 

Consider a state $s \in \wFP$. Consider $(L^{m})_{0 \leq m \leq n}$, the finite sequence of sets $L$ that are computed by $\finiteProblemAlg$, with $L_{0} = \unbNegWindowAlg(\game)$; and $(\wFP^{m})_{0 \leq m \leq n}$, the corresponding finite sequence of sets $\wFP$ where $\wFP^{0} = \emptyset$ is empty and $\wFP^{n} = \wFP$ is the returned set of winning states. For all $m', m$, $0 \leq m' < m \leq n$, we have that $\wFP^{m} \supset \wFP^{m'}$ and $L^{m} \subset L^{m'}$. By construction, there exists $m$, $1 \leq m \leq n$ such that $s \in \wFP^{m} = \attr^{\playerOne}_{\game}(\states \setminus L^{m-1})$. In the subgame $\game \downharpoonright ((\states \setminus L^{m-1}) \setminus \wFP^{m-1})$, $\playerOne$ has a memoryless~\cite{gimbert2004} winning strategy for the supremum total-payoff objective. Hence, consider the strategy $\strat_{1}$ of $\playerOne$ which is to reach the set $(\states \setminus L^{m-1})$ (in at most $\statesSize$ steps) and then play the memoryless total-payoff strategy in the subgame. It is possible for $\playerTwo$ to force leaving this subgame for a lower subset $\wFP^{m'} \subset \wFP^{m}$ with $m' < m$ but since the sequence is finite, any outcome is ultimately trapped in some subgame $\game \downharpoonright ((\states \setminus L^{m''}) \setminus \wFP^{m''})$. Therefore, repeating the strategy $\strat_{1}$ in each subgame ensures that after a finite number of steps (and hence a finite number of positions for which windows never close), a bottom subgame $\game \downharpoonright ((\states \setminus L^{m''}) \setminus \wFP^{m''})$ is reached and, by Lemma~\ref{lem:oneDimFinTech1}, strategy $\strat_{1}$ ensures satisfaction of the good window objective for $\sizeMax = (\statesSize - 1)\cdot(\statesSize\cdot\largestW + 1)$ in this subgame. Moreover, since this strategy never visits states out of the bottom subgame, it ensures an inductive window from every state, regardless of the past. Hence, all intermediate windows are also closed and this strategy is winning for $\fixedWindowMPObjT{0} \subseteq \finiteWindowMPObjT{0}$ from the initial state $s$. The states that are only visited finitely often before reaching the bottom subgame have no consequence thanks to the prefix-independence of the $\bounded$ window mean-payoff objective. 

As for $\playerTwo$, consider a state $s \in \states \setminus \wFP$. Consider $(L_{p})_{0 \leq p \leq q}$, the finite sequence of sets $L$ that are computed in the last call to $\unbNegWindowAlg$ by $\finiteProblemAlg$, with $L_{0} = \emptyset$. We define the sequences $(N_{p})_{1 \leq p \leq q}$ and $(A_{p})_{1 \leq p \leq q}$ as $N_{p} = \negSupTPAlg(\game \downharpoonright (\states \setminus L_{p-1}))$ and $A_{p} = L_{p} \setminus L_{p-1} = \attr^{\playerTwo}_{\game \downharpoonright (\states \setminus L_{p-1})} (N_{p})$. We have that $s \in L_{p}$ for some $p$ between~$1$ and $q$. An infinite memory winning strategy for $\playerTwo$ is played in rounds. In round $n$, $\playerTwo$ acts as follows. (a) If the current state is in $A_{p}$, play the attractor to $N_{p}$ and then play the optimal strategy for the supremum total-payoff in $N_{p}$ to ensure that no window will have a non-negative sum for $n$ steps. (b) $\playerOne$ can leave the set $A_{p}$ for some lower set $A_{p'}$, $1 \leq p' < p$. If so, play the attractor to $N_{p'}$ and continue. Ultimately, any outcome is trapped in some set $N_{p''}\setminus A_{p''-1}$, with $1 \leq p'' \leq q$ and $A_{0} = \emptyset$, as in $N_{1}$, $\playerOne$ cannot leave. There $\playerOne$ cannot prevent the window being strictly negative for $n$ steps. When such a window has been enforced for $n$ steps, move to round $n+1$ and start again. This strategy ensures that the $\bounded$ window problem is not satisfied as, infinitely often, windows stay open for arbitrary large periods along any outcome. 

Finally, we discuss the complexity of algorithm $\finiteProblemAlg$. Let $\mathbb{C}$ be the complexity of routine $\negSupTPAlg$, that is, the complexity of solving a one-dimen\-sion supremum total-payoff game. The total complexity of subalgorithm $\unbNegWindowAlg$ is $\mathcal{O}(\statesSize\cdot(\edgesSize + \mathbb{C}))$ as the sequence of computations is of length at most $\statesSize$ and each computation takes time $\mathcal{O}(\edgesSize + \mathbb{C})$. The overall complexity of $\finiteProblemAlg$ is thus $\mathcal{O}(\mathbb{C} + \statesSize\cdot(\edgesSize + \statesSize\cdot(\edgesSize + \mathbb{C}))) = \mathcal{O}(\statesSize^{2} \cdot (\edgesSize + \mathbb{C}))$.
\end{proof}

An interesting corollary of Lemma~\ref{lem:oneDimFinTech1} and Lemma~\ref{lem:oneDimFiniteNPinter} is that the sets of winning states coincide for objectives $\fixedWindowMPObjTW{0}{\sizeMax = (\statesSize - 1)\cdot(\statesSize\cdot\largestW + 1)}$ and $\finiteWindowMPObjT{0}$, therefore proving $\NPinter$-membership for the subset of fixed window problems with window size at least $\sizeMax$ (hence an algorithm independent of the window size whereas Lemma~\ref{lem:directWinAlg} gives an algorithm which is polynomial in the window size).

\begin{corollary}
\label{cor:oneDimCor}
In two-player one-dimension games, the fixed window mean-payoff problem is in $\NPinter$ for window size at least equal to $(\statesSize - 1)\cdot(\statesSize\cdot\largestW + 1)$.
\end{corollary}

\smallskip\noindent\textbf{$\Bounded$ window: lower bounds.} Algorithm $\finiteProblemAlg$ (Lemma~\ref{lem:oneDimFiniteNPinter}) provides memoryless winning strategies for $\playerOne$ (attractor + memoryless strategy for total-payoff) and infinite-memory winning strategies for $\playerTwo$ (delaying the closing of windows for increasing number of steps each round) in one-dimension $\bounded$ window mean-payoff games. Lemma~\ref{lem:oneDimFiniteMemory} states that infinite memory is necessary for $\playerTwo$, as discussed in Section~\ref{subsec:wmp_def}: $\playerTwo$ cannot use the zero cycle forever, but he must cycle long enough to defeat any finite window. Hence, its strategy needs to cycle for longer and longer, which requires infinite memory.

\begin{lemma}
\label{lem:oneDimFiniteMemory}
In one-dimension games with a $\bounded$ window mean-payoff objective, (a) memoryless strategies suffice for $\playerOne$, and (b) infinite-memory strategies are needed for $\playerTwo$ in general.
\end{lemma}

In Lemma~\ref{lem:oneDimMPReduction}, we give a polynomial reduction from mean-payoff games to $\bounded$ window mean-payoff games, therefore showing that a polynomial algorithm for the $\bounded$ window problem would solve the long-standing question of the P-membership of the mean-payoff threshold problem. The proof relies on technical lemmas providing intermediary reductions. First, we prove that given a game $\game$, deciding if $\playerOne$ has a strategy to ensure a non-negative mean-payoff can be reduced to deciding if $\playerOne$ has a strategy to ensure a strictly positive mean-payoff when weights are shifted positively by a sufficiently small $\varepsilon$ (Lemma~\ref{lem:oneDimMPReductionTech1}). Second, we apply Lemma~\ref{lem:relationBoundedClassical} on the shifted game to prove that winning this objective implies winning the $\bounded$ window problem. This gives one direction of the reduction. For the other one, we show that given a game $\game$, if $\playerOne$ has a strategy to win the $\bounded$ window problem when weights are shifted positively by a sufficiently small $\varepsilon$, he has one to win the mean-payoff threshold problem in $\game$.

We define the following notation: given a two-player one-dimension game $\gameFullOneDim$ and $\varepsilon \in \rat$, let $\game_{+\varepsilon} = (\states_{1}, \states_{2}, \edges, \weight_{+\varepsilon})$ be the game obtained by shifting all weights by $\varepsilon$, that is, for all $e \in \edges$, $\weight_{+\varepsilon}(e) = \weight(e) + \varepsilon$.\footnote{Note that $\weight_{+\varepsilon}$ can be transformed into an integer valued function without changing the answers to the considered decision problems.}

\begin{lemma}
\label{lem:oneDimMPReductionTech1}
For all one-dimension game $\gameFullOneDim$ with integer weights, for all $\varepsilon$, $0 < \varepsilon <1/\statesSize$, for all initial state $\state \in \states$, $\playerOne$ has a strategy to ensure a non-negative mean-payoff in $\game$ if and only if $\playerOne$ has a strategy to ensure a strictly positive mean-payoff in $\game_{+\varepsilon}$.
\end{lemma}

\begin{proof}
Consider a memoryless winning strategy of $\playerOne$ in $\game$ from initial state $\state \in \states$. All simple cycles in consistent outcomes have a sum of weights at least equal to zero. Hence, the corresponding outcome in $\game_{+\varepsilon}$ is such that all simple cycles of length $n$ have sums at least equal to $n\cdot \varepsilon > 0$, which proves that the strategy is also winning in $\game_{+\varepsilon}$.

Consider a memoryless winning strategy of $\playerTwo$ in $\game$ from initial state $\state \in \states$. All simple cycles in consistent outcomes have a strictly negative sum of weights, that is the sum is at most equal to $-1$. Hence, the corresponding outcome in $\game_{+\varepsilon}$ is such that all simple cycles of length $n$ have sums at most equal to $-1 + n\cdot \varepsilon$. Since $n \leq \statesSize$ and $\varepsilon <1/\statesSize$, we have that the sum is strictly negative,  which proves that the strategy is also winning in $\game_{+\varepsilon}$.

By determinacy of mean-payoff games, we obtain the claim.
\end{proof}

\begin{lemma}
\label{lem:oneDimMPReductionTech2}
For all one-dimension game $\gameFullOneDim$ with integer weights, for all $\varepsilon$, $0 < \varepsilon <1/\statesSize$, for all initial state $\state \in \states$, if $\playerOne$ has a strategy to win the $\bounded$ window mean-payoff problem in $\game_{+\varepsilon}$, then $\playerOne$ has a strategy to win the mean-payoff threshold problem in $\game$.
\end{lemma}

\begin{proof}
Assume there exists a winning strategy of $\playerOne$ for the $\bounded$ window mean-payoff problem in $\game_{+\varepsilon}$ from initial state $\state \in \states$. By Lemma~\ref{lem:relationBoundedClassical}, assertion (a), we have that this strategy ensures a non-negative mean-payoff in $\game_{+\varepsilon}$. By shifting weights by $-\varepsilon$, this can be equivalently expressed as (Prop. A) the existence of a strategy of $\playerOne$ ensuring a mean-payoff at least equal to $-\varepsilon$ in the game $\game$.

For sufficiently small values of $\varepsilon$, that is for $0 < \varepsilon <1/\statesSize$, we claim that (Prop. A) implies that (Prop. B) $\playerOne$ has a strategy to ensure a non-negative mean-payoff in $\game$. By contradiction, assume this implication is false, that is we have that (Prop. A) is true and (Prop. B) is not. It implies the following.
\begin{itemize}
\item \textit{(Prop. A) is true}: $\playerOne$ has a memoryless strategy to ensure that the mean-payoff is at least equal to $-\varepsilon$, i.e., strictly greater than $-1/\statesSize$.
\item \textit{(Prop. B) is false}: $\playerTwo$ has a memoryless strategy to ensure that all simple cycles in consistent outcomes have a sum of weights at most $-1$. Hence, this strategy ensures a mean-payoff at most equal to $-1/\statesSize$.
\end{itemize}
Obviously, it is not possible to have both (Prop. A) true and (Prop. B) false for any initial state $\state \in \states$, hence proving our claim.
\end{proof}

\begin{lemma}
\label{lem:oneDimMPReduction}
The one-dimension mean-payoff problem reduces in polynomial time to the $\bounded$ window mean-payoff problem.
\end{lemma}

\begin{proof}
Let $\gameFullOneDim$ be a game with integer weights, and $\initState \in \states$ be the initial state. Let $\varepsilon$ be any rational value such that $0 < \varepsilon < 1/\statesSize$. We claim that the answer to the mean-payoff threshold problem in $\game$ is $\yes$ if and only if the answer to the $\bounded$ window mean-payoff problem in $\game_{+\varepsilon}$ is $\yes$.

The left-to-right implication is proved in two steps. Assume the answer to the mean-payoff threshold problem in $\game$ is $\yes$. First, by Lemma~\ref{lem:oneDimMPReductionTech1}, we have that $\playerOne$ has a strategy to ensure a strictly positive mean-payoff in $\game_{+\varepsilon}$. Second, by Lemma~\ref{lem:relationBoundedClassical}, assertion (b), this implies that the answer to the $\bounded$ window mean-payoff problem in $\game_{+\varepsilon}$ is $\yes$.

The right-to-left implication is straightforward application of Lemma~\ref{lem:oneDimMPReductionTech2}.
\end{proof}

\begin{remark}
The reduction established in Lemma~\ref{lem:oneDimMPReduction} cannot be reversed in order to solve $\bounded$ window mean-payoff games via classical mean-payoff games. Indeed, the reduction relies on the absence of simple cycles of value zero in the game $\game_{+\varepsilon}$, which is not verified in general if the reduction starts from arbitrary $\bounded$ window mean-payoff games. Indeed it does not suffice to shift the weights symmetrically by $-\varepsilon$ to obtain an equivalent mean-payoff game, as witnessed by Fig.~\ref{fig:wmpEx1}, for which any negative shift gives a game losing for the mean-payoff threshold problem, while the $\bounded$ window problem on the original game is satisfied.
\end{remark}

\smallskip\noindent\textbf{$\Bounded$ window: summary.} We close our study of two-player one-dimension games with Theorem~\ref{thm:oneDimFinite}.

\begin{theorem}
\label{thm:oneDimFinite}
In two-player one-dimension games, the $\bounded$ window mean-payoff problem is in $\NPinter$ and at least as hard as mean-payoff games. Memoryless strategies suffice for $\playerOne$ and infinite-memory strategies are required for $\playerTwo$ in general.
\end{theorem}

\subsection{\textbf{Games with $\dimension$ dimensions}}
\label{subsec:wmp_multiDim}

In this section, we address the case of two-player games with multi-dimension weights. For the \textit{fixed window mean-payoff problem}, we first present an EXPTIME algorithm that computes the winning states of $\playerOne$. We also establish lower bounds on the complexity of the fixed window problem: we show that the problem is EXPTIME-hard (both in the case of fixed weights and arbitrary dimensions, and in the case of a fixed number of dimensions and arbitrary weights) for arbitrary window sizes, whereas it is PSPACE-hard for polynomial window sizes. We show that exponential memory is both sufficient and necessary in general for both players, even for polynomial window sizes. For the \textit{$\bounded$ window mean-payoff problem}, we prove non-primitive recursive hardness.

\smallskip\noindent\textbf{Fixed window: algorithm.} We start by providing an EXPTIME algorithm via a reduction from a fixed window mean-payoff game $\gameFull$ to an exponentially larger unweighted \coBuchi game $\gameCB$ (where the objective of $\playerOne$ is to avoid visiting a set of bad states infinitely often).

\begin{lemma}
\label{lem:multiDimReducCB}
The fixed window mean-payoff problem over a multi-weighted game $\game$ reduces in exponential time to the \coBuchi problem on an exponentially larger game $\gameCB$.
\end{lemma}

Recall that a winning play is such that, starting in some position $i \geq 0$, in all dimensions, all opening windows are closed in at most $\sizeMax$ steps.
We keep a counter of the sum over the sequence of edges and as soon as it turns non-negative (in at most $\sizeMax$ steps), we reset the sum counter and start a new sequence (which also must become non-negative in at most $\sizeMax$ steps). Hence, the reduction is based on accounting for each dimension the current negative sum of weights since the last reset, and the number of steps that remain to achieve a non-negative sum. This accounting is encoded in the states of $\gameCBFull$, as from the original state space $\states$, we go to the extended state space $\states \times (\{-\sizeMax\cdot\largestW, \ldots{}, 0\} \times \{1, \ldots{}, \sizeMax\})^{\dimension}$: states of $\gameCB$ are tuples representing a state of $\game$ and the current status of open windows in all dimensions (sum and remaining steps). 
We add states reached whenever a window reaches its maximum size $\sizeMax$ without closing. We label those as \textit{bad} states. We have one bad state for every state of $\game$.
Transitions in $\gameCB$ are built in order to accurately model the effect of transitions of $\game$ on open windows.
Clearly, a play is winning for the fixed window problem if and only if the corresponding play in $\gameCB$ is winning for the \coBuchi objective that asks that the set of bad states is not visited infinitely often, as that means that from some point on, all windows close in the required number of steps.

\begin{proof}
Let $\gameFull$ be a game with objective $\fixedWindowMPObjTW{\zeroVector}{\sizeMax \in \natStrict}$ and initial state $\initState \in \states$. Let $\largestW$ denote the maximal absolute value of any edge in $\edges$. We construct the unweighted game $\gameCBFull$ in the following way.
\begin{itemize}
\item $\statesCB_{1} = \left( \states_{1} \times \left( \{ -\largestW \cdot \sizeMax, \ldots{}, 0\} \times \{1, \ldots{}, \sizeMax\}\right)^{\dimension}\right)  \cup \{ \sink_{1}, \ldots{}, \sink_{\statesSize}\}$. States $\sink_{1}, \ldots{}, \sink_{\statesSize}$ denote special added \textit{bad sta\-tes}, one for each of the original states $\state_{1}, \ldots{}, \state_{\statesSize} \in \states$. The other states are built as tuples that represent (a) a visited state in $\game$, (b) for each dimension, a couple modeling (b.1) the current sum of weights since the last time the sum in this dimension was non-negative, and (b.2) the number of steps that remain to reach a non-negative sum in this dimension (i.e., before reaching the maximum window size).
\item $\statesCB_{2} = \states_{2} \times \left( \{ -\largestW \cdot \sizeMax, \ldots{}, 0\} \times \{1, \ldots{}, \sizeMax\}\right)^{\dimension} $.
\item We construct the edges $((\state_{a}, (\sumCB^{1}_{a}, \stepsCB^{1}_{a}), \ldots{}, (\sumCB^{\dimension}_{a}, \stepsCB^{\dimension}_{a})), (\state_{b}, (\sumCB^{1}_{b}, \stepsCB^{1}_{b}), \ldots{}, (\sumCB^{\dimension}_{b}, \stepsCB^{\dimension}_{b}))$ of $\edgesCB$ as follows. For all $(\state_{a}, \state_{b}) \in \edges$, let $\we = \weight((\state_{a}, \state_{b}))$, we have
\begin{itemize}
\item $((\state_{a}, (\sumCB^{1}_{a}, \stepsCB^{1}_{a}), \ldots{}, (\sumCB^{\dimension}_{a}, \stepsCB^{\dimension}_{a})), \sink_{b}) \in \edgesCB$, with $\sink_{b}$ the bad state associated to state $s_{b}$, iff $\exists\, t,\, 1 \leq t \leq \dimension$ such that $\stepsCB_{a}^{t} = 1$ and $\sumCB_{a}^{t} + \we(t) < 0$,
\item $((\state_{a}, (\sumCB^{1}_{a}, \stepsCB^{1}_{a}), \ldots{}, (\sumCB^{\dimension}_{a}, \stepsCB^{\dimension}_{a})), (\state_{b}, (\sumCB^{1}_{b}, \stepsCB^{1}_{b}), \ldots{}, (\sumCB^{\dimension}_{b}, \stepsCB^{\dimension}_{b})) \in \edgesCB$ iff $\forall\, t,\, 1 \leq t \leq \dimension$, we have
\begin{itemize}
\item if $\sumCB^{t}_{a} + \we(t) \geq 0$ then $\sumCB^{t}_{b} = 0, \stepsCB_{b}^{t} = \sizeMax$,
\item if $\sumCB^{t}_{a} + \we(t) < 0 \,\wedge\, \stepsCB_{a}^{t} > 1$ then $\sumCB^{t}_{b} = \sumCB^{t}_{a} + \we(t), \stepsCB_{b}^{t} = \stepsCB_{a}^{t} - 1$,
\end{itemize}
\end{itemize}
and we add edges $(\sink_{i}, (\state_{i}, (0, \sizeMax, \ldots{}, (0, \sizeMax))$ to $\edgesCB$ for all states $\state_{i} \in \states$.
\end{itemize}

Intuitively, the game $\gameCB$ is built by unfolding the game $\game$ and integrating the current sum of weights in the states of $\gameCB$, as well as the number of steps that remain to close a window, both for each dimension separately. The game~$\gameCB$ starts in the initial state $(\initState, (0, \sizeMax), \ldots{}, (0, \sizeMax))$, and each time a transition $(s, s')$ in the original game $\game$ is taken, the game $\gameCB$ is updated to a state $(s', (\sumCB^{1}, \stepsCB^{1}), \ldots{}, (\sumCB^{\dimension}, \stepsCB^{\dimension}))$ such that (a) if the current sum becomes positive in a dimension $t$, the corresponding sum counter is reset to zero and the step counter is reset to its maximum value, $\sizeMax$, (b) if the sum is still strictly negative in a dimension $t$ and the window for this dimension is not at its maximal size, the sum is updated and the step counter is decreased, and (c) if the sum stays strictly negative and the maximal size is reached in any dimension, the game visits the corresponding bad state and then, all counters are reset for all dimensions.

We argue that a play $\play$ in $\game$ is winning for the fixed window mean-payoff objective if and only if the corresponding play $\playCB$ in $\gameCB$ is winning for the \coBuchi objective asking not to visit the set $\states_{\sink} = \{ \sink_{1}, \ldots{}, \sink_{\statesSize}\}$ infinitely often. Indeed, consider a play $\play$ winning for objective $\fixedWindowMPObjTW{\zeroVector}{\sizeMax}$. By Eq.~\eqref{eq:fixedWindowObj}, this play only sees a finite number of bad windows (windows that are not closed in $\sizeMax$ steps in some dimension). By construction of $\gameCB$, the corresponding play $\playCB$ only visits the set $\states_{\sink}$ a finite number of times, hence it is winning for the \coBuchi objective. Now, let $\playCB$ be a winning play for the \coBuchi objective. By definition, there exists a position $i$ in $\playCB$ such that all states appearing after position $i$ belong to $\states \setminus \states_{\sink}$. It remains to prove that for any position $j \geq i$, for any dimension $t$, $1 \leq t \leq k$, there is a valid window of size at most $\sizeMax$. Again we use the inductive property of windows. We know by construction that a reset of the sum happens in at most $\sizeMax$ steps, otherwise we go to a bad state. Assume $j$ is a position with a sum counter of zero in some dimension $t$, and $j'$ is the next such position. Since resets are done \textit{as soon as} the sum becomes non-negative, all suffixes of the sequence from $j$ to $j'$ are non-negative. Hence, it is clear that for all position~$j''$, $j < j'' < j'$, the window from $j''$ to $j'$ in dimension $t$ is closed. Consequently, the corresponding play $\play$ in~$\game$ is winning for the fixed window mean-payoff objective of threshold $0$ and window size $\sizeMax$.
\end{proof}

As a direct corollary of this reduction, we obtain an EXPTIME algorithm to solve the fixed window mean-payoff problem on multi-dimension games, as solving \coBuchi games takes quadratic time in the size of the game~\cite{DBLP:journals/jacm/ChatterjeeH14}.

\begin{corollary}
\label{cor:multiDimEXPTIMEeasy}
Given a two-player multi-dimension game $\gameFull$ and a window size $\sizeMax \in \natStrict$, the fixed window mean-payoff problem can be solved in time $\mathcal{O}(\statesSize^{2} \cdot (\sizeMax)^{4\cdot \dimension} \cdot \largestW^{2\cdot \dimension})$ via a reduction to \coBuchi games.
\end{corollary}

\begin{proof}
Lemma~\ref{lem:multiDimReducCB} uses a \coBuchi game whose state space is of size 
\begin{equation*}
\Big\vert \states \times \big( \{ -\largestW \cdot \sizeMax, \ldots{}, 0\} \times \{1, \ldots{}, \sizeMax\}\big)^{\dimension}\Big\vert + \statesSize = \mathcal{O}\Big(\statesSize \cdot (\sizeMax)^{2\cdot \dimension} \cdot \largestW^{\dimension}\Big).
\end{equation*}
The quadratic algorithm for \coBuchi games described in~\cite{DBLP:journals/jacm/ChatterjeeH14} implies the result.
\end{proof}

A natural question is whether a distinct algorithm is useful in the one-dimension case. Remark~\ref{rem:multiDimOneDimAlgo} notes that it is.

\begin{remark}
\label{rem:multiDimOneDimAlgo}
The multi-dimension algorithm described in Corollary~\ref{cor:multiDimEXPTIMEeasy} yields a procedure which is polynomial in the size of the state space, the window size, and the largest weight for the subclass of one-dimension games, hence only \textit{pseudo-polynomial} (i.e., exponential in $\bits$, the length of the encoding of weights), whereas Lemma~\ref{lem:windowMPAlg} gives a truly polynomial algorithm. 
\end{remark}

\smallskip\noindent\textbf{Fixed window: lower bounds.} We first consider the fixed \textit{arbitrary} window mean-payoff problem for which we show (i) in Lemma~\ref{lem:multiDimFixedHardMembership}, EXPTIME-hardness for $\{-1, 0, 1\}$ weights and arbitrary dimensions via a reduction from the \textit{membership problem for alternating polynomial-space Turing machines (APTMs)}~\cite{chandra_JACM1981}, and (ii) in Lemma~\ref{lem:multiDimFixedHardCountdown}, EXPTIME-hardness for two dimensions and arbitrary weights via a reduction from \textit{countdown games}~\cite{jurdzinski_LMCS2008}.

Given an APTM $\aptm$ and a word $\word \in \{0, 1\}^{\ast}$, such that the tape contains at most $p(\wordSize)$ cells, where $p$ is a polynomial function, the membership problem asks to decide if $\aptm$ accepts $\word$. We build a fixed arbitrary window mean-payoff game~$\game$ so that $\playerOne$ has to simulate the run of $\aptm$ on $\word$, and $\playerOne$ has a winning strategy in $\game$ if and only if the word is accepted by the machine. For each tape cell $\tapeCell \in \{1, 2, \ldots{}, p(\wordSize)\}$, we have two dimensions, $(\tapeCell, 0)$ and $(\tapeCell, 1)$ such that a sum of weights of value $-1$ (i.e., an open window) in dimension $(\tapeCell, i)$, $i \in \{0, 1\}$ encodes that in the current configuration of $\aptm$, tape cell $\tapeCell$ contains a bit of value $i$. In each step of the simulation (Fig.~\ref{fig:multiDimFixedMembership}), $\playerOne$ has to disclose the symbol under the tape head: if in position $\tapeCell$, $\playerOne$ discloses a $0$ (resp. a $1$), he obtains a reward $1$ in dimension $(\tapeCell, 0)$ (resp. $(\tapeCell, 1)$). To ensure that $\playerOne$ was faithful, $\playerTwo$ is then given the choice to either let the simulation continue, or assign a reward $1$ in all dimensions except $(\tapeCell, 0)$ and $(\tapeCell, 1)$ and then restart the game after looping in a zero self-loop for an arbitrary long time. If $\playerOne$ cheats by not disclosing the correct symbol under tape cell $\tapeCell$, $\playerTwo$ can punish him by branching to the restart state and ensuring a sufficiently long open window in the corresponding dimension before restarting (as in Fig.~\ref{fig:wmpEx2}). But if $\playerOne$ discloses the correct symbol and $\playerTwo$ still branches, all windows close. In the accepting state, all windows are closed and the game is restarted. The window size $\sizeMax$ of the game is function of the existing bound on the length of an accepting run. To force $\playerOne$ to go to the accepting state, we add an additional dimension, with weight~$-1$ on the initial edge of the game and weight $1$ on reaching the accepting state.

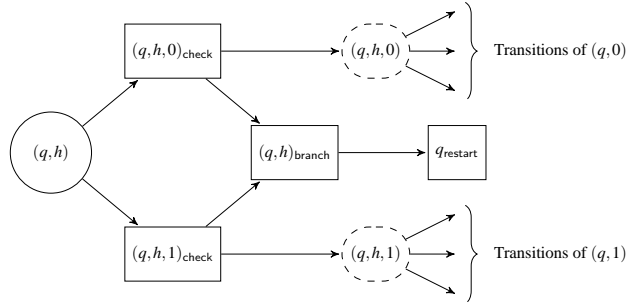
\begin{figure}[thb]
  \centering   
  \scalebox{0.9}{\begin{tikzpicture}[->,>=stealth',shorten >=1pt,auto,node
    distance=2.5cm,bend angle=45,scale=0.6, font=\scriptsize]
    \tikzstyle{p1}=[draw,circle,text centered,minimum size=12mm]
    \tikzstyle{p2}=[draw,rectangle,text centered,minimum size=8mm]
    \tikzstyle{p3}=[draw, rounded rectangle, dashed, text centered,minimum size=8mm]
    \node[p1]  (0)  at (0, 0) {$(q, \tapeCell)$};
    \node[p2]  (1) at (3, 2.5) {$(q, \tapeCell, 0)_{{\sf check}}$};
    \node[p2]  (2) at (3, -2.5)  {$(q, \tapeCell, 1)_{{\sf check}}$};
    \node[p2]  (3) at (6, 0)  {$(q, \tapeCell)_{{\sf branch}}$};
    \node[p2]  (4)  at (10, 0) {$q_{{\sf restart}}$};
    \node[p3]  (5)  at (8, 2.5) {$(q, \tapeCell, 0)$};
    \node[p3]  (6)  at (8, -2.5) {$(q, \tapeCell, 1)$};
    \path
    (0) edge (1)
    (0) edge (2)
    (1) edge (3)
    (1) edge (5)
    (2) edge (3)
    (2) edge (6)
    (3) edge (4)
    (5) edge (10, 2.5)
    (5) edge (10, 3.5)
    (5) edge (10, 1.5)
    (6) edge (10, -2.5)
    (6) edge (10, -3.5)
    (6) edge (10, -1.5);
    \draw[-,decorate,decoration={brace,amplitude=5pt},xshift=-3mm,yshift=0pt] (10.5,3.7) -- (10.5,1.3) node [black,midway,xshift=9pt] {Transitions of $(q, 0)$};
    \draw[-,decorate,decoration={brace,amplitude=5pt},xshift=-3mm,yshift=0pt] (10.5,-1.3) -- (10.5,-3.7) node [black,midway,xshift=9pt] {Transitions of $(q, 1)$};
      \end{tikzpicture}}
      \caption{Gadget ensuring a correct simulation of the APTM on tape cell $\tapeCell$.}
\label{fig:multiDimFixedMembership}
  \end{figure}

\begin{lemma}
\label{lem:multiDimFixedHardMembership}
The fixed arbitrary window mean-payoff problem is EXPTIME-hard in multi-dimension games with $\{-1, 0, 1\}$ weights and arbitrary dimensions.
\end{lemma}

\begin{proof}
An \textit{alternating Turing machine} (ATM)~\cite{chandra_JACM1981} is a tuple $\aptm = (Q, q_{0}, \alphabet_{{\sf in}}, \delta, q_{{\sf acc}})$ where:
\begin{itemize}
\item $Q$ is the finite set of control states with a partition $(Q_{\vee}, Q_{\wedge})$ of $Q$ into existential and universal states;
\item $q_{0} \in Q$ is the initial state;
\item $\alphabet_{{\sf in}} = \{0, 1\}$ is the input alphabet and $\alphabet_{{\sf tape}} = \alphabet_{{\sf in}} \cup \{\#\}$ the tape alphabet, with $\#$ the blank symbol;
\item $\delta \subseteq Q \times \alphabet_{{\sf tape}} \times Q \times \alphabet_{{\sf tape}} \times \{-1, 1\}$ is a transition relation;
\item there is a special accepting state $q_{{\sf acc}} \in Q_{\vee}$ (without loss of generality).
\end{itemize}
We say that $\aptm$ is a \textit{polynomial-space} alternating Turing machine (APTM) if for some polynomial function $p$, the space used by $\aptm$ on any input word $\word \in \alphabet_{{\sf in}}^{\ast}$ is bounded by $p(\wordSize)$.

We define the AND-OR graph of the APTM $(\aptm, p)$ on the input word $\word \in \alphabet_{{\sf in}}^{\ast}$ as $\mathcal{G}(\aptm, p) = \left\langle S_{\vee}, S_{\wedge}, s_{0}, \Delta, R \right\rangle $ where
\begin{itemize}
\item $S_{\vee} = \{ (q, h, t) \,\vert\, q \in Q_{\vee},\, 1 \leq h \leq p(\wordSize)$ and $t \in \alphabet_{{\sf tape}}^{p(\wordSize)}\}$; 
\item $S_{\wedge} = \{ (q, h, t) \,\vert\, q \in Q_{\wedge},\, 1 \leq h \leq p(\wordSize)$ and $t \in \alphabet_{{\sf tape}}^{p(\wordSize)}\}$;
\item $s_{0} = (q_{0}, 1, t)$ where $t = \word \cdot \#^{p(\wordSize) - \wordSize}$;
\item $((q_{1}, h_{1}, t_{1}), (q_{2}, h_{2}, t_{2})) \in \Delta$ iff there exists $(q_{1}, t_{1}(h_{1}), q, \gamma, d) \in \delta$ such that $q_{2} = q$, $h_{2} = h_{1} + d$, $t_{2}(h_{1}) = \gamma$ and $t_{2}(\tapeCell) = t_{1}(\tapeCell)$ for all $\tapeCell \neq h_{1}$;
\item $R = \{(q, h, t) \in S_{\vee} \,\vert\, q = q_{{\sf acc}}\}$.
\end{itemize}
Intuitively, states of the graph correspond to configurations $(q, h, t)$ where $q$ is a control state of the machine, $h$ the position of the tape head, and $t$ the current word written on the tape. Given a state $q$ of the machine $\aptm$, tape head on cell $\tapeCell$ and a word $t$ on the tape, a transition from $(q, \tapeCell, t)$ to $(q', \tapeCell', t')$ exists in the graph $\mathcal{G}(\aptm, p)$ if the transition relation $\delta$ of the machine $\aptm$ admits a transition that given this configuration, updates the content of cell $\tapeCell$ to the symbol $t'(\tapeCell)$, such that the tape now contains the word $t'$, and then goes to control state $q'$ and moves the tape head to an adjacent cell $\tapeCell'$.

A word $\word \in \alphabet_{{\sf in}}^{\ast}$ is \textit{accepted} by an APTM $(\aptm, p)$ if there exists a run tree (obtained by choosing a child in existential nodes and keeping all children in universal nodes) of $\aptm$ on $\word$ such that all leafs are accepting configurations. That is, a word is accepted if and only if, in the two-player game defined by $\mathcal{G}(\aptm, p)$, player $\player_{\vee}$ has a strategy to reach the set of accepting states $R$. Deciding the acceptance of a word by an APTM is an EXPTIME-complete problem, known as the membership problem~\cite{chandra_JACM1981}.

We construct a fixed window mean-payoff game $\gameFull$ simulating the machine $(\aptm, p)$ as follows. Let $\dimension = 2\cdot p(\wordSize) + 1$: there is a dimension for each pair $(\tapeCell, 0)$ and $(\tapeCell, 1)$, for all $1 \leq \tapeCell \leq p(\wordSize)$, and one additional dimension. The set of states $\states$ of the game is
\begin{align*}
S = &\{q_{{\sf restart}}\} \cup \{q_{{\sf in}}\} \cup \{\widehat{q_{{\sf acc}}}\} \cup \{(q, \tapeCell) \,\vert\, q \in Q,\, 1 \leq \tapeCell \leq p(\wordSize)\} \cup \{(q, \tapeCell, i)_{{\sf check}} \,\vert\, q \in Q,\, 1 \leq \tapeCell \leq p(\wordSize), i \in \{0, 1\}\}\\ &\cup \{(q, \tapeCell)_{{\sf branch}} \,\vert\, q \in Q,\, 1 \leq \tapeCell \leq p(\wordSize)\} \cup \{(q, \tapeCell, i) \,\vert\, q \in Q,\, 1 \leq \tapeCell \leq p(\wordSize), i \in \{0, 1\}\}.
\end{align*}
States of the form $(q, \tapeCell)$ belong to $\playerOne$. States of the form $(q, \tapeCell, i)$ belong to $\playerOne$ if $q \in Q_{\vee}$ in the machine $\aptm$. All other states belong to $\playerTwo$. The initial state is $q_{{\sf restart}}$. It has two outgoing edges with weights zero in all dimensions: one self-loop, and one edge to $q_{{\sf in}}$. The latter is assigned the following weights: $-1$ for dimension $(\tapeCell, i)$ if the letter at position $\tapeCell$ of $\word$ is $i$, $-1$ in the very last dimension ($2\cdot p(\wordSize) + 1$), and zero everywhere else. From $q_{{\sf in}}$, the game goes to $(q_{0}, 1)$ and the simulation of $\aptm$ begins.
  
The game mimics runs of $\aptm$, and it is ensured that if the current state of the game is $(q, \tapeCell)$ and the cell content is~$i$, then the sum of weights since the last visit of $q_{{\sf in}}$ in dimension $(\tapeCell, i)$ is $-1$. We refer to the segment of play since the last visit of $q_{{\sf in}}$ as the \textit{current round}. We depict a step of the simulation in Fig.~\ref{fig:multiDimFixedMembership}. At state $(q, \tapeCell)$, $\playerOne$ has the choice between states $(q, \tapeCell, 0)_{{\sf check}}$ and $(q, \tapeCell, 1)_{{\sf check}}$, resp. corresponding to declaring a content $0$ or $1$ of the tape cell $\tapeCell$. The reward for dimension $(\tapeCell, i)$, $i \in \{0, 1\}$ is $1$ on state $(q, \tapeCell, i)_{{\sf check}}$. At state $(q, \tapeCell, i)_{{\sf check}}$, a state of $\playerTwo$, $\playerTwo$ checks whether $\playerOne$ has correctly revealed the tape content as follows: (i) Player $\playerTwo$ can choose to go to state $(q, \tapeCell)_{{\sf branch}}$, in which all dimensions other than $(\tapeCell, 0)$ and $(\tapeCell, 1)$, including the very last, are increased by $1$, and then go to $q_{{\sf restart}}$ on which $\playerTwo$ will be able to delay the play; (ii) Player $\playerTwo$ can choose to proceed and continue the simulation: the game then goes to state $(q, \tapeCell, i)$. State $(q, \tapeCell, i)$ is either a state of $\playerOne$ or $\playerTwo$, depending on the affiliation of state $q$ in the APTM. Such a gadget ensures that if $\playerOne$ cheats by not disclosing the correct symbol, $\playerTwo$ can force an open window of arbitrary length in the current round by looping on $q_{{\sf restart}}$ for some time, and then restarting the game. On the other hand, if $\playerOne$ is faithful and $\playerTwo$ still decides to branch to $(q, \tapeCell)_{{\sf branch}}$, then all windows will be closed for the current round.

If $\playerOne$ does not cheat and $\playerTwo$ acknowledges it by not branching, the game advances to a state of the form $(q, \tapeCell, i)$. At such a state, we add transitions as follows: if there exists a transition from $(q, \tapeCell, i)$ to $(q', \tapeCell', i')$ in $\aptm$, then we add an edge from $(q, \tapeCell, i)$ to $(q', \tapeCell')$ in the game $\game$, and assign weight $-1$ in dimension $(\tapeCell, i')$, as the tape cell at position $\tapeCell$ contains $i'$ and we ensure that the sum in dimension $(\tapeCell, i')$ in the current round is $-1$. At the accepting states $(q_{{\sf acc}}, \tapeCell)$, all dimensions are assigned reward $1$, and the next state is $\widehat{q_{{\sf acc}}}$. State $\widehat{q_{{\sf acc}}}$ is followed by $q_{{\sf restart}}$. Again there is no risk in looping as all dimensions are now non-negative.

Formally, blank symbols need to be added. For brevity and simplicity of the presentation, we omit these technical details.

We fix the window size $\sizeMax$ equal to three times the size of the configuration graph (bound on the length of a run) plus three, and we argue that the game $\game$ is a faithful simulation of the machine $\aptm$, that is, $\playerOne$ wins the fixed window mean-payoff game if and only if the word $\word$ is accepted by $\aptm$. Notice that the construction ensures that if $\playerOne$ cheats in the current round, $\playerTwo$ can make this round losing, as discussed before. Similarly, if $\playerOne$ does not cheat but does not reach the accepting state, dimension $2\cdot p(\wordSize) + 1$ will remain negative when arriving in $q_{{\sf restart}}$ and $\playerTwo$ will be able to cycle long enough to make the round losing as the window in the last dimension will remain open for $\sizeMax$ steps. Clearly, $\playerOne$ cannot see losing rounds infinitely often otherwise the play is losing. Assume the word $\word$ is accepted by the machine. Then there is an accepting run tree, and the winning strategy of $\playerOne$ is to follow this run tree and always reveal the correct symbol. This way, either $\playerTwo$ restarts and the round is winning because all dimensions are non-negative, or $\playerTwo$ does not restart and an accepting state $(q_{{\sf acc}}, \tapeCell)$ is reached within the maximum allowed window size. Indeed, in the APTM, there is a strategy to reach the accepting state in a number of steps bounded by the size of the configuration graph. In that case, the round is also winning. Conversely, assume that the word $\word$ is not accepted by the APTM. Consider any strategy $\strat_{1}$ of $\playerOne$. Clearly, $\playerOne$ cannot cheat as otherwise, he loses. So assume he does not cheat. Then there is a path in the run tree obtained from playing the strategy $\strat_{1}$ in $\aptm$ such that the path never reaches an accepting state. Hence, the strategy $\strat_{2}$ of $\playerTwo$ that follows this path in the game $\game$ ensures that the sum in dimension $2\cdot p(\wordSize) + 1$ is always strictly negative, and after waiting till the bound $\sizeMax$ on the window size is met, $\playerTwo$ has made the round losing and he can restart the game safely. Acting this way infinitely often, $\playerTwo$ can violate the fixed window objective for $\playerOne$. It follows that $\playerOne$ wins in $\game$ if and only if the word $\word$ is accepted by the APTM $\aptm$.
\end{proof}

We now prove EXPTIME-hardness for two dimensions and arbitrary weights via a reduction from countdown games. A countdown game $\cdGame$ consists of a weighted graph $(\cdStates, \cdTransitions)$, with $\cdStates$ the set of states and $\cdTransitions \subseteq \cdStates \times \natStrict \times \cdStates$ the transition relation. Configurations are of the form $(s, c)$, $s \in \cdStates$, $c \in \nat$. The game starts in an initial configuration $(\initState, c_{0})$ and transitions from a configuration $(s, c)$ are performed as follows: first $\playerOne$ chooses a duration $d$, $0 < d \leq c$ such that there exists $t = (s, d, s') \in \cdTransitions$ for some $s' \in \cdStates$, second $\playerTwo$ chooses a state $s' \in \cdStates$ such that $t = (s, d, s') \in \cdTransitions$. Then, the game advances to $(s', c-d)$. Terminal configurations are reached whenever no legitimate move is available. If such a configuration is of the form $(s, 0)$, $\playerOne$ wins the play. Otherwise, $\playerTwo$ wins the play. Deciding the winner in countdown games given an initial configuration $(\initState, c_{0})$ is EXPTIME-complete~\cite{jurdzinski_LMCS2008}.

Given a countdown game $\cdGame$ and an initial configuration $(\initState, c_{0})$, we create a game $\gameFull$ with $\dimension = 2$ and a fixed window objective for $\sizeMax = 2\cdot c_{0} + 2$.
The two dimensions are used to store the value of the countdown counter and its opposite. Each time a duration $d$ is chosen, an edge of value $(-d, d)$ is taken. The game simulates the moves available in $\cdGame$: a strict alternation between states of $\playerOne$ (representing states of $\cdStates$) and states of $\playerTwo$ (representing transitions available from a state of $\cdStates$ once a duration has been chosen). On states of $\playerOne$, we add the possibility to branch to a state $s_{{\sf restart}}$ of $\playerTwo$, in which $\playerTwo$ can either take a zero cycle, or go back to the initial state and force a restart of the game. By placing weights $(0, -c_{0})$ on the initial edge, and $(c_{0}, 0)$ on the edge branching to $s_{{\sf restart}}$, we ensure that the only way to win for $\playerOne$ is to accumulate a value exactly equal to $c_{0}$ in the game before switching to $s_{{\sf restart}}$. This is possible if and only if $\playerOne$ can reach a configuration of value zero in $\cdGame$.

\begin{lemma}
\label{lem:multiDimFixedHardCountdown}
The fixed arbitrary window mean-payoff problem is EXPTIME-hard in multi-dimension games with two dimensions and arbitrary weights.
\end{lemma}

\begin{proof}
We establish a polynomial-time reduction from the countdown game problem to the fixed arbitrary window problem. Let $\cdGame = (\cdStates, \cdTransitions)$ be a countdown game~\cite{jurdzinski_LMCS2008}, with initial configuration $(\initState, c_{0})$. We create a corresponding game $\gameFull$ as follows.
\begin{itemize}
\item $\states_{1} = \cdStates$.
\item Let $\states^{\cdTransitions} \subseteq \cdStates \times \natStrict$ be the subset of pairs $(s, d)$ such that there exists a transition $(s, d, s') \in \cdTransitions$. Then, $\states_{2} = \states^{\cdTransitions} \cup \{s_{{\sf restart}}\}$. State $s_{{\sf restart}}$ is the initial state of game $\game$.
\item For each transition $(s, d, s') \in \cdTransitions$, we add edges $(s, (s,d))$, with $s \in \states_{1}$ and $(s,d) \in \states_{2}$, and $((s,d), s')$, with $s' \in \states_{1}$, to the set of edges $\edges$. Edge $(s, (s,d))$ has weight $(-d, d)$ and edge $((s, d), s')$ has weight $(0, 0)$.
\item For all $s \in \states_{1}$, we add an edge $(s, s_{{\sf restart}})$ of weight $(c_{0}, 0)$.
\item From $s_{{\sf restart}}$, we add an edge $(s_{{\sf restart}}, \initState)$ of value $(0, -c_{0})$.
\item On $s_{{\sf restart}}$, we add a self-loop $(s_{{\sf restart}}, s_{{\sf restart}})$ of weight $(0, 0)$.
\end{itemize}

We fix the window size $\sizeMax = 2\cdot c_{0} + 2$, and we claim that $\playerOne$ wins the fixed window problem if and only if he wins the countdown game. Recall that to win a countdown game, $\playerOne$ must be able to reach a configuration $(s, 0)$ in the game~$\cdGame$. The key idea to our construction is that in the game $\game$, the only way to avoid seeing infinitely often open windows of size larger than $\sizeMax$ is to accumulate exactly $c_{0}$ before restarting, which is equivalent to reaching a configuration of value $0$ in $\cdGame$.

Notice that the game $\game$ starts by visiting an edge of value $(0, -c_{0})$ and afterwards, all edges from states of $\playerOne$ have a value $(-d, d)$ corresponding to the duration he chooses in the countdown game. All except the edge he can decide to take to go to $s_{{\sf restart}}$, which value is $(c_{0}, 0)$. Clearly, if $\playerOne$ decides to go in $s_{{\sf restart}}$, he has to close all windows, as otherwise $\playerTwo$ can use the self-loop to delay the play long enough and provoke a sufficiently long bad window, which if done repeatedly, induces a losing play. On the other hand, if $\playerOne$ decides to never go towards $s_{{\sf restart}}$, he will keep accumulating negative values in the first dimension and he is guaranteed to lose. So obviously the behavior of $\playerOne$ should be to play as in the countdown game to accumulate exactly $c_{0}$ in dimension $2$ (and $-c_{0}$ in dimension $1$) before switching to $s_{{\sf restart}}$, so that $\playerTwo$ can do no harm by delaying the play as all windows will be closed. The accumulated value has to be \textit{exactly} $c_{0}$ as (a) if it is less than $c_{0}$, dimension $2$ will remain negative, and (b) if it is more than $c_{0}$, dimension $1$ will stay negative (i.e., the edge $(s, s_{{\sf restart}})$ will not suffice to get it back above zero). Since the minimal increase is of $1$ every two edges by construction, the allowed window size $\sizeMax$ is sufficient to enforce such a behavior, if possible. This shows that $\playerOne$ wins the fixed window problem from initial state $s_{{\sf restart}}$ in $\game$ if and only if he also wins the countdown game $\cdGame$ from $(\initState, c_{0})$, as accumulating $c_{0}$ in $\game$ is equivalent to reaching a configuration of value zero in $\cdGame$.
\end{proof}

For the case of polynomial windows, Lemma~\ref{lem:multiDimGenReachReduc} proves PSPACE-hardness via a reduction from generalized reachability games~\cite{fijalkow_CORR2010}. Filling the gap with the EXPTIME-membership given by Corollary~\ref{cor:multiDimEXPTIMEeasy} is an open problem. The generalized reachability objective is a conjunction of reachability objectives: a winning play has to visit a state of each of a series of $\dimension$ reachability sets. If $\playerOne$ has a winning strategy in a generalized reachability game $\game^{r} = (\states_{1}^{r}, \states_{2}^{r}, \edges^{r})$, then he has one that guarantees visit of all sets within $\dimension \cdot \statesSizeP{r}$ steps. We create a modified weighted version of the game, $\gameFull$, such that the weights are $\dimension$-dimension vectors. The game starts by opening a window in all dimensions and the only way for $\playerOne$ to close the window in dimension $t$, $1 \leq t \leq \dimension$ is to reach a state of the $t$-th reachability set. We modify the game by giving $\playerTwo$ the ability to close all open windows and restart the game such that the prefix-independence of the fixed window objective cannot help $\playerOne$ to win without reaching the target sets. Then, a play is winning in $\game$ for the fixed window objective of size $\sizeMax = 2\cdot\dimension\cdot\statesSizeP{r}$ if and only if it is winning for the generalized reachability objective in $\game^{r}$.

\begin{lemma}
\label{lem:multiDimGenReachReduc}
The fixed polynomial window mean-payoff problem is PSPACE-hard.
\end{lemma}

\begin{proof}
We show the PSPACE-hardness by a reduction from the generalized reachability problem~\cite{fijalkow_CORR2010}. Given a game graph $\game^{r} = (\states_{1}^{r}, \states_{2}^{r}, \edges^{r})$, a series of reachability sets $R_{t} \subseteq \states^{r}$, for $1 \leq t \leq \dimension$, with $\dimension \leq \vert\states^{r}\vert$, and an initial state $\initState^{r} \in \states^{r}$, the generalized reachability problem asks if there exists a strategy of $\playerOne$ such that any consistent outcome starting in $\initState^{r}$ visits a state of each set $R_{t}$ at least once. It is known that if such a strategy exists, then there exists one which ensures reaching all sets in at most $\dimension \cdot \vert\states^{r}\vert$ steps.

We build a $\dimension$-dimension fixed window mean-payoff game $\gameFull$ as follows. 
We define $\states_{{\sf branch}} \subset \states_{2}$, a set of additional states belonging to $\playerTwo$ and of the form $b_{s, s'}$, one for each $(s,s') \in \edges^{r}$. Let $\states_{1} = \states_{1}^{r}$ and $\states_{2} = \states_{2}^{r} \cup \states_{{\sf branch}} \cup \{s_{{\sf restart}}\}$. Let $\edges$ be the set of edges such that for all $(s, s') \in \edges^{r}$, we have that $(s, b_{s,s'}) \in \edges$, $(b_{s,s'}, s') \in \edges$, $(b_{s,s'}, s_{{\sf restart}}) \in \edges$, and such that $(s_{{\sf restart}}, \initState^{r}) \in \edges$.
That is, we introduce in all edges of $\edges^{r}$ a state of $\playerTwo$ that let him branch to an added state $s_{\text{restart}}$ or continue as in $\game^{r}$. The new initial state in $\game$ is $s_{\text{restart}}$, and there is an edge from $s_{\text{restart}}$ to the old initial state $\initState^{r}$. The weights are as follows: all edges from states $b_{s,s'}$ to $s_{\text{restart}}$ have value $1$ in all dimensions. The edge from $s_{\text{restart}}$ to $\initState^{r}$ has value $-1$ in all dimensions. All other edges of the game have value zero, except edges entering a state that belongs to a reachability set $R_{t}$, which have value $1$ in dimension $t$ and $0$ in the other dimensions. If a state belongs to several sets, then all corresponding dimensions get a $1$.

We claim that $\playerOne$ has a winning strategy for $\fixedWindowMPObjTW{\zeroVector}{\sizeMax = 2 \cdot \dimension \cdot \vert\states^{r}\vert}$ if and only if he has a winning strategy for the generalized reachability objective in $\game^{r}$. Consider the game $\game$. The only edge involving negative values is $(s_{\text{restart}}, \initState^{r})$, which value is $(-1, \ldots{}, -1)$. Therefore, a losing play for Eq.~\eqref{eq:fixedWindowObj} should see this edge infinitely often, as it is the starting position of all open windows. Notice that on the other hand, going from a state $b_{s,s'}$ to $s_{\text{restart}}$ involves an edge of value $(1, \ldots{}, 1)$, hence if the open window starting in $s_{\text{restart}}$ comes back in $s_{\text{restart}}$ before hitting its maximal size, the window will close. So the strategy of $\playerTwo$ should be to wait for $\sizeMax = 2 \cdot \dimension \cdot \vert\states^{r}\vert$ steps before forcing a restart. Now, consider a winning strategy $\strat_{1}$ of $\playerOne$ in $\game$. Because of the strategy of $\playerTwo$, $\strat_{1}$ has to ensure obtaining~$+1$ in all dimensions by only using transitions entering in states of $\states^{r}$. By construction, this implies that $\strat_{1}$ enforces a visit of all reachability sets, and thus is winning for the generalized reachability problem. Consider the converse implication. Let $\strat_{1}^{r}$ be a winning strategy in $\game^{r}$. There exists such a strategy that ensures seeing all reachability sets (thus closing all windows) in at most $\sizeMax = 2 \cdot \dimension \cdot \vert\states^{r}\vert$ steps if $\playerTwo$ does not branch to $s_{\text{restart}}$. On the other hand, if $\playerTwo$ does branch before $\sizeMax$ steps, all windows also close, as branching edges have value $(1, \ldots{}, 1)$. Hence, this strategy is also winning for $\fixedWindowMPObjT{\zeroVector}$. This shows the correctness of the reduction and concludes our proof.
\end{proof}

We conclude our study of the multi-dimension fixed window problem by considering memory bounds. A direct corollary of Lemma~\ref{lem:multiDimReducCB} is the existence of winning strategies of at most exponential size for both players, as memoryless strategies are sufficient in \coBuchi games~\cite{emerson_FOCS1991}. A corollary of the reduction from generalized reachability games to the fixed polynomial window problem used to prove Lemma~\ref{lem:multiDimGenReachReduc} and the results of~\cite[Lemma 2]{fijalkow_CORR2010} (showing exponential lower bounds on memory for generalized reachability objectives) is that such memory is needed in general, again for both players.

Another example of a family of games in which $\playerOne$ requires exponential memory (in the number of dimensions) is given by the family defined in~\cite[Lemma 8]{DBLP:journals/acta/ChatterjeeRR14} (Fig.~\ref{fig:multiDimExpFamily}), introduced in the context of multi energy games.

\begin{example}
We define a family of games $(\game(\gadgets))_{\gadgets \geq 1}$ which is an assembly of $\dimension = 2\cdot\gadgets$ gadgets, the first $\gadgets$ belonging to $\playerTwo$, and the remaining $\gadgets$ belonging to $\playerOne$ (Fig.~\ref{fig:multiDimExpFamily}). Precisely, we have $\vert\states_{1}\vert = \vert\states_{2}\vert = 3\cdot\gadgets$, $\vert\states\vert = \vert\edges\vert = 6\cdot\gadgets = 3\cdot \dimension$ (linear in $\dimension$), $\dimension = 2\cdot\gadgets$, and $\weight$ defined as:
\begin{align*}
\forall\, 1 \leq i \leq \gadgets,\, &\weight((\circ, s_{i})) = \weight((\circ, t_{i})) = (0, \ldots{}, 0),\\
&\weight((s_{i}, s_{i, L})) = - \weight((s_{i}, s_{i, R})) = \weight((t_{i}, t_{i, L})) = - \weight((t_{i}, t_{i, R})),\\
&\forall\, 1 \leq j \leq \dimension,\, \weight((s_{i}, s_{i, L}))(j) = \begin{cases}1 \text{ if } j = 2\cdot i - 1\\-1 \text{ if } j = 2\cdot i\\0 \text{ otherwise}\end{cases},
\end{align*}
where $\circ$ denotes any valid predecessor state.

Essentially, in each state $s_{i}$, $\playerTwo$ can open a window on either dimension $2\cdot i -1$ or dimension $2 \cdot i$ by choosing the corresponding edge. In this game, $\playerOne$ wins objective $\fixedWindowMPObjTW{\zeroVector}{\sizeMax = \statesSize/2}$ only if he is able to make in $t_{i}$ the opposite choice of $\playerTwo$ in $s_{i}$, as this ensures closure of the corresponding window. This requires a strategy encoded as a Moore machine with at least $2^{\dimension / 2}$ states. Indeed, if $\playerOne$ cannot differentiate between the exponential number of histories from $s_{i}$ up to $t_{i}$, he is not able to enforce closure of the needed windows.
\end{example}

\begin{figure}[tb]
  \centering   
  \scalebox{0.8}{\begin{tikzpicture}[->,>=stealth',shorten >=1pt,auto,node
    distance=2.5cm,bend angle=45,scale=0.4, font=\small]
    \tikzstyle{p1}=[draw,circle,text centered,minimum size=8mm]
    \tikzstyle{p2}=[draw,rectangle,text centered,minimum size=8mm]
    \node[p2]  (0)  at (0, 0) {$s_{1}$};
    \node[p2]  (1) at (4, 2) {$s_{1,L}$};
    \node[p2]  (2) at (4, -2)  {$s_{1,R}$};
    \node[p2]  (3) at (8, 0)  {$s_{\gadgets}$};
    \node[p2]  (4)  at (12, 2) {$s_{\gadgets,L}$};
    \node[p2]  (5)  at (12, -2) {$s_{\gadgets,R}$};
    \node[p1]  (6)  at (16, 0) {$t_{1}$};
    \node[p1]  (7) at (20, 2) {$t_{1,L}$};
    \node[p1]  (8) at (20, -2)  {$t_{1,R}$};
    \node[p1]  (9) at (24, 0)  {$t_{\gadgets}$};
    \node[p1]  (10)  at (28, 2) {$t_{\gadgets,L}$};
    \node[p1]  (11)  at (28, -2) {$t_{\gadgets,R}$};
    \coordinate[shift={(-5mm,0mm)}] (init) at (0.west);
    \path
    (init) edge (0);
	\draw[->,>=latex] (0) to (1);
	\draw[->,>=latex] (0) to (2);
	\draw[->,>=latex] (3) to (4);
	\draw[->,>=latex] (3) to (5);
	\draw[dotted,->,>=latex] (1) to (3);
	\draw[dotted,->,>=latex] (2) to (3);
	\draw[->,>=latex] (4) to (6);
	\draw[->,>=latex] (5) to (6);
	\draw[->,>=latex] (6) to (7);
	\draw[->,>=latex] (6) to (8);
	\draw[dotted,->,>=latex] (7) to (9);
	\draw[dotted,->,>=latex] (8) to (9);
	\draw[->,>=latex] (9) to (10);
	\draw[->,>=latex] (9) to (11);
	\draw[->,>=latex] (10) to[out=160,in=60] (0);
	\draw[->,>=latex] (11) to[out=200,in=300] (0);
      \end{tikzpicture}}
\vspace*{-8mm}
      \caption[Family of games requiring exponential memory]{Family of multi-dimension games requiring exponential memory for $\playerOne$, for the fixed window objective.}
      \label{fig:multiDimExpFamily}
  \end{figure}
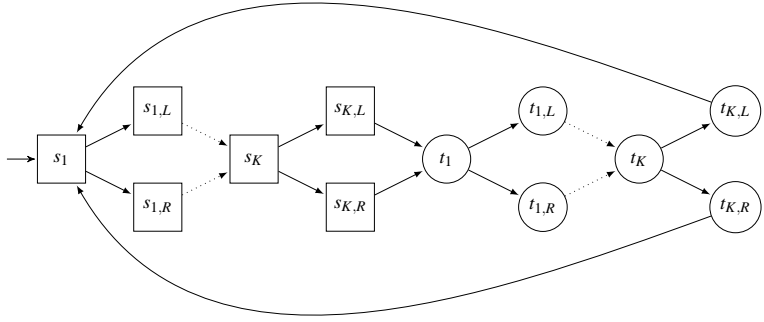

\begin{lemma}
\label{lem:multiDimFixedMemory}
In multi-dimension games with a fixed window mean-payoff objective, exponential memory is both sufficient and necessary for both players in general, even for polynomial window sizes.
\end{lemma}

\smallskip\noindent\textbf{Fixed window: summary.} We summarize the complexity of the fixed window problem in Theorem~\ref{thm:multiDimFixed}.

\begin{theorem}
\label{thm:multiDimFixed}
In two-player multi-dimension games, the fixed arbitrary window mean-payoff problem is EXPTIME-complete, and the fixed polynomial window mean-payoff problem is PSPACE-hard. For both players, exponential memory is sufficient and is required in general.
\end{theorem}

\smallskip\noindent\textbf{$\Bounded$ window.}
Unlike the one-dimension case, in which it is easier to decide the $\bounded$ problem than the fixed arbitrary one (i.e., the problem becomes easier when the fixed window size is sufficiently large), we prove that the complexity of the $\bounded$ window problem in multi-weighted games is at least non-primitive recursive.\footnote{That is, there exists no primitive recursive function that computes the answer to the $\bounded$ window problem. A well-known example of a decidable but non-primitive recursive function is the Ackermann function~\cite{ackermann1928}.} Hence, there is no hope for efficient algorithms on the complete class of two-player multi-weighted games.
This result is obtained by reduction from the problem of deciding the existence of an infinite execution in a \textit{marked reset net}, also known as the \textit{termination problem}. A marked reset net~\cite{dufourd98} is a Petri net with \textit{reset arcs} together with an initial marking of its places. Reset arcs are special arcs that reset a place (i.e., empty it of all its tokens). The termination problem for reset nets is decidable but non-primitive recursive hard (as follows from~\cite{schnoebelen02}, also discussed in~\cite{lazic08}).
 
\begin{figure}[htb]
  \centering   
  \scalebox{0.85}{\begin{tikzpicture}[->,>=stealth',shorten >=1pt,auto,node
    distance=2.5cm,bend angle=45,scale=0.3, font=\small]
    \tikzstyle{p1}=[draw,circle,text centered,minimum size=10mm]
    \tikzstyle{p2}=[draw,rectangle,text centered,minimum size=10mm]
    \tikzstyle{p3}=[draw,dashed,rectangle,text centered,minimum height=30mm, minimum width=40mm]
    \node[p3] (gadget1) at (0,0) {};
    \node[p3] (gadget2) at (30,0) {};
    \node[p1] (zt1) at (4,-2) {${\sf test}_{t}$};
    \node[p2] (trans1) at (0,2.8) {${\sf fire}$}; 
    \node[p1] (zt2) at (34,-2) {${\sf test}_{t}$};
    \node[p2] (trans2) at (30,2.8) {${\sf fire}$};
    \node[p1] (close1) at (11,-2) {${\sf close}_p$};
    \node[p2] (delay1) at (19,-2) {${\sf delay}$};
    \node[p1] (close2) at (19,10) {${\sf close}_p$};
    \node[p2] (delay2) at (11,10) {${\sf delay}$};
    \path
    (close1) edge [loop below, out=240, in=300,looseness=3, distance=4cm] node [below] {$(\resetNetOneZeroVector, 1, 1, -1)$} (close1)
    (delay1) edge [loop below, out=240, in=300,looseness=3, distance=4cm] node [below] {$(\resetNetZeroVector, 0, 1, 1)$} (delay1)
    (close2) edge [loop above, out=60, in=120,looseness=3, distance=4cm] node [above] {$(\resetNetOneZeroVector, 1, -1, 1)$} (close2)
    (delay2) edge [loop above, out=60, in=120,looseness=3, distance=4cm] node [above] {$(\resetNetZeroVector, 0, 1, 1)$} (delay2)
    ;
	\draw[->,>=latex] (delay1) to[out=90,in=90,looseness=1.5] node[left,yshift=-8mm,xshift=-11mm] {$(-\resetNetInitMarking - \resetNetOneVector, 0, 0, 0)$} node[left,yshift=-5mm,xshift=-18mm] {\textit{restart}} (trans2);
	\draw[->,>=latex] (delay2) to[out=180,in=90] node[above,yshift=2mm,xshift=-5mm] {$(-\resetNetInitMarking - \resetNetOneVector, 0, 0, 0)$} node[above,yshift=5.5mm,xshift=-5.5mm] {\textit{restart}} (trans1);
	\draw[->,>=latex] (zt2) .. controls ([xshift=8cm,yshift=4cm] zt2) and ([xshift=20cm] delay2) .. node[above,yshift=6mm,xshift=-4mm] {\textit{place} $p$} (close2);
	\draw[->,>=latex] (close1) to (delay1);
	\draw[->,>=latex] (close2) to (delay2);
	\draw[->,>=latex] (zt1) to node[below] {\textit{place} $p$} (close1);
      \end{tikzpicture}}
      \caption{Careful alternation between gadgets is needed in order for $\playerOne$ to win.}
\label{fig:multiDimFiniteAlternatingGadgets}
  \end{figure}
  
Given a reset net $\resetNet$ with an initial marking $\resetNetInitMarking \in \nat^{\resetNetPlacesSize}$ (where $\resetNetPlaces$ is the set of places of the net), we build a two-player multi-weighted game $\game$ with $\dimension = \resetNetPlacesSize + 3$ dimensions such that $\playerOne$ wins the $\bounded$ window objective for threshold $\zeroVector$ if and only if $\resetNet$ does not have an infinite execution from $\resetNetInitMarking$.

A high level description of our reduction is as follows. The structure of the game (Fig.~\ref{fig:multiDimFiniteAlternatingGadgets}) is based on the alternance between two gadgets simulating the net (Fig.~\ref{fig:multiDimFiniteGadgetZoom}). Edges are labeled by $\dimension$-dimension weight vectors such that the first $\resetNetPlacesSize$ dimensions are used to encode the number of tokens in each place. In each gadget, $\playerTwo$ chooses transitions to simulate an execution of the net. During a faithful simulation, there is always a running open window in all the first $\resetNetPlacesSize$ dimensions: if place $p$ contains $n$ tokens then the negative sum from the start of the simulation is $-(n+1)$. This is achieved as follows: if a transition $t$ consumes $\resetNetInputT{t}(p)$ tokens from $p$, then this value is added on the corresponding dimension, and if $t$ produces $\resetNetOutputT{t}(p)$ tokens in $p$, then $\resetNetOutputT{t}(p)$ is removed from the corresponding dimension. When a place $p$ is reset, a gadget ensures that dimension $p$ reaches value $-1$ (the coding of zero tokens). This is thanks to the monotonicity property of reset nets: if $\playerOne$ does not simulate a full reset, then the situation gets easier for $\playerTwo$ as it leaves him more tokens available. If all executions terminate, $\playerTwo$ has to choose an unfireable transition at some point, consuming unavailable tokens from some place $p \in \resetNetPlaces$. If so, the window in dimension $p$ closes. After each transition choice of $\playerTwo$, $\playerOne$ can either continue the simulation or branch out of the gadget to close all windows, except in some dimension $p$ of his choice. Then $\playerTwo$ can arbitrarily extend any still open window in the first $(\resetNetPlacesSize + 1)$ dimensions and restart the game afterwards. Dimension $(\resetNetPlacesSize + 1)$ prevents $\playerOne$ from staying forever in a gadget. If an infinite execution exists, $\playerTwo$ simulates it and never has to choose an unfireable transition. Hence, when $\playerOne$ branches out, the window in some dimension $p$ stays open. The last two dimensions force him to alternate between gadgets so that he cannot take profit of the prefix-independence to win after a faithful simulation. So, $\playerTwo$ can delay the closing of the open window for longer and longer, thus winning the game.

\vspace{-3mm}

\begin{figure}[thb]
  \centering   
  \scalebox{0.85}{\begin{tikzpicture}[->,>=stealth',shorten >=1pt,auto,node
    distance=2.5cm,bend angle=45,scale=0.35, font=\small]
    \tikzstyle{p1}=[draw,circle,text centered,minimum size=10mm]
    \tikzstyle{p2}=[draw,rectangle,text centered,minimum size=10mm]
    \tikzstyle{p3}=[draw,dashed,rectangle,text centered,minimum height=60mm, minimum width=60mm]
    \tikzstyle{p4}=[]
    \node[p2] (trans) at (5,5) {${\sf fire}$};
    \node[p1] (test) at (5,-5) {${\sf test}_{t_{i}}$};
    \node[p1] (reset) at (-5,-5) {${\sf reset}_q$};
    \node[p1] (out) at (-5,5) {${\sf out}$};
    \node[p4] (trans1) at (2, 1) {};
    \node[p4] (trans2) at (8, 1) {};
    \node[p4] (a) at (3.5, 2) {$\ldots{}$};
    \node[p4] (b) at (6.5, 2) {$\ldots{}$};
    \coordinate[shift={(0mm,10mm)}] (init) at (trans.north);
    \coordinate[shift={(16mm,0mm)}] (close) at (test.east);
    \path
    (trans) edge node[left,xshift=-6mm,yshift=-1mm] {$t_{i}$} node[left,yshift=-4mm] {$(\resetNetInputT{t_{i}},-1,0,0)$} (test)
    (test) edge node[below] {$(\resetNetZeroVector, -1, 0, 0)$} (reset)
    (reset) edge node [above,rotate=90] {$(\resetNetZeroMinusOneVector, -1, 0, 0)$} (out)
    (reset) edge [loop below, out=240, in=300,looseness=3, distance=3cm] node [below] {$(\resetNetZeroOneVector, -1, 0, 0)$} (reset)
    (out) edge node[above] {$(-\resetNetOutputT{t_{i}},-1,0,0)$} (trans)
    (test) edge node[below,xshift=-2mm] {\textit{place} $p$} (close)
    (init) edge node[above,yshift=1mm,xshift=-12mm] {$(-\resetNetInitMarking - \resetNetOneVector, 0, 0, 0)$} node[above,yshift=5mm,xshift=-12mm] {\textit{restart}} (trans)
    ;
    \draw[->,>=latex] (trans) to[out=210, in=90] node[left] {$t_{1}$} (trans1);
    \draw[->,>=latex] (trans) to[out=330, in=90] node[right] {$t_{\resetNetTransSize}$} (trans2);
    \draw[-,dashed] (-9, 7.5) -- (10, 7.5) -- (10,-10) -- (-9,-10) -- (-9, 7.5);
      \end{tikzpicture}}
      \caption{Gadget simulating an execution of the reset net.}
\label{fig:multiDimFiniteGadgetZoom}
  \end{figure}
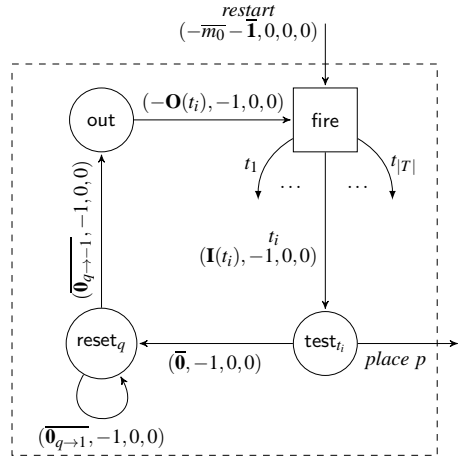
  
\vspace{-2mm}
  
\begin{theorem}
\label{thm:multiDimFinite}
In two-player multi-dimension games, the $\bounded$ window mean-payoff problem is non-primitive recursive hard.
\end{theorem}

\begin{proof}
We prove a reduction from the termination problem on reset nets to the $\bounded$ window problem on two-player multi-weighted games. The former is known to be non-primitive recursive hard~\cite{schnoebelen02,lazic08}.

Let $\resetNet = \left\langle \resetNetPlaces, \resetNetTrans, \resetNetInput, \resetNetOutput, \resetNetReset\right\rangle $ be a \textit{reset net} such that
\begin{itemize}
\item $\resetNetPlaces = \{ p_{1}, p_{2}, \ldots{}, p_{\resetNetPlacesSize}\}$ is the set of places;
\item $\resetNetTrans = \{ t_{1}, t_{2}, \ldots{}, t_{\resetNetTransSize}\}$ is the set of transitions;
\item $\resetNetInput\colon \resetNetTrans \rightarrow \nat^{\resetNetPlacesSize}$ is the input function, such that for each transition $t \in \resetNetTrans$, $\resetNetInputT{t}$ is a $\resetNetPlacesSize$-dimen\-sion vector such that for all dimension $p \in \{1, \ldots{}, \resetNetPlacesSize\}$, $\resetNetInputT{t}(p)$ specifies the number of tokens from place $p$ consumed by the transition~$t$;\footnote{For simplicity, we use $p$ to refer to a place $p \in \resetNetPlaces$ and to the number $i \in \{1, \ldots{}, \resetNetPlacesSize\}$ such that $p_{i} = p$, that is $p$ indistinctly refers to the place and the corresponding dimension in the weight vectors.}
\item $\resetNetOutput\colon \resetNetTrans \rightarrow \nat^{\resetNetPlacesSize}$ is the output function, such that for each transition $t \in \resetNetTrans$, $\resetNetOutputT{t}$ is a $\resetNetPlacesSize$-dimension vector such that for all dimension $p \in \{1, \ldots{}, \resetNetPlacesSize\}$, $\resetNetOutputT{t}(p)$ specifies the number of tokens produced in place $p$ by the transition~$t$;
\item $\resetNetReset\colon \resetNetTrans \rightarrow \resetNetPlaces$ is the reset function, such that for all transition $t \in \resetNetTrans$, $\resetNetReset(t)$ specifies the unique place (w.l.o.g.) which is reset by transition~$t$.
\end{itemize}
Given an initial marking of the places (i.e., an initial number of tokens in each place) $\resetNetInitMarking \in \nat^{\resetNetPlacesSize}$, the termination problem asks if there exists an infinite execution of the net, that is, if there exists an infinite sequence of transitions that can be fired from $\resetNetInitMarking$. A transition $t$ is \textit{fireable} from marking $\resetNetMarking \in \nat^{\resetNetPlacesSize}$ if for all place $p \in \resetNetPlaces$, $\resetNetInput(t)(p) \leq \resetNetMarking(p)$. An execution terminates if no transition can be fired because the necessary tokens are unavailable. We first note an important \textit{monotonicity} property of reset nets: for all reset net $\resetNet = \left\langle \resetNetPlaces, \resetNetTrans, \resetNetInput, \resetNetOutput, \resetNetReset\right\rangle$, for all markings $\resetNetMarking, \resetNetMarkingBis \in \nat^{\resetNetPlacesSize}$, if $\resetNetMarking \leq \resetNetMarkingBis$ and $\rho \in \resetNetTrans^{\omega}$ is an infinite sequence of transitions fireable from $\resetNetMarking$, then $\rho$ is also fireable from $\resetNetMarkingBis$. This property is used later on.

We claim that given a reset net $\resetNet$ and an initial marking $\resetNetInitMarking$, we can build in polynomial time a multi-weighted game $\game$ in which $\playerOne$ has a winning strategy for objective $\finiteWindowMPObjT{0}$ if and only if there exists no infinite execution of the net $\resetNet$ from $\resetNetInitMarking$.

We build the game $\gameFull$ with $\dimension = \resetNetPlacesSize + 3$ as represented in Fig.~\ref{fig:multiDimFiniteAlternatingGadgets} and Fig.~\ref{fig:multiDimFiniteGadgetZoom}. Unlabeled edges have value zero in all dimensions. For clarity, we define the following $\resetNetPlacesSize$-dimension integer vectors: $\resetNetOneVector = (1, \ldots{}, 1)$ is the unit vector, $\resetNetZeroVector = (0, \ldots{}, 0)$ is the zero vector, and, for $a, b \in \integ$, $p \in \resetNetPlaces$, the vector $\resetNetVector$ represents the vector $(a, \ldots{}, a, b, a, \ldots{}, a)$ which has value $b$ in dimension $p$ and $a$ in the other dimensions. The first $\resetNetPlacesSize$ dimensions of the game are used to encode the tokens present in each place, whereas the last three are used to compel $\playerOne$ to act fairly. Our construction will ensure that at all times along a valid execution of the net in a gadget, if a place $p \in \resetNetPlaces$ possess $n$ tokens, then the running sum of weights over the largest open window has value $(-n-1)$ in dimension $p$.

The states and edges of the game are built as follows.
\begin{itemize}
\item Inside a gadget, we have a state ${\sf fire}$ belonging to $\playerTwo$, with $\resetNetTransSize$ outgoing edges corresponding to the $\resetNetTransSize$ transitions of the net. Each transition $t$ is encoded as follows:
\begin{itemize}
\item an edge from ${\sf fire}$ to a state ${\sf test}_{t}$ belonging to $\playerOne$, of value $(\resetNetInputT{t},-1,0,0)$, such that the running sum is updated to accurately encode the consumption of tokens;
\item in state ${\sf test}_{t}$, $(\resetNetPlacesSize + 1)$ outgoing edges, giving $\playerOne$ the possibility to either branch out of the gadget, going to the state ${\sf close}_p$ corresponding to the dimension $p$ of his choice, or continuing via an edge of value $(\resetNetZeroVector, -1, 0, 0)$ to the ${\sf reset}_q$ state, a state of $\playerOne$ such that $q = \resetNetReset(t)$ is the unique place reset by transition $t$;
\item a self-loop of value $(\resetNetZeroOneVector, -1, 0, 0)$ on the ${\sf reset}_q$ state;
\item an edge from ${\sf reset}_q$ to ${\sf out}_{t}$ of value $(\resetNetZeroMinusOneVector, -1, 0, 0)$ which purpose is to ensure that in dimension $q$, there is a new open window of sum $-1$ after a full reset (i.e., it encodes that the number of tokens in place $q$ is zero);
\item an edge from ${\sf out}_{t}$ back to ${\sf fire}$ of value $(-\resetNetOutputT{t},-1,0,0)$, producing tokens according to the output of transition $t$. 
\end{itemize}
\item Branching from the left gadget leads to a state ${\sf close}_{p}^{\text{left}}$ of $\playerOne$ with a self-loop of weight $(\resetNetOneZeroVector, 1, 1, -1)$ and an outgoing edge to state ${\sf delay}^{\text{left}}$ of $\playerTwo$.
\item State ${\sf delay}^{\text{left}}$ possess a self-loop of value $(\resetNetZeroVector, 0, 1, 1)$ and an edge going to the right gadget with value $(-\resetNetInitMarking - \resetNetOneVector, 0, 0, 0)$.
\item The right gadget is constructed symmetrically, the only change being that the self-loop on states ${\sf close}_{p}^{\text{right}}$ of $\playerOne$ now has value $(\resetNetOneZeroVector, 1, -1, 1)$.
\end{itemize}
The game starts in the left gadget with an initial edge of value $(-\resetNetInitMarking - \resetNetOneVector, 0, 0, 0)$ corresponding to the initial marking of the net.

We claim that (i) if there exists no infinite execution $\rho \in \resetNetTrans^{\omega}$ of the net $\resetNet$, then $\playerOne$ has a winning strategy in $\game$ for the $\bounded$ window objective, and (ii) if there exists such an execution, then $\playerTwo$ has a winning strategy in $\game$. By determinacy, proving both claims will conclude our proof.

Case (i). Assume that there exists no infinite execution $\rho \in \resetNetTrans^{\omega}$ of the net. Then there exists a bound $b \in \nat$ on the length of any valid execution. Hence, $\playerTwo$ can only simulate the net faithfully for $b$ steps, so after at most $(b+1)$ steps, he needs to use an unfireable transition. That is, the next chosen transition requires more tokens than available in some place $p \in \resetNetPlaces$. We define a winning strategy $\strat_{1} \in \strats_{1}$ of $\playerOne$ in $\game$ as follows:
\begin{enumerate}
\item In a state ${\sf test}_{t}$, if the last transition $t$ was valid (i.e., all first $\resetNetPlacesSize$ dimensions have a negative running sum), go to the corresponding ${\sf reset}_{q}$ state. Otherwise, there exists a dimension $p$ in which the sum has become non-negative and all windows are closed: exit the gadget and go to the corresponding state ${\sf close}_{p}$.
\item In a state ${\sf reset}_{q}$, cycle until the sum in dimension $q$ takes value $0$, then go to state ${\sf out}_{t}$.
\item In a state ${\sf close}_{p}$, take the loop exactly $f(b)$ times before going to state ${\sf delay}$, where $f\colon \nat \rightarrow \nat$ is a well-chosen function that we define below (hence $f(b)$ is constant along the play).
\end{enumerate}
We claim that it is possible to define $f(b)$ sufficiently large to ensure that this strategy is winning. Let $\resetNetLargestM \in \nat$ be the largest number of tokens produced as output of any transition of the net, on any place. We consider the value of the negative sum in any of the first $(\resetNetPlacesSize +1)$ dimensions at the moment when $\playerOne$ decides to exit the gadget according to the strategy $\strat_{1}$. Notice that for any dimension $p \in \{1, \ldots{}, \resetNetPlacesSize\}$, this sum is bounded by $x = (- \resetNetInitMarking(p) - 1 - b\cdot \resetNetLargestM)$. Hence, the number of loops taken on any visit of state ${\sf reset}_{q}$ is bounded by $x$. The sum in dimension $(\resetNetPlacesSize +1)$ is thus bounded by $(b\cdot (4 + x) + 1)$, which we define as $f(b)$. The last two dimensions are not modified inside a gadget. Now clearly, looping in state ${\sf close}_{p}$ for $f(b)$ steps is sufficient to close all windows in all dimensions corresponding to places (recall that dimension $p$ is closed by $\playerTwo$ cheating on place $p$), as well as in dimension $(\resetNetPlacesSize +1)$. However, this loop opens a window in one of the last two dimensions (the last for the left gadget, and the second to last for the right gadget). As the ${\sf delay}$ state of $\playerTwo$ has a positive effect in those dimensions, if $\playerTwo$ decides to delay the play for $f(b)$ steps, all windows will be closed. If he does not delay, the play will proceed to the next gadget, in which $\playerTwo$ is also forced to cheat before $(b+1)$ transitions. Hence after looping for $f(b)$ steps in the corresponding ${\sf close}_{p}$ state, the open window will close (and another will open in the other dimension which will in turn be closed after the next gadget). By keeping this behavior, $\playerOne$ can thus enforce that any open window along the play will close in at most $(4 \cdot f(b) + 4)$ steps. Thus the outcome is winning for the $\bounded$ window objective.

Case (ii). Assume that there exists an infinite execution $\rho \in \resetNetTrans^{\omega}$ of the net. We define a winning strategy $\strat_{2} \in \strats_{2}$ of $\playerTwo$ as follows. The strategy is played in rounds, with the initial round being round $1$.
\begin{enumerate}
\item Every time a gadget is entered, start playing in state ${\sf fire}$ according to the infinite execution $\rho$, that is, choose transitions in order to obtain the same trace.
\item When a state ${\sf delay}$ is visited during round $n$, take the self-loop $n$ times then continue to state ${\sf fire}$ and start round $n+1$.
\end{enumerate}
Notice that this strategy requires infinite memory. We claim that any consistent outcome of the game is winning for $\playerTwo$, that is, it does not belong to $\finiteWindowMPObjT{0}$. First, $\playerOne$ cannot stay forever in a gadget, thanks to dimension $(\resetNetPlacesSize +1)$: he has to branch at some point otherwise the play is lost. Second, if in state ${\sf reset}_{q}$, $\playerOne$ decides to cycle for less than necessary for a full reset, the situation gets better for $\playerTwo$ by the monotonicity property of the reset net (as $\playerTwo$ gets to continue with more tokens than expected). Notice that $\playerOne$ cannot accumulate positive values in the sum, as the next edge will restart a new window and all accumulation will be forgotten with regard to the objective. Third, if $\playerOne$ branches and exits the gadget to go to some state ${\sf close}_{p}$, then all dimensions corresponding to places, including dimension $p$, have a running open window (dimension $p$ has a strictly negative value since $\playerTwo$ does not cheat). Hence, no matter how long $\playerOne$ chooses the self-loop, the window in dimension $p$ will stay open (and $\playerOne$ cannot stay here forever because of the last two dimensions). Fourth, when the play reaches a state ${\sf delay}$ with an open window in dimension $p \in \{1, \ldots{}, \resetNetPlacesSize\}$, the strategy $\strat_{2}$ prescribes that $\playerTwo$ will loop for longer and longer periods of time, thus enforcing open windows of constantly growing length. As a consequence, any consistent outcome is such that the $\bounded$ window objective is not satisfied, which proves our point and further concludes our proof.
\end{proof}

\begin{remark}
Theorem~\ref{thm:multiDimFinite} establishes that the $\bounded$ window mean-payoff problem is non-primitive recursive hard, and the decidability of the problem remains open. Note that Theorem~\ref{thm:multiDimFinite} also implies that $\playerOne$ may require a window size of non-primitive recursive length to win a multi-dimension bounded window mean-payoff game (in contrast to the pseudo-polynomial bound of the one-dimension case given in Corollary~\ref{cor:oneDimCor}). The main motivation to study window objectives as a strengthening and approximation of the original objectives is to ensure the objectives in every sliding window of reasonable size. A prohibitively large window size of non-primitive recursive length suggests that the decidability of the bounded window problem is purely of theoretical interest, and the fixed window problem is the more relevant question.
\end{remark}

\subsection{\textbf{On direct objectives}}
\label{sec:directObj}
Through this paper, we have studied the prefix-independent versions of the objectives defined in Sec.~\ref{subsec:wmp_def}. In this section, we briefly argue that similar complexity results are obtained for the \textit{direct} variants (Table~\ref{table:complexityAndMemoryDirect}), by slight modifications of the presented proofs. Notice that memory requirements however change, as it is now sufficient to force one sufficiently long (for the fixed problem) or never closing (for the bounded problem) window to make an outcome losing.

\vspace{-4mm}
\renewcommand{\arraystretch}{1.2}
\begin{table}[htb]
  \centering   
\begin{footnotesize}
\begin{tabular}{|c||c|c|c||c|c|c|}
\cline{2-7} \multicolumn{1}{c|}{} &  \multicolumn{3}{c||}{~one-dimension~} & \multicolumn{3}{c|}{~$\dimension$-dimension~} \\ 
\cline{2-7} \multicolumn{1}{c|}{} &  ~~complexity~~ & ~~$\playerOne$ mem.~~ & ~~$\playerTwo$ mem.~~ & ~~complexity~~ & ~~$\playerOne$ mem.~~ & ~~$\playerTwo$ mem.~~\\ 
\hline ~~direct fixed~~ & \multirow{2}{*}{P-c.} & \multicolumn{2}{c||}{}  & PSPACE-h. & \multicolumn{2}{c|}{\multirow{4}{*}{exponential}}\\
~polynomial window~ & & \multicolumn{2}{c||}{mem. req.} & EXP-easy & \multicolumn{2}{c|}{} \\
\cline{1-2}\cline{5-5} ~~direct fixed~~ & \multirow{2}{*}{P($\vert\states\vert, \bits, \sizeMax$)} & \multicolumn{2}{c||}{$\leq$ linear($\vert\states\vert \cdot \sizeMax$)} & \multirow{2}{*}{EXP-c.} & \multicolumn{2}{c|}{} \\
~~arbitrary window~~ & & \multicolumn{2}{c||}{} & & \multicolumn{2}{c|}{}\\
\hline ~~direct $\bounded$~~ & \multirow{2}{*}{\NPinter} & \multirow{2}{*}{mem-less} & \multirow{2}{*}{\textbf{linear}} & \multirow{2}{*}{NPR-h.} & \multirow{2}{*}{-} & \multirow{2}{*}{-}\\
~~window problem~~ & & & & & &\\
\hline
\end{tabular}
\end{footnotesize}
\vspace*{2mm}
\caption{Complexities and memory requirements for the direct objectives. Differences with the prefix-independent objectives are in bold.}
\label{table:complexityAndMemoryDirect}
\end{table}
\vspace{-6mm}

\smallskip\noindent\textbf{One-dimension direct fixed window problem}. The polynomial algorithm in the size of the game and the size of the window is given by Lemma~\ref{lem:directWinAlg}. For polynomial windows, we obtain P-hardness using the proof of Lemma~\ref{lem:oneDimPHardness} and window size $\sizeMax = 2\cdot\statesSize$, as if $\playerOne$ can win the reachability game, he has a strategy to do it in at most $\statesSize$ steps. Lemma~\ref{lem:memoryOneDimFixed} extends to direct objectives, and provides linear upper bounds on memory with the same arguments. In particular, the provided examples of games require memory for both players when the direct fixed window objective is considered.

\smallskip\noindent\textbf{One-dimension direct $\bounded$ window problem}. We obtain an $\NPinter$ algorithm for the direct $\bounded$ problem by simplifying $\finiteProblemAlg$ (Lemma~\ref{lem:oneDimFiniteNPinter}) as follows: $\finiteProblemAlg(\game) = \states \setminus \unbNegWindowAlg(\game)$. Indeed, as the objective is no longer prefix-independent, it is sufficient for $\playerTwo$ to force one window that never closes to make the play losing. Hence, the attractor of the set $\states\setminus L$ in algorithm $\finiteProblemAlg$ cannot be declared winning for $\playerOne$. While memoryless strategies still suffice for $\playerOne$ (applying the arguments of Lemma~\ref{lem:oneDimFiniteNPinter}), winning strategies for $\playerTwo$ do not need infinite memory anymore, but at most linear memory. Indeed, a winning strategy of $\playerTwo$ is the one described in the proof of Lemma~\ref{lem:oneDimFiniteNPinter}, but without taking rounds into account (i.e., the play stays forever in round one). To illustrate that memoryless strategies still do not suffice for $\playerTwo$, consider a variation of Fig.~\ref{fig:wmpEx2}, with the initial state being $s_{2}$. Clearly, $\playerTwo$ must first take the cycle to $s_{1}$ then loop forever on $s_{2}$ to ensure a never closing window. Corollary~\ref{cor:oneDimCor} extends in the direct case and gives the same bound on the window size. Finally, the reduction of mean-payoff games developed in Lemma~\ref{lem:oneDimMPReduction} carries over to the direct $\bounded$ window objective, as the game with shifted weights is such that the mean-payoff is strictly positive. In which case, the supremum total-payoff is infinite and Lemma~\ref{lem:relationBoundedClassical} applies, implying the result.

\smallskip\noindent\textbf{Multi-dimension direct fixed window problem}. The following results extend to the direct case.
\begin{itemize}%
\item \textit{EXPTIME algorithm}. Lemma~\ref{lem:multiDimReducCB} presents a reduction from fixed window games to exponentially larger co-B\"uchi games. It is easy to obtain a similar reduction from direct fixed window games by considering a safety objective for $\playerOne$ (i.e., reachability for the set of bad states for $\playerTwo$). This also implies an exponential-time algorithm.
\item \textit{EXPTIME-hardness of the arbitrary window problem for weights $\{-1, 0, 1\}$ and arbitrary dimensions}. The reduction of the membership problem for polynomial space alternating Turing machines immediately yields the result for the direct objective. Indeed, the strategies proposed in the proof stay winning for this objective. Note that actually the strategy of $\playerTwo$ may be simpler, as he may cycle forever on $s_{{\sf restart}}$ after branching to punish an unfaithful symbol disclosure by keeping a window indefinitely open.
\item \textit{EXPTIME-hardness of the arbitrary window problem for two dimensions and arbitrary weights}. The reduction from countdown games established in Lemma~\ref{lem:multiDimFixedHardCountdown} extends straightforwardly to direct objectives, and $\playerTwo$ can use a simpler winning strategy consisting in looping forever in its zero cycle.
\item \textit{PSPACE-hardness of the polynomial window problem}. The reduction of generalized reachablity games also holds without modification for the direct fixed polynomial window objective.
\item \textit{Exponential memory bounds}. Exponential upper bounds follow from the modified Lemma~\ref{lem:multiDimReducCB}, using safety games. Lower bounds witnessed by Lemma~\ref{lem:multiDimFixedMemory} are also verified in the presented game as well as from the reduction of generalized reachablity games.
\end{itemize}

\smallskip\noindent\textbf{Multi-dimension direct $\bounded$ window problem}. Non-primitive recursive hardness (Theorem~\ref{thm:multiDimFinite}) extends to the direct objective with a simpler construction. Indeed, it is sufficient to consider the game using only the first $(\resetNetPlacesSize + 1)$ dimensions, and consisting of only one gadget, with the branching out of the gadget now going to an absorbing state with a self-loop of weight $\resetNetOneZeroVector$ such that when $\playerOne$ decides to branch, all windows get closed eventually, except in the dimension $p$ of his choice, for which the window is only closed if $\playerTwo$ cheats and stays open forever otherwise.

\section{Discussion}

\smallskip\noindent\textbf{Conclusion.} The strong relation between mean-payoff and total-payoff in single dimension breaks in multi-weighted games as the total-payoff threshold problem becomes undecidable. 
We introduced the concept of window objectives, which provide conservative approximations with timing guarantees. We believe that window objectives are interesting on the standpoint of \textit{expressiveness}, as they permit to consider quantitative objectives in a time frame context. Furthermore, window objectives constitute an attractive alternative in terms of \textit{tractability}.

We provided algorithms and optimal complexity bounds for one-dimension games. We notably showed that the fixed window variant can be solved in \textit{polynomial time}, which is not known to be the case for the mean-payoff and total-payoff objectives~\cite{ZP96,jurdzinski98,gawlitza2009,BCDGR11}.

In multi-dimensions, fixed window games hold an interesting position. While the associated decision problem is easier to solve than the mean-payoff threshold problem in one-dimension (P instead of $\NPinter$), it becomes comparatively harder in multi-dimension (PSPACE-hard even for polynomial windows instead of coNP). However, it remains EXPTIME-complete for arbitrary windows, in contrast to the total-payoff which becomes undecidable. In terms of complexity, the problem stands in an interesting middle ground between mean-payoff and total-payoff objectives. For the specific case of polynomial windows, there remains a gap between our exponential-time algorithm and the PSPACE lower bound. Whether we can obtain PSPACE-membership or EXPTIME-hardness for the fixed polynomial window problem in multi-dimension games is an open question. 

We also established a prohibitive lower bound on the complexity of multi-dimension bounded window games: they are at least non-primitive recursive hard. It would still be of theoretical interest to know if those games are decidable or not. Techniques used for the undecidability proof of multi-dimension total-payoff games (Thm.~\ref{thm:undecidableTP}) cannot be extended easily to the bounded window setting. In particular, our reduction to two-counter machines requires to ``memorize'' sums of weights both negatively and positively. In the window context, such sums can only be memorized negatively (i.e., while windows stay open), as positive windows are closed and forgotten immediately (this corresponds to so-called resets in Sect.~\ref{subsec:wmp_multiDim}). 

\smallskip\noindent\textbf{Future work.} We mention two interesting questions to investigate. 
First, in the multi-dimension setting, our definitions of window objectives (Sect.~\ref{subsec:wmp_def}) are asynchronous: windows on different dimensions are not required to close simultaneously. \textit{Synchronous variants} may be interesting to study but some useful properties are lost in that setting, such as the inductive property on windows. Hence our techniques cannot be extended straightforwardly. Second, conjunction of window objectives with a parity objective would be interesting to consider. Indeed, a similar notion of time bounds on liveness properties was studied by Chatterjee et al. through the concept of \textit{finitary} winning~\cite{CH06,DBLP:journals/tocl/ChatterjeeHH09}. Combining a similar approach with our window objectives seems natural.

\bibliographystyle{plain}
\bibliography{altMP_bib}

\end{document}